\documentclass[11pt]{article}
\usepackage[utf8]{inputenc}
\usepackage{amsmath, amssymb, amsthm, amsfonts}
\usepackage{mathtools}
\usepackage{url}
\usepackage{appendix}
\usepackage{fullpage}
\usepackage{caption}
\usepackage{subcaption}
\usepackage{graphicx}
\usepackage{wrapfig}
\usepackage{floatrow}
\usepackage[table,xcdraw]{xcolor}
\usepackage{multirow}
\usepackage{makecell}
\usepackage[ruled,vlined, linesnumbered]{algorithm2e}
\definecolor{forestgreen}{rgb}{0.13, 0.55, 0.13}

\SetCommentSty{mycommfont}

\definecolor{ForestGreen}{rgb}{0.0333,0.4451,0.0333}
\definecolor{DarkRed}{rgb}{0.65,0,0}
\definecolor{Red}{rgb}{1,0,0}
\usepackage[linktocpage=true,
pagebackref=true,colorlinks,
linkcolor=DarkRed,citecolor=ForestGreen,
bookmarks,bookmarksopen,bookmarksnumbered]{hyperref}

\usepackage{cleveref}
\usepackage{thm-restate}

\usepackage[colorinlistoftodos,prependcaption,textsize=tiny]{todonotes}

\newtheorem{theorem}{Theorem}[section]
\newtheorem{lemma}[theorem]{Lemma}
\newtheorem{claim}[theorem]{Claim}
\newtheorem{remark}[theorem]{Remark}

\newtheorem{definition}[theorem]{Definition}
\newenvironment{LabeledProof}[1]{\noindent{\bf Proof of #1: }}{\qed}

\newcommand{\poly}{\text{poly}}

\newcommand{\OPT}{\text{OPT}}

\newcommand*\mc[1]{\mathcal{#1}}
\DeclarePairedDelimiter\abs{\lvert}{\rvert}
\DeclarePairedDelimiter\norm{\lVert}{\rVert}%
\makeatletter
\let\oldabs\abs
\def\abs{\@ifstar{\oldabs}{\oldabs*}}
\let\oldnorm\norm
\def\norm{\@ifstar{\oldnorm}{\oldnorm*}}
\makeatother

\DeclareMathOperator{\dtr}{dtr}
\DeclareMathOperator{\final}{final}

\newcommand{\polylog}{{\rm polylog}}
\newcommand{\pw}{{\rm pw}}

\newcommand{\etal}{{et al. \xspace}}

\newcommand{\initOneLiners}{%
	\setlength{\itemsep}{0pt}
	\setlength{\parsep }{0pt}
	\setlength{\topsep }{0pt}
}
\newenvironment{OneLiners}[1][\ensuremath{\bullet}]
{\begin{list}
		{#1}
		{\initOneLiners}}
	{\end{list}}

\usepackage{booktabs}
\usepackage{caption}
\usepackage{enumerate}
\usepackage{float}
\usepackage{subcaption}
\usepackage[normalem]{ulem}
\usepackage{xspace}
\usepackage[numbers,sort&compress]{natbib}

\newcounter{note}[section]

\newcommand{\mcC}{\mathcal{C}}

\newcommand{\calC}{\mathcal{C}}

\allowdisplaybreaks

\newif\ifcomments
\commentstrue
\ifcomments

\newcommand{\atodo}[1]{\todo[linecolor=red,backgroundcolor=green!25,bordercolor=red]{\textbf{AF:~}#1}}

\newcommand{\dtodo}[1]{\todo[linecolor=red,backgroundcolor=yellow!25,bordercolor=red]{\textbf{DH:~}#1}}
\newcommand{\dtodoin}[1]{\todo[linecolor=red,backgroundcolor=yellow!25,bordercolor=red,inline]{\textbf{DH:~}#1}}

\newcommand{\rrtodo}[1]{\todo[linecolor=red,backgroundcolor=yellow!25,bordercolor=red]{\textbf{RR:~}#1}}

\newcommand{\tcr}[1]{{\color{red} #1}}
\newcommand{\rnote}[1]{\tcr{Ravi: #1}}

\newcommand{\ddim}{d}
\newcommand{\sddim}{d}

\else
\newcommand{\tnote}[1]{}
\newcommand{\yaowei}[1]{}

\newcommand{\dtodo}[1]{}
\newcommand{\dtodoin}[1]{}
\newcommand{\rnote}[1]{}
\newcommand{\rrtodo}[1]{}
\fi 

\title{
One Tree to Rule Them All:\\ Poly-Logarithmic Universal Steiner Tree}
\author{
\begin{tabular}[t]{c@{\extracolsep{1.3em}}cc} 
        Costas Busch  & \qquad Da Qi Chen & Arnold Filtser \\
        \small Augusta University  & \small \qquad University of Virginia & \small  Bar-Ilan University \\
        \\
        Daniel Hathcock  & \qquad D Ellis Hershkowitz & Rajmohan Rajaraman \\
        \small Carnegie Mellon University  & \small \qquad ETH Zürich & \small Northeastern University \\
\end{tabular}
}

\date{}
\begin{document}

\maketitle

\begin{abstract}
A spanning tree $T$ of graph $G$ is a $\rho$-approximate \emph{universal
Steiner tree} (UST) for root vertex $r$ if, for any subset of vertices
$S$ containing $r$, the cost of the minimal subgraph of $T$ connecting
$S$ is within a $\rho$ factor of the minimum cost tree connecting $S$
in $G$. Busch et al.\ (FOCS 2012) showed that every graph admits
$2^{O(\sqrt{\log n})}$-approximate USTs by showing that USTs are
equivalent to strong sparse partition hierarchies (up to
poly-logs). Further, they posed poly-logarithmic USTs and strong
sparse partition hierarchies as open questions.

We settle these open questions by giving polynomial-time algorithms for computing both $O(\log ^ 7 n)$-approximate USTs and poly-logarithmic strong sparse partition hierarchies. For graphs with constant doubling dimension or constant pathwidth we improve this to $O(\log n)$-approximate USTs and $O(1)$ strong sparse partition hierarchies. Our doubling dimension result is tight up to second order terms. We reduce the existence of these objects to the previously studied cluster aggregation problem and what we call dangling nets. 
\end{abstract}

\pagenumbering{gobble}
\newpage

\setcounter{secnumdepth}{5}
\setcounter{tocdepth}{3} \tableofcontents
\newpage

\pagenumbering{arabic}
\setcounter{page}{1}

\section{Introduction}
Consider the problem of designing a network that allows a server to broadcast a message to a single set of clients. If sending a message over a link incurs some cost then designing the best broadcast network is classically modelled as the Steiner tree problem \cite{hwang1992steiner}. Here, we are given an edge-weighted graph $G = (V,E, w)$, terminals $S \subseteq V$ and our goal is a subgraph $H \subseteq G$ connecting $S$ of minimum weight $w(H) := \sum_{e \in H} w(e)$. We let $\OPT_S$ be the weight of an optimal solution.

However, Steiner tree fails to model the fact that a server generally broadcasts different messages to different subsets of clients over time. If constructing network links is slow and labor-intensive, we cannot simply construct new links each time a new broadcast must be performed. Rather, in such situations we must understand how to construct a {\em single network} in which the broadcast cost from a server is small for every subset of clients. Ideally, we would like our network to be a tree since trees have a simple routing structure. Our goals are similar if our aim is to perform repeated aggregation of data of different subsets of clients. 
Motivated by these settings, Jia \etal \cite{jia2005universal} introduced the idea of universal Steiner trees (USTs), defined below and illustrated in \Cref{fig:UST}.

\begin{definition}[$\rho$-Approximate Universal Steiner Tree]
Given an edge-weighted graph $G = (V, E, w)$ and root  $r \in V$, a $\rho$-approximate universal Steiner tree is a spanning tree $T \subseteq G$ such that for every $S \subseteq V$ containing $r$, we have
\begin{align*}
    w(T\{S\}) \leq \rho \cdot \OPT_S
\end{align*}
where $T\{S\} \subseteq T$ is the minimal subtree of $T$ connecting $S$, and $\OPT_S$ is the minimum weight Steiner tree connecting $S$ in $G$.
\end{definition}

\begin{wrapfigure}{r}{0.32\textwidth}
		\vspace{-15pt}		
\scalebox{0.65}{\renewcommand{\arraystretch}{1.60}
\begin{tabular}{|c|c|c|}
	\hline
	\textbf{\begin{tabular}[c]{@{}l@{}}Family\end{tabular}}        &  \textbf{Approximation}            & \textbf{Ref.}      \\ \hline
	\multicolumn{3}{|c|}{\textbf{Complete Graphs}}    \\ \hline
	\multirow{2}{*}{\begin{tabular}[c]{@{}l@{}}{\large General}\end{tabular}}                                                       & {\Large $O(\log^2 n)$  }                   & \cite{gupta2006oblivious} \\ \cline{2-3} 
	& {\Large $\Omega(\log n)$}                  & \cite{jia2005universal}   \\ \hline
	
	\multirow{2}{*}{{\large Planar}}  	                                                     & {\Large $O(\log n)$ }                      & \cite{BLT14}              \\ \cline{2-3} 
	& {\Large $\tilde{\Omega}(\log n)$}          & \cite{jia2005universal}   \\ \hline

	\multirow{2}{*}{Doubling}                                                                                                       & {\Large $\tilde{O}(\ddim^3)\cdot\log n$}           & \cite{filtser20}              \\ \cline{2-3} 
& {\Large $\tilde{\Omega}(\log n)$ }         & \cite{jia2005universal}   \\ \hline
	
	Pathwidth                                                                                                     & {\Large $O(\pw\cdot\log n)$ }              & \cite{filtser20}              \\ \hline
	
	\multicolumn{3}{|c|}{\textbf{General Graphs}}    \\ \hline
	\multirow{2}{*}{\begin{tabular}[c]{@{}l@{}}{\large General}\\\end{tabular}}        & {\Large $2^{O(\sqrt{\log n})}$}            & \cite{busch2012split}     \\ \cline{2-3} 
	& {\Large $O(\log^7n)$ }                     & Thm. \ref{thm:USTGen}           \\ \hline

	Planar                                                  & {\Large $O(\log ^{18} n)$ }                & \cite{busch2012split}     \\ \hline 		
	
	Doubling	 & {\Large $\tilde{O}(\ddim ^7 )\cdot \log n$} & Thm.\ \ref{thm:USTDD}                \\ \hline	
	
	Pathwidth    & {\Large $O(\pw^8\cdot\log n)$ }              & Thm.\ \ref{thm:USTPW}                 \\ \hline
\end{tabular}
}		
	\vspace{-15pt}
\end{wrapfigure}

Known UST results are given to the right. Surprisingly, it is known that every $n$-vertex graph admits a $2^{O(\sqrt{\log n})}$-approximate and poly-time-computable UST, as proven by \cite{busch2012split} more than a decade ago. On the other hand, the best known lower bound is $\rho \geq \Omega(\log n)$ \cite{jia2005universal}. 
In fact, even when $G$ is the complete graph whose distances are induced by an $\sqrt{n}\times\sqrt{n}$ grid, there is an
$\Omega(\log n / \log \log n)$ lower bound. Improved upper bounds are known for several special cases:
fixed minor-free (e.g.\ planar) graphs admit $O(\log ^{18} n)$-approximate USTs \cite{busch2012split}.
Complete graphs induced by a metric admit $O(\log^2 n)$-approximate USTs~\cite{gupta2006oblivious}. If the inducing metric has doubling dimension $\ddim$, then the complete graph admits $O(\ddim^3\cdot\log n)$-approximate USTs \cite{filtser20}. Furthermore, better bounds are known for complete graphs when the inducing metric is the shortest path metric of a restricted graph $H$: if $H$ is planar or has pathwidth $\pw$ then the complete graph admits $O(\log n)$- \cite{BLT14} and $O(\pw\cdot\log n)$- \cite{filtser20} approximate USTs, respectively. 

\begin{figure}[t]
	\centering
	\begin{subfigure}[b]{0.3\textwidth}
		\centering
		\includegraphics[width=\textwidth,trim=200mm 100mm 200mm 60mm, clip]{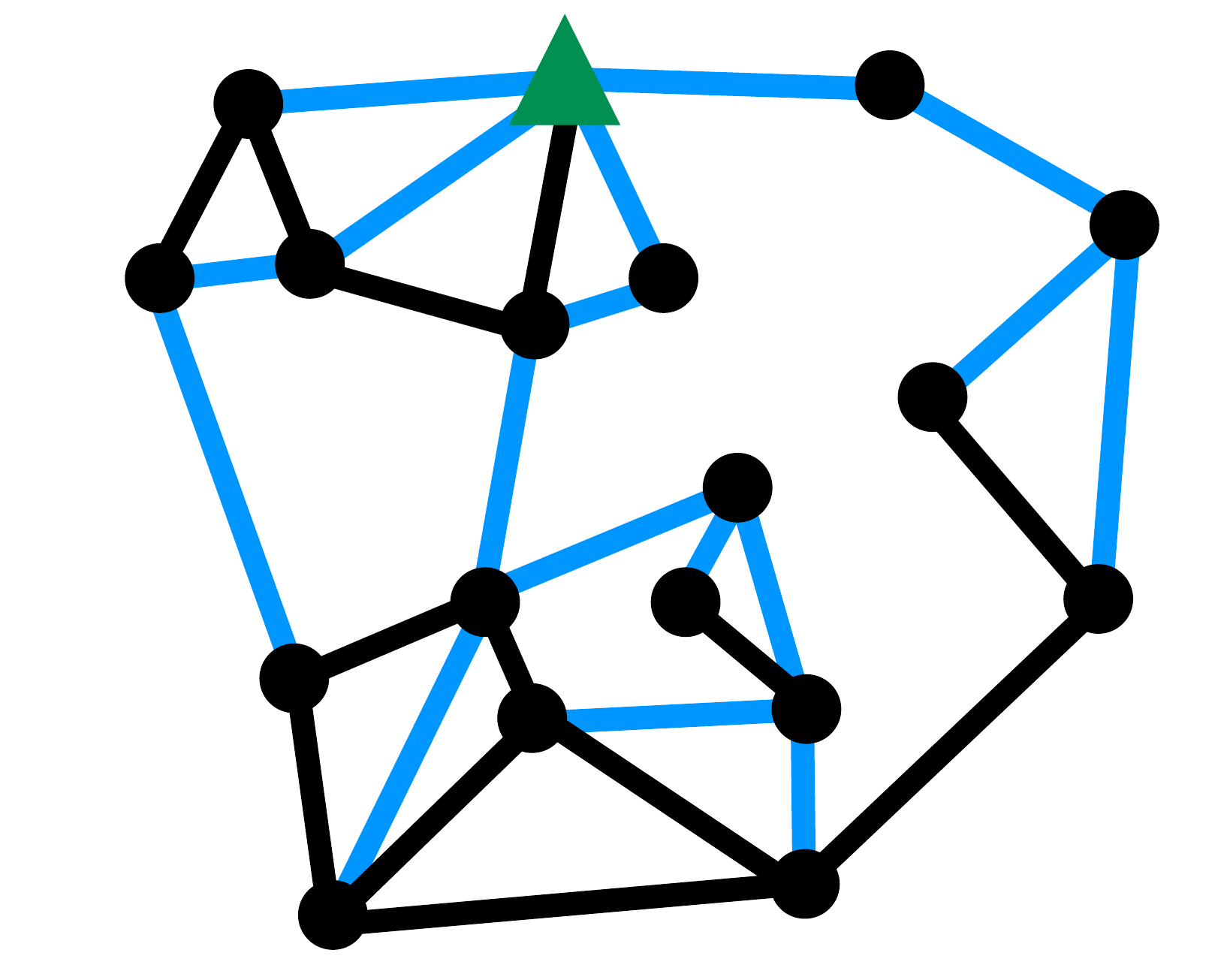}
		\caption{UST $T$.}\label{sfig:UST1}
	\end{subfigure}    \hfill
	\begin{subfigure}[b]{0.3\textwidth}
		\centering
		\includegraphics[width=\textwidth,trim=200mm 100mm 200mm 60mm, clip]{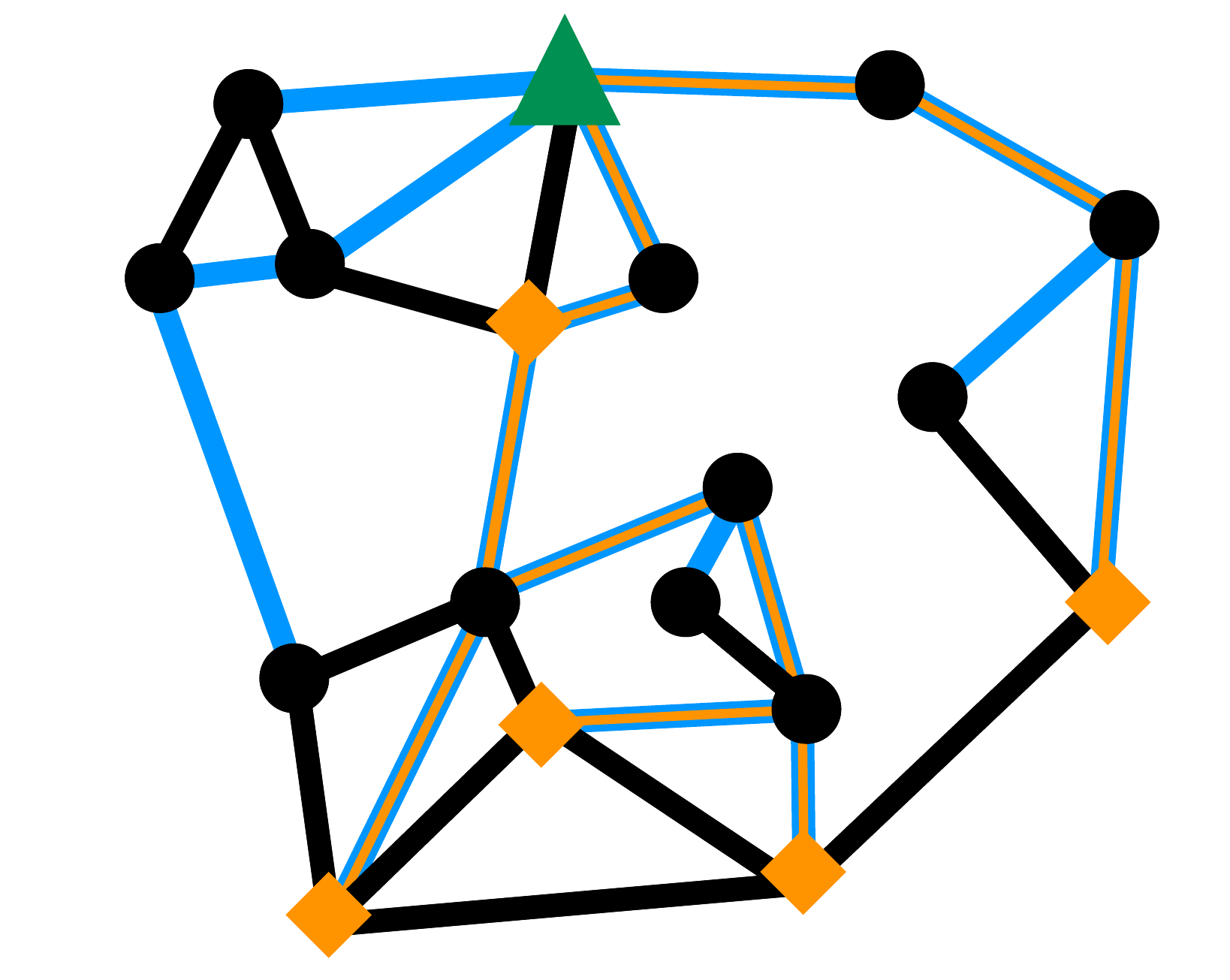}
		\caption{$T\{S\} \subseteq T$.}\label{sfig:UST2}
	\end{subfigure}    \hfill
	\begin{subfigure}[b]{0.3\textwidth}
		\centering
		\includegraphics[width=\textwidth,trim=200mm 100mm 200mm 60mm, clip]{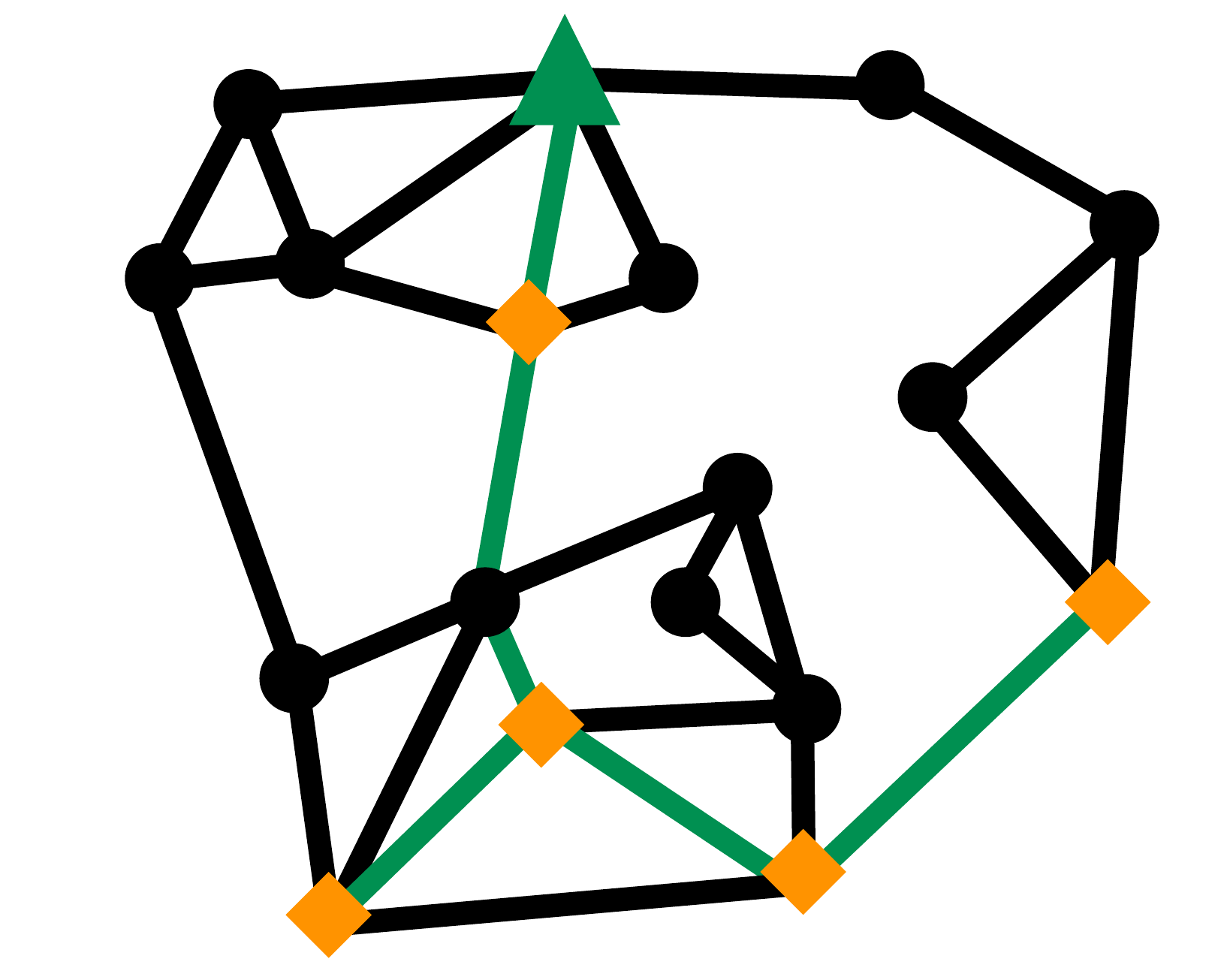}
		\caption{Optimal Steiner tree for $S$.}\label{sfig:UST3}
	\end{subfigure}
	\caption{\ref{sfig:UST1} is a UST (blue) in unit-weight $G$ with root $r$ (green triangle). \ref{sfig:UST2} is induced subtree $T\{S\}$ (orange) of weight $11$ for $S$ (orange diamonds). \ref{sfig:UST3} is weight $6$ optimal Steiner tree (green).}\label{fig:UST}
\end{figure}

Thus, for general graphs, there is a huge gap between the upper and lower bounds of  $2^{O(\sqrt{\log n})}$ and $\Omega(\log n)$. Closing this gap has been posed as an open question \cite{busch2012split, jia2005universal,filtser20}. 
\begin{quote}\textit{
Our main result is a poly-time $O(\log ^ 7 n)$-approximate UST, settling this open question.}
\end{quote}

\noindent Furthermore, if $G$ has constant doubling dimension or pathwidth, we provide an $O(\log n)$-approximate UST. The doubling dimension result is tight up to second order terms \cite{jia2005universal}.

We obtain our results by proving the existence of certain graph hierarchies---strong sparse partition hierarchies---
and leveraging a previously-established connection between these hierarchies and USTs. We prove the existence of these hierarchies, in turn, by reducing their existence to two objects: (1) low distortion solutions to the (previously studied) cluster aggregation problem, and (2) a certain kind of net which we call dangling nets that provide \emph{additive} sparsity guarantees. The existence of these nets can be inferred from an analysis in \cite{filtser20} of the random-shift techniques of \cite{MillerPX2013}. For cluster aggregation, we improve the best bounds in general graphs from $O(\log ^ 2 n)$ \cite{busch2012split} to $O(\log n)$ and prove $O(1)$ bounds for trees, constant doubling dimension and constant pathwidth graphs.  Our results are summarized in \Cref{tab:overview}. We spend the rest of this section describing them in greater detail.\footnote{We make use of standard graph notation and concepts throughout this work; see \Cref{sec:notation} for definitions.}

\def\arraystretch{1.2}
\begin{table}[h!]
	\begin{center}
		\begin{tabular}{|c|c|cc|cc|cc|} 
			\hline
			\textbf{Problem} & \textbf{Param.}& \textbf{General} & \textbf{Thm.} & \textbf{Doubling} &\textbf{Thm.}& \textbf{Pathwidth}& \textbf{Thm.}\\
			\hline
			UST & $\rho$ & $O(\log ^ 7 n)$& (\ref{thm:USTGen})& $\tilde{O}(\ddim^7 )\cdot \log n$ & (\ref{thm:USTDD})& $O(\pw^8 \cdot \log n)$& (\ref{thm:USTPW}) \\
			\hline
			Strong  Sparse & $\alpha$ &  $O(\log n)$ & & $O(\ddim)$ & & $O(\pw)$ & \\
			Hierarchy &$\tau$ & $O(\log n)$  & (\ref{thm:SSPGen})& $\tilde{O}(\ddim)$ &(\ref{thm:SSPDD})& $O(\pw^2)$ & (\ref{thm:SSPPW})\\
			&$\gamma$ & $O(\log^2 n)$  && $\tilde{O}(\ddim^3)$ && $O(\pw^2)$ &\\
			\hline
			Dangling  Nets & $\alpha$ & $O(\log n)$ && $O(\ddim)$ && $O(\pw)$ &\\
			\cite{filtser20}  & $\tau$ & $O(\log n)$ &&  $\tilde{O}(\ddim)$ && $O(\pw^2)$ &\\ \hline
			Cluster Agg.\ & $\beta$ & $O(\log n)$ & (\ref{thm:caGen})&  $\tilde{O}(\ddim^2)$ & (\ref{thm:caDD}) & $O(\pw)$ & (\ref{thm:caPW})\\ 
			\hline
		\end{tabular}
		\caption{A summary of our results. Our solutions for cluster aggregation on doubling dimension $\ddim$ only work for the instances of cluster aggregation we must solve to compute hierarchies of strong sparse partitions. $\tilde{O}$ notation hides $\poly(\log \ddim)$ factors. The results for dangling nets are proven implicitly in \cite{filtser20}. Our algorithms for general and doubling metrics are randomized and succeed with high probability ($1 - n^{-\Omega(1)}$), while our result for pathwidth graphs and trees are deterministic.
		}\label{tab:overview}
	\end{center}
\end{table}


\subsection{Poly-Logarithmic USTs}
As mentioned, our main result is a polynomial-time computable $O(\log ^ 7 n)$-approximate UST in general graphs. Not only is this an exponential improvement for general graphs\footnote{Following the conventions in theoretical computer science, we call $f=\polylog(g)$ an exponential improvement.}, it significantly improves upon the best bounds known for planar graphs (previously $O(\log^{18} n)$ \cite{busch2012split}). 
%

We also give improved UST bounds for graphs with doubling dimension $\ddim$ and graphs with pathwidth $\pw$: $\poly(\ddim)\cdot\log n$ and $\poly(\pw)\cdot\log n$ respectively.
Bounded doubling dimension graphs are a well-studied graph class that generalizes the ``bounded growth'' of low-dimensional Euclidean space to arbitrary graphs \cite{fraigniaud2006doubling,abraham2006routing,abraham2016forbidden,konjevod2008dynamic,filtser2016greedy,filtser2019relaxed}.
Bounded pathwidth graphs are a fundamental graph class that plays a key role in the celebrated graph minor theorem \cite{robertson1986graph}. 
As discussed above, it was previously known that $O(\log n)$-approximate USTs are possible if $G$ is a complete graph whose edge lengths are induced by either a constant doubling dimension metric or the shortest path metric of a constant pathwidth graph. 
Our results significantly strengthen this, showing $O(\log n)$-approximate USTs are possible for these two cases without the additional assumption that $G$ is the complete graph.

\subsection{Strong Sparse Partitions via Cluster Aggregation and Dangling Nets}\label{sec:contriSparsePartitions}

As mentioned, we achieve our UST algorithm by way of new results in graph hierarchies. 

We build on works over the past several decades on efficiently decomposing and extracting structure from graphs and metrics. Notable examples of this work are ball carvings, low-diameter decompositions (LDDs), network decompositions, padded decompositions and sparse neighborhood covers, all of which have numerous algorithmic applications, especially in parallel and distributed computing \cite{AP90,linial1993low,KPR93,forster2019dynamic,blelloch2011near,Fil19Padded,fox2010decompositions,elkin2022deterministic,awerbuch1996fast,abraham2008nearly,chang2021strong,bartal2004graph,kamma2017metric,rozhovn2020polylogarithmic,FL22LSO}. Generally speaking, these constructions separate a graph into clusters of nearby vertices while respecting the graph's distance structure.

The decompositions of our focus are \emph{strong sparse partitions}, first defined by \cite{jia2005universal} (in their weak diameter version) and studied in several later works \cite{busch2012split,czumaj2022streaming,filtser20}.


\begin{definition}[Strong $\Delta$-Diameter $(\alpha, \tau)$-Sparse Partitions]\label{dfn:sparsePart}
    Given edge-weighted graph $G = (V, E, w)$, a strong $\Delta$-diameter $(\alpha, \tau)$-sparse partition is a partition $\calC$ of $V$ such that:
	\begin{itemize}
		\item \textbf{Low (Strong) Diameter:} $\forall C\in\calC$, the induced graph $G[C]$ has diameter at most $\Delta$;
		\item \textbf{Ball Preservation:} $\forall v \in V$, the ball $B_G(v, \frac{\Delta}{\alpha})$ intersects at most $\tau$ clusters from $\calC$ .
	\end{itemize}
\end{definition}


\noindent Sparse partitions with \emph{weak diameter} and poly-logarithmic parameters can be constructed directly from classic sparse covers \cite{AP90,jia2005universal,filtser20} or ball carving techniques \cite{Bar96,czumaj2022streaming}.
However, to date, the only known techniques for poly-logarithmic sparse partitions with \emph{strong diameter} guarantees in general graphs are the $(O(\log n),O(\log n))$-sparse partitions of \cite{filtser20}, constructed using exponentially-shifted starting times.\footnote{Worse partitions with $2^{O(\sqrt{\log n})}$ parameters are possible by adapting the greedy approach of \cite{busch2012split}.} These start times were first used by \cite{MillerPX2013} to compute low diameter decompositions and spanners.\footnote{Even sparse partitions for pathwidth $\pw$ graphs (with parameters depending only on $\pw$) still use \cite{MillerPX2013}.} 

The simple class of tree graphs cannot do much better than the strong $(O(\log n),O(\log n))$-sparse partitions in general graphs: both $\alpha$ and $\tau$ have to be essentially $\Omega(\log n)$. As such, bounded pathwidth and doubling dimension graphs are of particular interest in this context. In particular, graphs with bounded pathwidth are exactly the graph family that excludes a fixed tree as a minor, circumventing this barrier with constant parameter strong sparse partitions.\footnote{Another interesting example comparing pathwidth with trees: trees require distortion $\Omega(\sqrt{\log\log n})$ for embeddings into $\ell_2$ \cite{Bou86,Mat99} while pathwidth $\pw$ graphs admit distortion $\sqrt{\pw}$ embeddings \cite{AFGN22}. }
On the other hand, trees that do not have good sparse partitions have doubling dimension $\Omega(\log n)$.


Little is known about graph decompositions in hierarchical settings; in particular, if our goal is a series of decompositions of increasing diameter where each decomposition coarsens the previous. One notable such hierarchy introduced by \cite{busch2012split} is a hierarchy of strong sparse partitions.\footnote{We assume that the minimal pairwise distance in $G$ is $1$. Otherwise, we can scale all distances accordingly.}

\begin{definition}[$\gamma$-Hierarchy of Strong $(\alpha, \tau)$-Sparse Partitions]\label{dfn:hierarchies}
    Given edge-weighted graph $G = (V, E, w)$, a $\gamma$-hierarchy of strong $(\alpha, \tau)$-sparse  partitions consists of vertex partitions $\{ \{v\} : v \in V\} = \mcC_0, \mcC_1, \ldots, \mcC_k = \{\{V\}\}$ such that: 
\begin{itemize}
		\item \textbf{Strong Partitions:} $\mcC_i$ is a strong $\gamma^i$-diameter $(\alpha, \tau)$-sparse partition for every $i$;
		\item \textbf{Coarsening:} $\mcC_{i+1}$ coarsens $\mcC_{i}$, i.e.\ for each $U \in \mcC_i$ there is a $U' \in \mcC_{i+1}$ such that $U \subseteq U'$.%
	\end{itemize}
\end{definition}
\noindent If we did not enforce the above coarsening property, we could trivially compute the above partitions with poly-logarithmic parameters by using the strong sparse partitions from \cite{filtser20} independently for each level of the hierarchy. However, the coarsening property renders computing such hierarchies highly non-trivial as it prevents such independent construction. Indeed, while hierarchies with poly-logarithmic parameterizations have been stated as an open question (see, e.g. \cite{filtser20}), the previous best bounds known for such hierarchies are $\gamma = \alpha = \tau = 2^{O(\sqrt{\log n})}$ \cite{busch2012split}. Thus, there is an exponential gap between the bounds known for ``one-level'' and hierarchical strong sparse partitions.

Nonetheless, previous work has demonstrated that these hierarchies can serve as the foundation of remarkably powerful algorithmic result such as USTs.
\begin{theorem}[\cite{busch2012split}]\label{thm:decompToUST} Given edge-weighted graph $G = (V,E,w)$ and a $\gamma$-hierarchy of strong $(\alpha, \tau)$-sparse partitions, one can compute an $O(\alpha^2\tau^2 \gamma \log n)$-approximate UST in polynomial time.
\end{theorem}

\noindent \cite{busch2012split} gave $2^{O(\sqrt{\log n})}$-approximate USTs by combining the above theorem with their hierarchies.

Our second major contribution is a reduction of such hierarchies to the previously-studied cluster aggregation problem \cite{busch2012split} and what we call dangling nets. Informally, cluster aggregation takes a vertex partition and a collection of portal vertices and coarsens it to a partition with a portal in each coarsened part. The goal is to guarantee that the portal in each vertex's coarsened cluster is nearly as close in its cluster as its originally-closest portal. Crucially for our purposes, the distortion of a solution is measured \emph{additively}. See \Cref{fig:CA} for an illustration of cluster aggregation.
\begin{figure}
    \centering
    \begin{subfigure}[b]{0.3\textwidth}
        \centering
        \includegraphics[width=\textwidth,trim=200mm 70mm 150mm 50mm, clip]{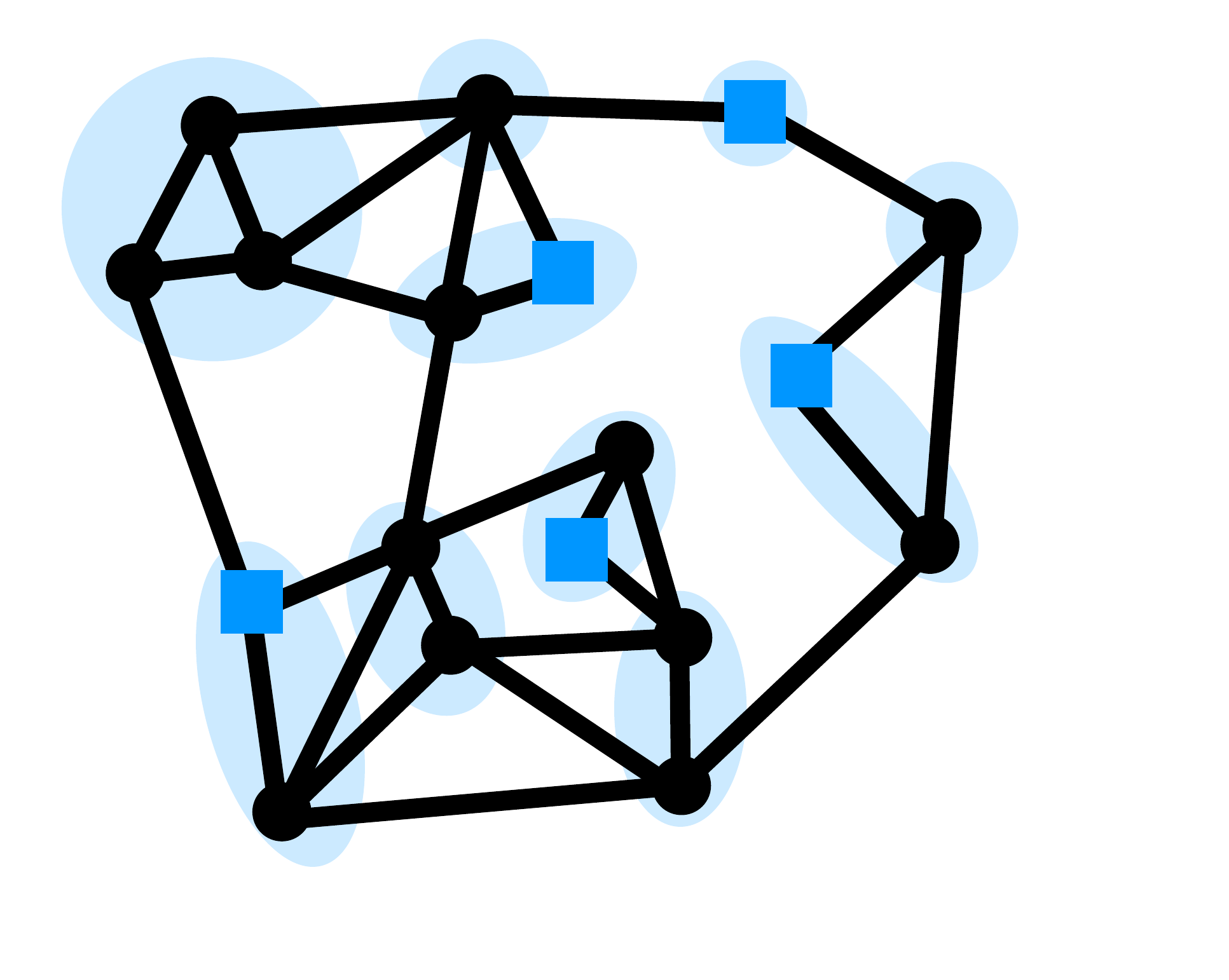}
        \caption{Cluster aggregation instance.}\label{sfig:CA1}
    \end{subfigure}    \hfill
    \begin{subfigure}[b]{0.3\textwidth}
        \centering
        \includegraphics[width=\textwidth,trim=200mm 70mm 150mm 50mm, clip]{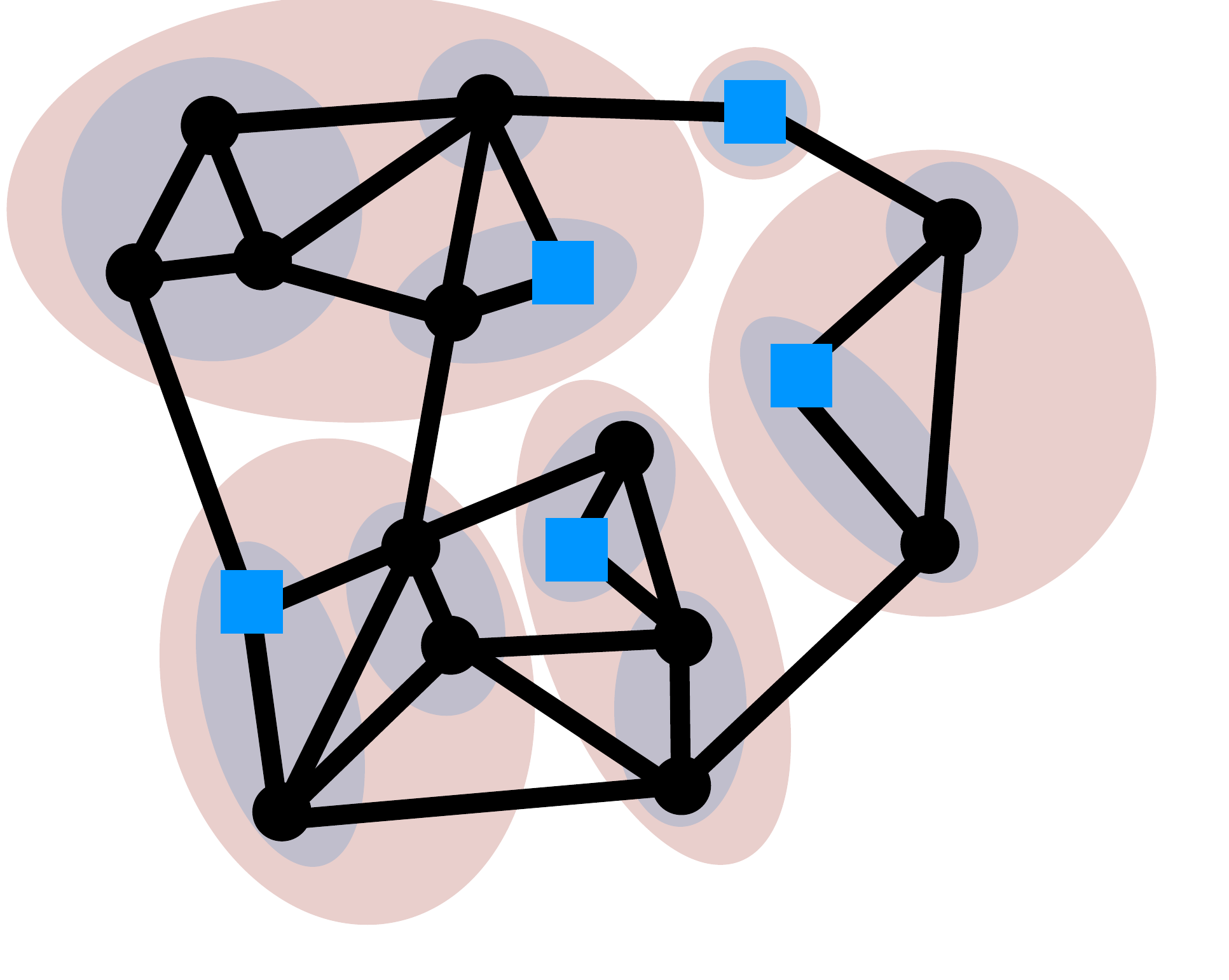}
        \caption{Cluster aggregation\ solution.}\label{sfig:CA2}
    \end{subfigure}    \hfill
    \begin{subfigure}[b]{0.3\textwidth}
        \centering
        \includegraphics[width=\textwidth,trim=200mm 70mm 150mm 50mm, clip]{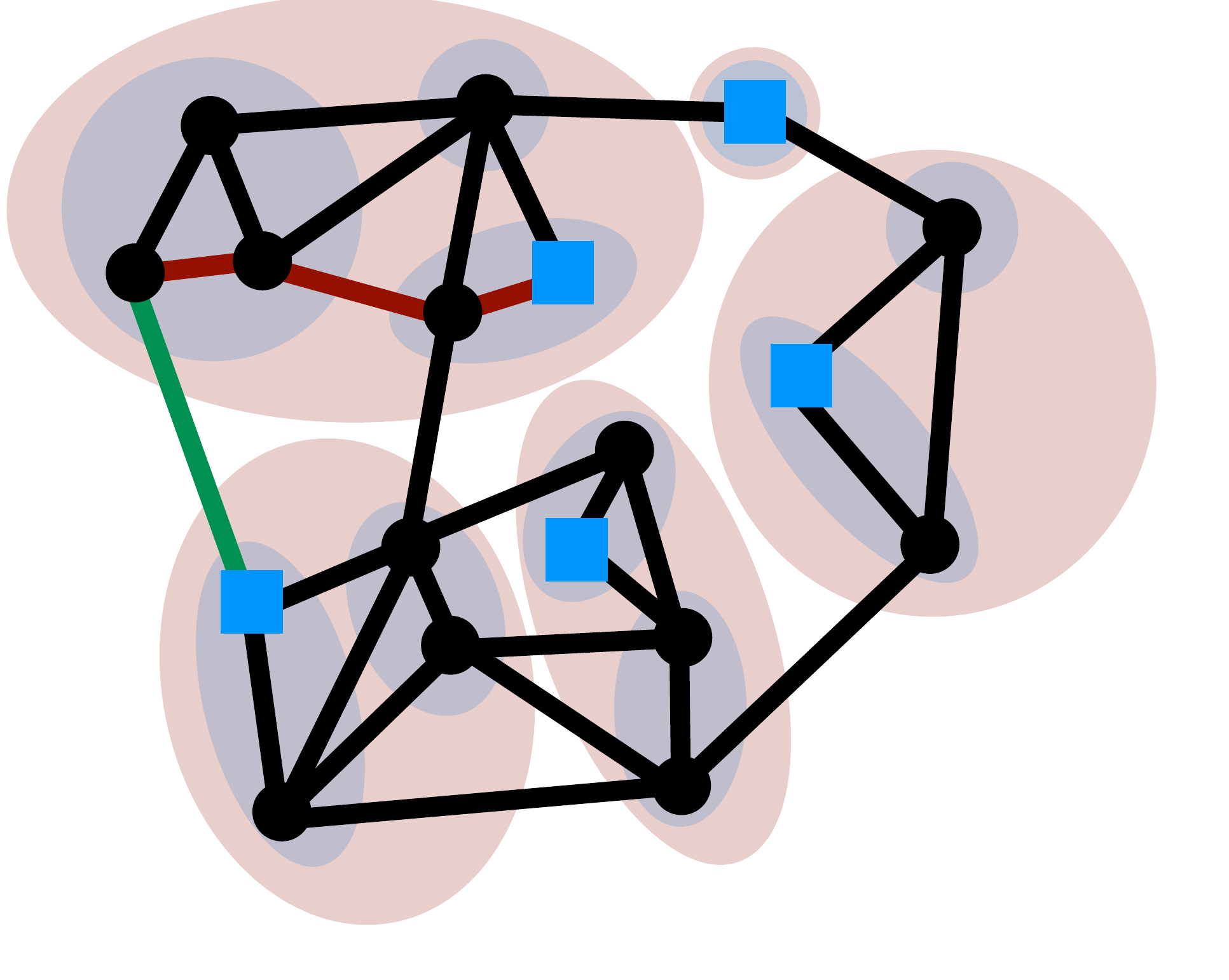}
        \caption{Solution distortion.}\label{sfig:CA3}
    \end{subfigure}
    \caption{A cluster aggregation instance with unit-weight edges. \ref{sfig:CA1} gives the instance; portals $P$ are blue squares and the input partition parts $\calC$ are blue ovals. \ref{sfig:CA2} gives solution where each red oval is the pre-image of some portal. \ref{sfig:CA3} illustrates why the solution is $2$-distortion with the path of a vertex to its nearest portal in green and to its nearest portal in its coarsened cluster in red.}\label{fig:CA}
\end{figure}

\begin{definition}[Cluster Aggregation]\label{dfn:clusterAgg}
    An instance of cluster aggregation consists of an edge-weighted graph $G = (V, E, w)$, a partition $\mcC$ of $V$ into clusters of strong diameter $\Delta$ and a set of portals $P \subseteq V$. A $\beta$-distortion solution is an assignment $f : \mcC \to P$ such that for every $v \in V$
    \begin{align*}
        d_{G[f^{-1}(f(v))]}(v, f(v)) \leq d_G(v, P)+\beta \cdot \Delta
    \end{align*}
    where $C_v \in \mcC$ is the cluster containing $v$ and we let $f(v) := f(C_v)$ and $f^{-1}(p) :=\{v : f(v)=p\}$.
\end{definition}
\noindent In other words, a $\beta$-distortion cluster aggregation solution $f$ requires that the distance from $v$ to $p=f(v)$ in the cluster induced by $p$, is at most $\beta\cdot \Delta$ larger than the distance from $v$ to it's closest portal in $G$. Observe that any solution $f$ on input cluster aggregation partition $\mcC$ naturally corresponds to a coarser partition $\mcC'$. Also, observe that, in general, we have that $\beta \geq 1$ by \Cref{fig:CATrivialLB}. 

Informally, a dangling net is a collection of net vertices we ``dangle'' off of a graph so that every vertex is close to a net vertex but no vertex has too many net vertices nearby. Crucially, the sense of ``nearby'' is also measured \emph{additively}. See \Cref{sfig:alg2} for an illustration.
\begin{definition}[$\Delta$-Covering $(\alpha,\tau)$-Sparse Dangling Net]\label{dfn:danglingNet}
    A dangling net for graph $G=(V,E,w)$ consists of vertices $N$ where $N \cap V = \emptyset$ and a matching $M$ with edge weights $w_M$ from $N$ to $V$. We let $G+N := (V \sqcup N, E \sqcup M, w \sqcup w_M)$ be the resulting graph. $N$ is $\Delta$-covering $(\alpha, \tau)$-sparse if
    \begin{itemize}
        \item \textbf{Covering:} $d_{G+N}(v, N) \leq \Delta$ for every $v \in V$;
        \item \textbf{Additive Sparsity:} $\left|\left\{t \in N : d_{G+N}(v,t)\le d_{G+N}(v,N)+\frac{\Delta}{\alpha}\right\} \right|\le\tau~$ for all $v \in V$.
    \end{itemize}
\end{definition}


\noindent While not explicitly stated in terms of dangling nets, the random shift analysis of \cite{filtser20} implicitly prove the existence of good parameter dangling nets: e.g.\ $\alpha=\tau=O(\log n)$ for general graphs. See Theorems \ref{thm:MPXbasedClusteringGeneral}, \ref{thm:MPXbasedClusteringPathwidth} and \ref{thm:MPXbasedClusteringDoubling} for details.

We state our reduction of strong sparse hierarchies to cluster aggregation and dangling nets.
\begin{restatable}{theorem}{HSSPmainp}\label{thm:MetaHierarchical}
	Fix edge-weighted graph $G$ and $\alpha, \beta, \tau \geq 0$. If for every $\Delta > 0$:
        \begin{itemize}
            \item \textbf{Dangling Net:} there is a dangling net $N$ that is $\Delta$-covering $(\alpha, \tau)$-sparse and;
            \item \textbf{Cluster Aggregation:} $G+N$ cluster aggregation on portals $N$ is always $\beta$-distortion solvable;
        \end{itemize}
        then, $G$ has a $2\beta\cdot\left(2\alpha+1\right)$-hierarchy of strong $(8\alpha+4,\tau)$-sparse partitions.  Furthermore, if each $N$ and cluster aggregation solution is poly-time computable then the hierarchy is poly-time computable.
\end{restatable}





\subsection{Improved Cluster Aggregation}

The connection we establish between cluster aggregation, strong sparse partition hierarchies and USTs---as well as the fact that \cite{busch2012split} posed improvements on their $O(\log ^ 2 n)$-distortion cluster aggregation solutions as an open question---motivates further study of cluster aggregation.

Our third major contribution is an improvement to cluster aggregation distortion in a variety of graph classes. Notably, we improve the $O(\log ^ 2 n)$-distortion solutions of \cite{busch2012split} to $O(\log n)$-distortion for general graphs and give improved bounds for trees, bounded pathwidth and bounded doubling dimension graphs. For bounded doubling dimension graphs we must make assumptions on the input (see \Cref{thm:caDD}). We know of no bounds prior to our work for cluster aggregation other than the previous $O(\log ^ 2 n)$ of \cite{busch2012split} for general graphs. We summarize our cluster aggregation results below ($\kappa \leq n$ is the number of clusters in the input partition).

\begin{table}[H]
	\begin{tabular}{|c|c|c|c|}
		\hline
		\textbf{Family}                                                           & \textbf{Distortion} & \textbf{Theorem}& \textbf{Location} \\ \hline
		\multirow{2}{*}{\begin{tabular}[c]{@{}l@{}}General\end{tabular}} &  $O(\log ^ 2 n)$                   &   -   & \cite{busch2012split}      \\ \cline{2-4} 
		&    $O(\log \kappa)$                 &   \Cref{thm:caGen}      &\Cref{sec:CAgeneral}      \\ \hline
		Trees                                                                      &  $4$                   & \Cref{thm:caTree}  &\Cref{sec:CAtree}            \\ \hline
		Doubling                                                                  &   $O(\ddim^2 \cdot \log \ddim)$                  &  \Cref{thm:caDD}    &\Cref{sec:CAdoubling}         \\ \hline
		Pathwidth                                                                 &    $8( \pw + 1 )$                 &  \Cref{thm:caPW}   &\Cref{sec:CApathwidth}          \\ \hline
	\end{tabular}
\end{table}





\noindent Combining our reduction (\Cref{thm:MetaHierarchical}) with the above cluster aggregation algorithms and dangling nets (Theorems \ref{thm:MPXbasedClusteringGeneral}, \ref{thm:MPXbasedClusteringPathwidth}, \ref{thm:MPXbasedClusteringDoubling}) gives our strong sparse partition hierarchies. Combining these hierarchies with \Cref{thm:decompToUST} gives our UST solutions. See \Cref{sec:putTogether} for proof details and again, see \Cref{tab:overview} for an overview of the resulting bounds.

As the notation we use is quite standard, we defer a description of it and our (mostly) standard preliminaries to \Cref{sec:notation}. Likewise, we defer additional related work to \Cref{sec:relatedWork}.


\begin{figure}
    \centering
        \centering
        \includegraphics[width=0.4\columnwidth,trim=80mm 130mm 80mm 130mm, clip]{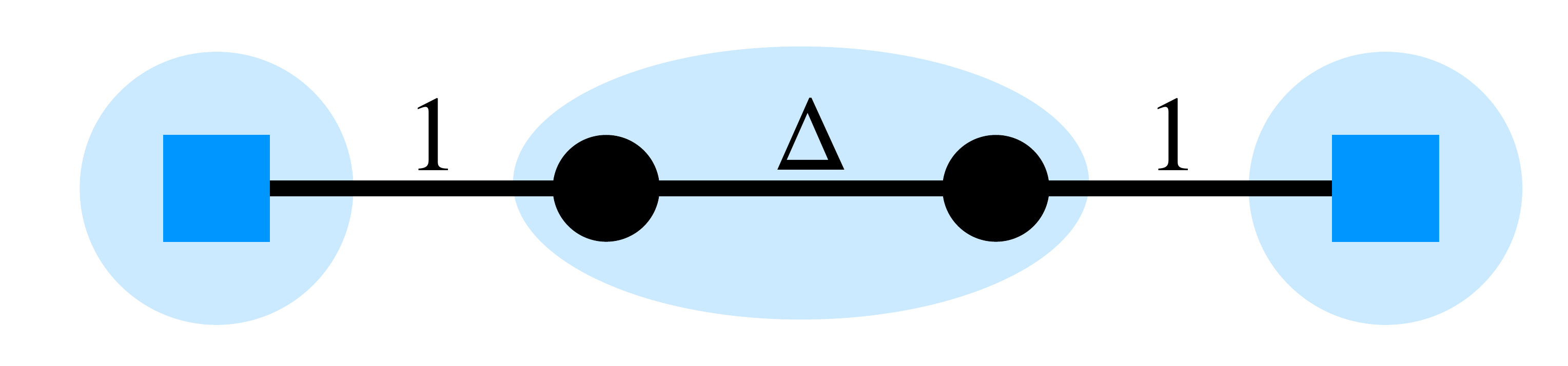}
    \caption{Why $\beta \geq 1$ for cluster aggregation. One vertex in the center cluster must traverse its $\Delta$-diameter cluster to get to a portal in any cluster aggregation solution.}\label{fig:CATrivialLB}
\end{figure}

\subsection{Additional Related Work}\label{sec:relatedWork}

We review additional related work not discussed earlier.

\subsubsection{(Online and Oblivious) Steiner Tree}
As it is an elementary NP-hard problem \cite{garey1979computers}, there has been extensive work on polynomial-time approximation algorithms for Steiner tree and related problems \cite{agrawal1991trees,byrka2013steiner,byrka2010improved,robins2005tighter,haeupler2021tree,filtser2022hop,garg2000polylogarithmic}. 

The subset of this work most closely related to our own is work on online and oblivious Steiner tree. In online Steiner tree the elements of $S\setminus\{r\}$ arrive one at a time and the algorithm must add a subset of edges to its solution so that it is feasible and cost-competitive with the optimal Steiner tree for the so-far arrived subset of $S\setminus\{r\}$. Notably, the greedy algorithm is a tight $O(\log n)$-approximation \cite{imase1991dynamic}, though improved bounds are known if elements of $S \setminus\{r\}$ leave rather than arrive \cite{gupta2014online,gu2013power}. See \cite{alon1992line,naor2011online,angelopoulos2007improved,xu2022learning} for further work. Even harder, in oblivious Steiner tree, for each possible vertex $v \in V \setminus \{r\}$, the algorithm must pre-commit to a path $P_v$ from $r$ to $v$. Then, a subset $S$ containing $r$ is revealed and the algorithm must play as its solution the union of its pre-commited-to paths for $S$, namely $\bigcup_{v \in S \setminus \{r\}}P_v$. The goal of the algorithm is for its played solution to be cost-competitive with the optimal Steiner tree for $S$ for every $S$. Notably, unlike USTs, the union of the paths played by the algorithm need not induce a tree. \cite{gupta2006oblivious} gave an $O(\log ^ 2 n)$-approximate polynomial-time algorithm for this problem and its more general version ``oblivious network design.''

Observe that any $\rho$-approximate UST immediately gives a $\rho$-approximate oblivious Steiner tree algorithm which, in turn, gives a $\rho$-approximate online Steiner tree algorithm. Thus, in this sense UST is at least as hard as both online and oblivious Steiner tree.

\subsubsection{Tree Embeddings and (Hierarchical) Graph Decompositions}

There has been extensive work on approximating arbitrary graphs by distributions over trees by way of so-called probabilistic tree embeddings \cite{bartal1998approximating,dhamdhere2006improved,abraham2012using,blelloch2016efficient,fakcharoenphol2003tight,abraham2020ramsey,filtser2021clan,haeupler2021tree,filtser2022hop}. Notably, any graph admits a distribution over trees that $O(\log n)$-approximate distances in expectation \cite{fakcharoenphol2003tight} and a distribution over subtrees that $O(\log n \log \log n)$-approximate distances in expectation \cite{abraham2012using}. 

USTs and probabilistic tree embeddings both attempt to flatten the weight structure of a graph to a tree. 
However, tree embeddings only aim to provide pairwise guarantees in expectation, while USTs provide guarantees for every possible subset of vertices deterministically. While one can always sample many tree embeddings to provide pairwise guarantees with high probability, the corresponding subgraph will not be a single tree, unlike a UST.

As mentioned in \Cref{sec:contriSparsePartitions}, decompositions of graphs into nearby vertices that  respect distance structure have been extensively studied. 
See, for example, \cite{czumaj2022streaming} for a recent application of sparse partitions in streaming algorithms. The graph decomposition most similar to sparse partitions are the scattering partitions of \cite{filtser20}. Informally, scattering partitions provide the same guarantees as sparse partitions but with respect to shortest paths rather than balls.

These sorts of decompositions (and, in particular, hierarchies of them) are intimately related to tree embeddings. For example, the tree embeddings of \cite{bartal1998approximating} can be viewed as a hierarchy of low-diameter decompositions. However, we note that, unlike strong sparse partition hierarchies, these hierarchies generally do not provide deterministic guarantees and, for example, \cite{bartal1998approximating}'s hierarchy only provides weak diameter guarantees. Somewhat similarly, \cite{abraham2020ramsey} produce a strong diameter padded decomposition hierarchy.

\subsubsection{Universal Problems}
In addition to Steiner tree, there are a number of problems whose universal versions have been studied. For example, the universal travelling salesman problem has been extensively studied \cite{schalekamp2008algorithms,gorodezky2010improved,HKL06,bhalgat2011optimal,jia2005universal,platzman1989spacefilling,bertsimas1989worst}. There are also works on universal set cover \cite{jia2005universal,grandoni2008set} and universal versions of clustering problems \cite{ganesh2023universal}.

\section{Overview of Challenges and Intuition}\label{sec:overview}

Before moving on to our formal results, we give a brief overview of our techniques.

\subsection{Reducing Hierarchies to Cluster Aggregation and Dangling Nets}

Similarly to previous work, we take a bottom-up approach to compute strong sparse partition hierarchies. We begin with the singleton partition $\mcC_0 = \{\{v\} : v \in V\}$ and then compute each $\mcC_{i+1}$ using $\mcC_i$. Recall that our goal is a strong $\gamma^{i+1}$-diameter partition $\mcC_{i+1}$ which coarsens $\mcC_i$ and which guarantees that any ball of radius $\gamma^{i+1}/\alpha$ intersects at most $\tau$ clusters of $\mcC_{i+1}$. 

\paragraph*{Previous Approach.}  A natural strategy for computing the cluster $C_j' \in \mcC_{i+1}$ containing cluster $C_j \in \mcC_i$ is to start with $C_j$ and expand it whenever it intersects a ``violated'' ball. Namely, if this cluster is incident to a diameter $\gamma^{i+1}/\alpha$ ball $B$ intersecting more than $\tau$ clusters, grow this cluster to contain all clusters intersecting $B$. The issue with this is that we may end up with a very long sequence of violated balls, each of which forces us to grow $C_j'$ further. See Figures \ref{sfig:seq1} and \ref{sfig:seq2}. 

The main observation of \cite{busch2012split} was that if the number of clusters each violated ball is incident to is at least $2^{O(\sqrt{\log n})}$, this sequence of violated balls can have length at most $2^{O(\sqrt{\log n})}$, which gives strong sparse partition hierarchies with $\alpha = \tau = \gamma = 2^{O(\sqrt{\log n})}$. Notably, the approach of \cite{busch2012split} is ``all or nothing'' in that if there is a violated ball incident to more than $\tau$ clusters of $\mcC_i$, then \emph{all} of these clusters are forced to be in the same cluster of $\mcC_{i+1}$. See \Cref{sfig:seq3}.

\begin{figure}[t]
	\centering
	\begin{subfigure}{.32\textwidth}
		\includegraphics[width=\textwidth]{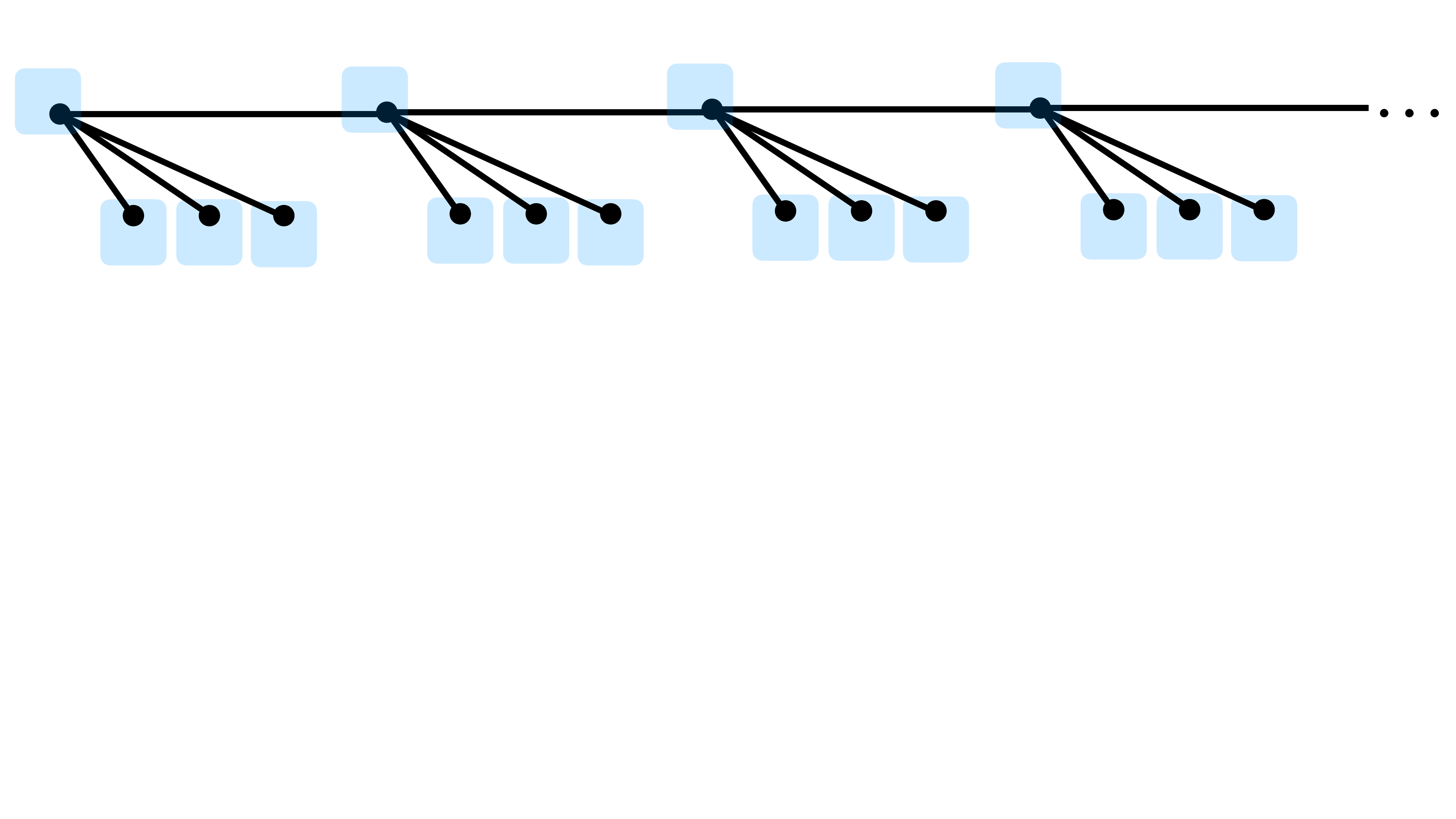}
    \vspace{-65pt}
		\caption{Partition $\mcC_{i}$.}\label{sfig:seq1}
	\end{subfigure}
	\begin{subfigure}{.32\textwidth}
		\includegraphics[width=\textwidth]{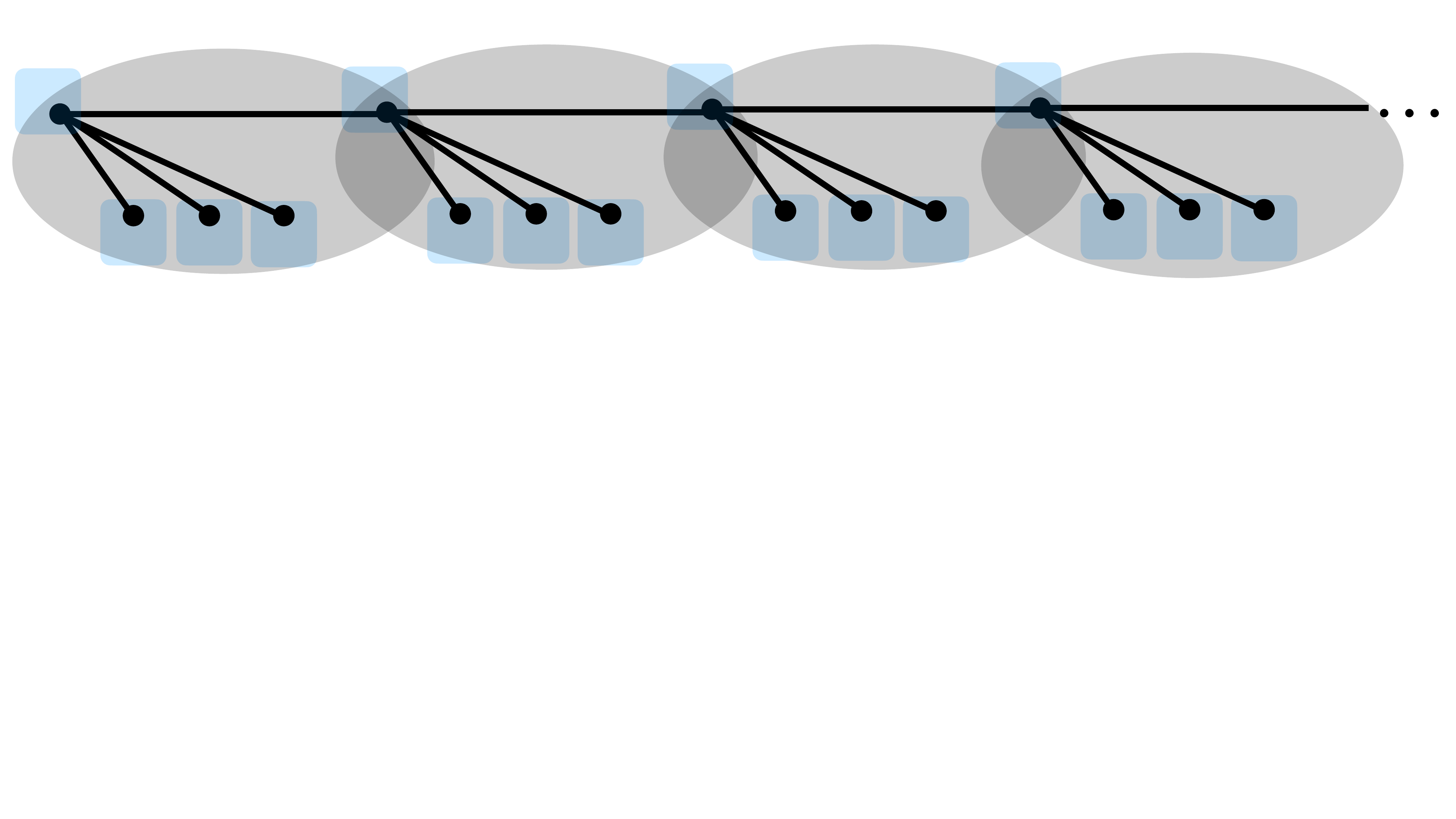}
    \vspace{-65pt}
		\caption{Sequence of violated balls.}\label{sfig:seq2}
	\end{subfigure}
	\begin{subfigure}{.32\textwidth}
		\includegraphics[width=\textwidth]{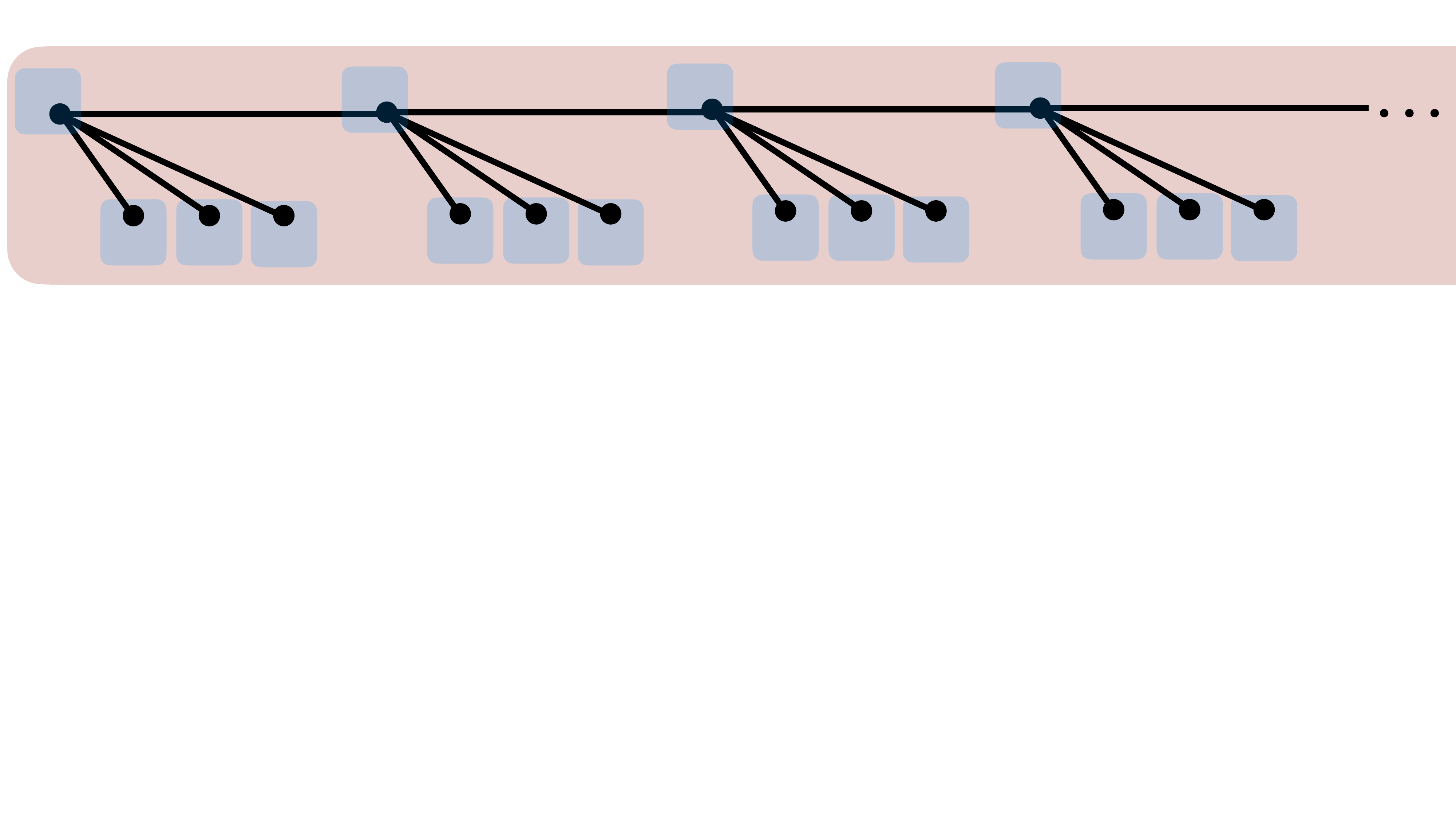}
    \vspace{-65pt}
		\caption{\cite{busch2012split} all or nothing $C_j'$.}\label{sfig:seq3}
	\end{subfigure}
	\caption{The challenge of a sequence of violated balls when constructing the cluster in $\mcC_{i+1}$ containing $C_j \in \mcC_i$. \ref{sfig:seq1} gives $\mcC_i$ as blue squares with $C_j$ upper-left. \ref{sfig:seq2} shows the ``violated balls'' of diameter $\gamma^{i+1}/\alpha$. 
 \ref{sfig:seq3} shows the solution computed by \cite{busch2012split} assuming that $\tau < 5$.}\label{fig:seq1To4}
\end{figure}

\paragraph*{Our Approach.} Our approach uses dangling nets to coordinate cluster aggregation in a way that coarsens $\mcC_i$ without being all or nothing. On one hand, a dangling net respects balls but not in a way that has anything to do with $\mcC_i$ or coarsening it. In particular, a dangling net $N$ corresponds to a natural sparse Voronoi partition (where each vertex goes to the closest net vertex in $N$) whose sparsity properties are robust to small (additive) changes. On the other hand, cluster aggregation provides a principled way of coarsening a partition $\mcC_i$ but does not necessarily respect balls. In particular, it coarsens a partition at the cost of small (additive) changes. We use dangling nets as portals for cluster aggregation to get the best of both techniques: cluster aggregation ensures that we coarsen with small additive costs while dangling nets ensure that these additive costs do not negatively impact sparsity. See \Cref{fig:alg} for an illustration of our approach (and its later analysis).

\begin{figure}
    \centering
    \begin{subfigure}[b]{0.3\textwidth}
        \centering
        \includegraphics[width=\textwidth,trim=0mm 0mm 0mm 0mm, clip]{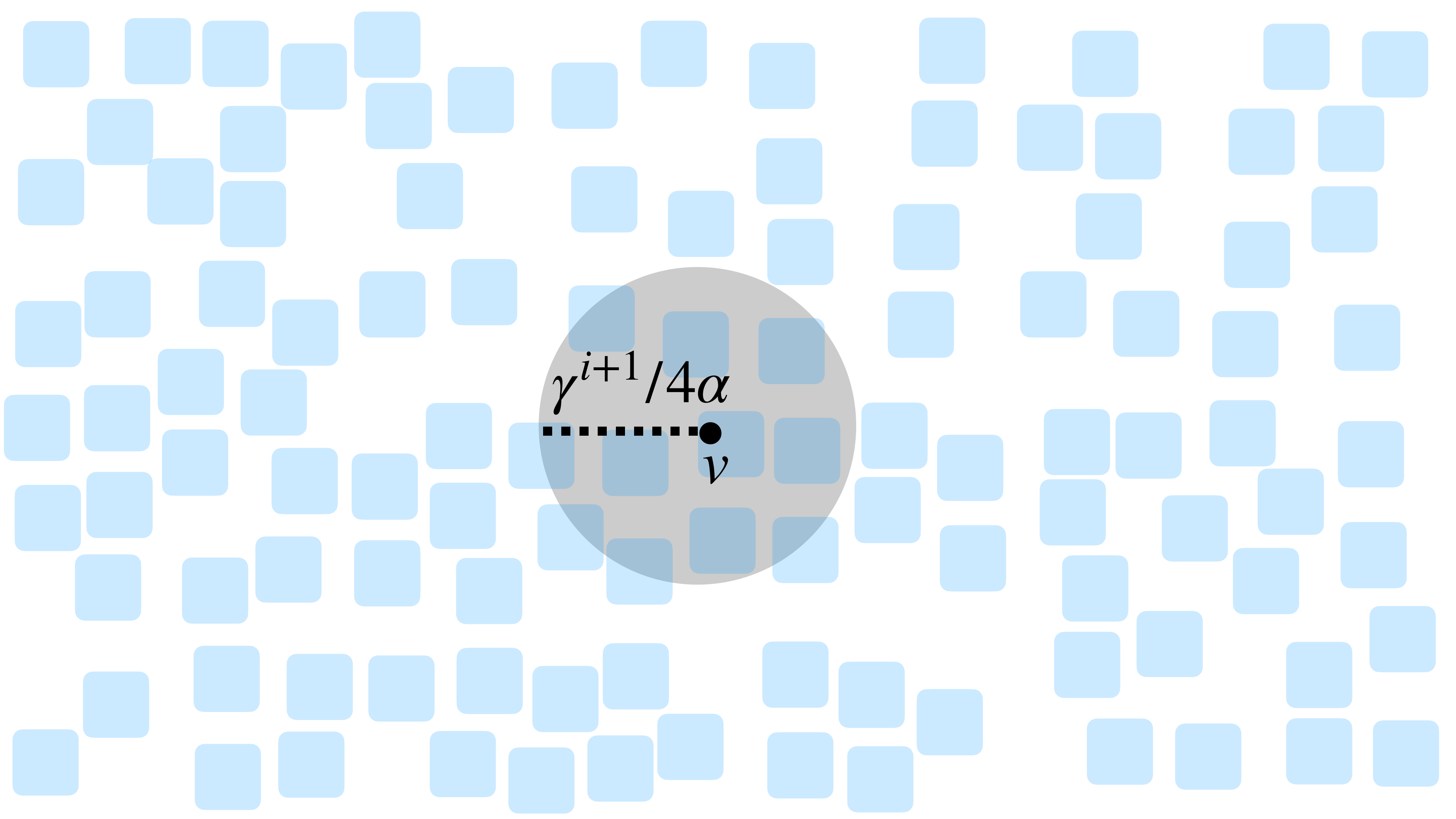}
        \caption{$\mcC_{i}$ and one $\gamma^{i+1}/\alpha$ ball.}\label{sfig:alg1}
    \end{subfigure}    \hfill
    \begin{subfigure}[b]{0.3\textwidth}
        \centering
        \includegraphics[width=\textwidth,trim=0mm 0mm 0mm 0mm, clip]{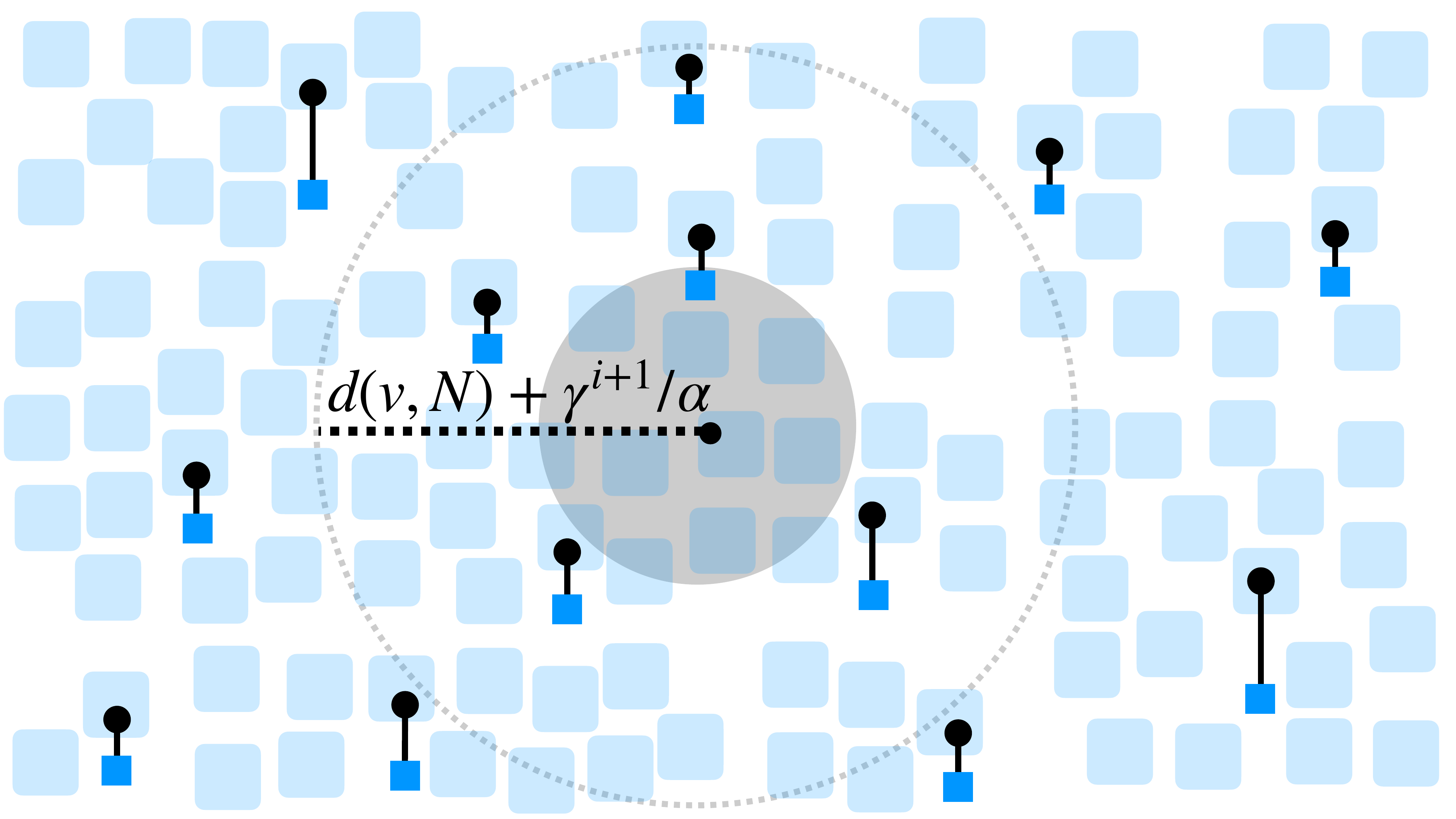}
        \caption{Dangling net $N$.}\label{sfig:alg2}
    \end{subfigure}    \hfill
    \begin{subfigure}[b]{0.3\textwidth}
        \centering
        \includegraphics[width=\textwidth,trim=0mm 0mm 0mm 0mm, clip]{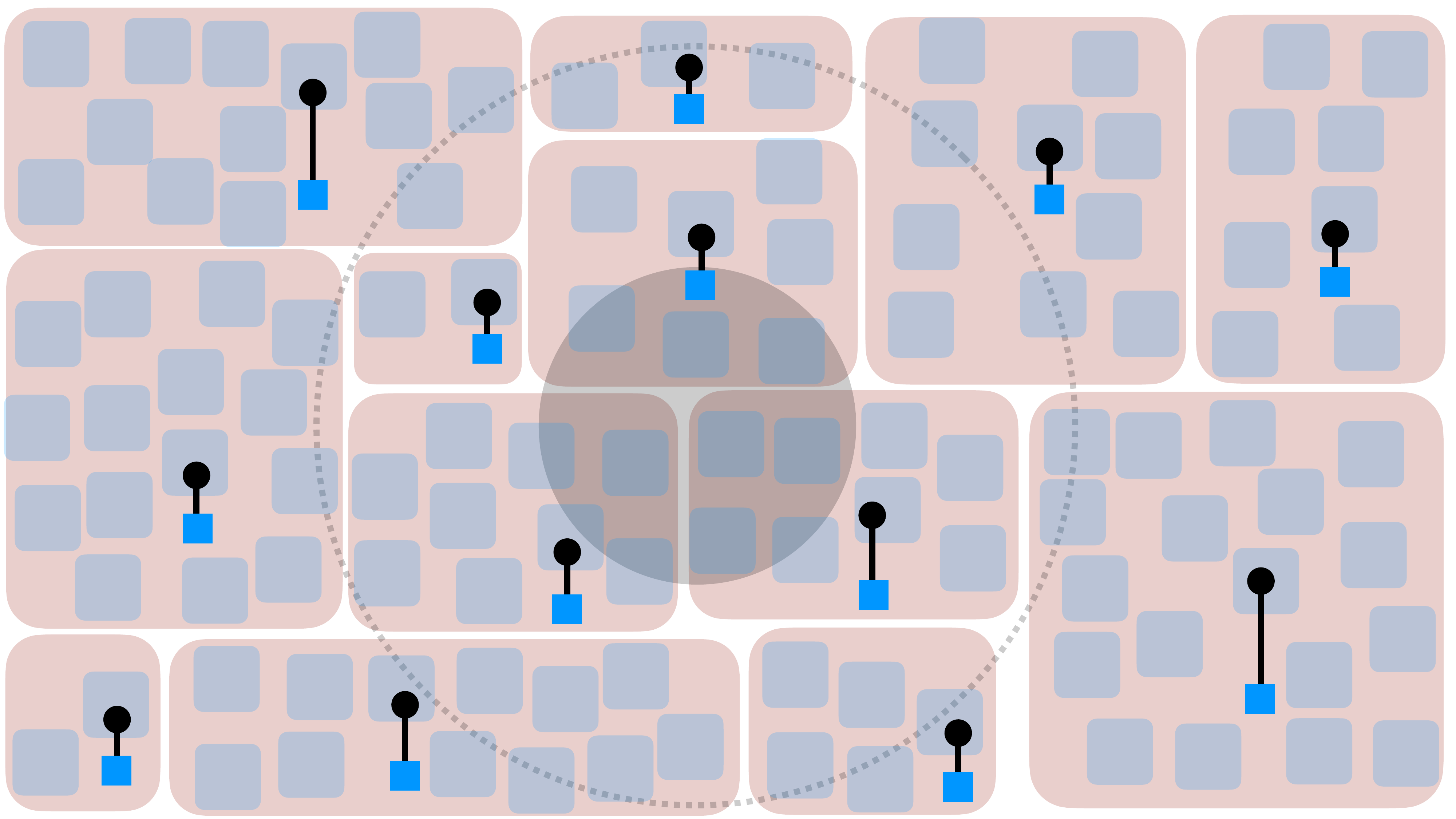}
        \caption{$\mcC_{i+1}$ via cluster aggregation.}\label{sfig:alg3}
    \end{subfigure}
    \caption{An illustration of our algorithm for coarsening $\mcC_{i}$ to $\mcC_{i+1}$. \ref{sfig:alg1} gives $\mcC_i$ as transparent blue squares and one ball of radius $\gamma^{i+1}/4\alpha$ centered at $v$. \ref{sfig:alg2} illustrates our dangling net $N$ (opaque blue squares) and the fact that there are $\tau$ net vertices within distance $d(v, N) + \gamma^{i+1}/\alpha$ of $v$; in this case $5$ net vertices. \ref{sfig:alg3} gives the $\mcC_{i+1}$ resulting from cluster aggregation (in red) which guarantees that every vertex 
    $u\in B_G(v,\frac{\gamma^{i+1}}{4\alpha})$ is sent to a portal at distance at most $d_{G}(v,N)+\beta\cdot\gamma^{i}+\frac{\gamma^{i+1}}{2\alpha}\le d_{G}(v,N)+\frac{\gamma^{i+1}}{\alpha}$
    from $v$, which are exactly the $5$ net vertices.}\label{fig:alg}
\end{figure}

\subsection{Improved Cluster Aggregation}
We now briefly discuss our techniques for producing improved cluster aggregation solutions.

\paragraph*{Our General Graphs Approach.} Our approach for achieving an $O(\log n)$-distortion cluster aggregation is a round-robin process of $O(\log n)$ phases. In each phase, each unassigned cluster has a constant probability of merging with a cluster containing a portal. We accomplish this as follows. Define the \emph{maximal internally disjoint} (MID) path of an unassigned cluster $C$ as the maximal prefix of the shortest path from some representative node in $C$ to a portal which is disjoint from all assigned clusters. In each phase we iterate through the clusters with portals. For each cluster $C_i'$ with a portal we repeatedly flip a fair coin until we get a tails at which point we move on to the next cluster with a portal. Each time we get a heads we do an ``expansion iteration'', merging $C_i'$ with all clusters incident to a MID path that ends at $C_i'$. Intuitively, this is a sort of geometric ball growing where MID paths are always treated as having weight $0$. See \Cref{fig:CAalg} for an illustration. 

Every unassigned cluster is assigned in each phase with constant probability. Therefore with high probability after $O(\log n)$ iterations, every cluster is assigned to some portal.
The additive distortion of this process can be bounded by the maximum number of heads any one cluster gets across all phases. The key to arguing $O(\log n)$ distortion is to observe that, while the distortion we incur may be as large as $\Theta(\log n)$ in one phase, the distortion any portal incurs \emph{across all phases} is also $O(\log n)$. Thus, we bound across all phases at once. This can be contrasted with \cite{busch2012split} who performed $O(\log n)$ phases of merging with $O(\log n)$ distortion per phase.

\paragraph*{Our Approaches for Special Graph Families.}
For special graph families, we exploit the structure of the family to argue that there are limited conflicts when merging clusters. 
\begin{itemize}
    \item \textbf{Trees.} The main observation for trees is the fact that every path has a monotone ``up'' part and then a monotone ``down'' part. We first merge clusters with shortest paths to portals that never ``switch directions.'' We then merge clusters that contain the vertex where one such path does switch directions. Lastly, all remaining clusters have a monotone shortest path towards an assigned cluster along which we can merge. 
    \item \textbf{Pathwidth.}  
    Using the structure of the path decomposition $T$, 
    we show it is possible to perform the cluster aggregation in $\pw + 1$ phases. In each phase, at least one node is assigned to a portal from each bag of $T$.
    Each phase incurs $O(1)$ additive distortion giving at most $O(\pw)$ total distortion.
    To control the distortion at each phase, 
    we use consecutive bags in $T$ to organize unassigned clusters into a sequence of groups.
    Clusters in non-adjacent groups are conflict-free, 
    but they may conflict with adjacent groups in the sequence.
    Hence, the aggregation can be done in two subphases, one for the odd subsequence and the other for the even subsequence of groups, where each subphase adds only $O(1)$ distortion.
    
    \item \textbf{Bounded doubling dimension.} Our approach for graphs with bounded doubling dimension is quite different from that for general graphs. In general graphs after $i$ phases each cluster remains unassigned with probability at most $2^{-i}$ and the expected distortion of each vertex is proportional to $i$. A union bound then shows that $O(\log n)$ phases suffice. One approach for the doubling case would be to reuse this algorithm and argue that far away clusters have independent probabilities of being clustered. Since the bounded doubling dimension guarantees that the number of ``nearby'' clusters is small, one might then attempt to use the  Lov\'asz Local Lemma to argue that no cluster remains unassigned after $\poly(\ddim)$ iterations.

    Unfortunately, this approach does not work. 
    Indeed, even if one ensures that the number of expansion iterations is always constant, it is possible for a portal $p$ to affect the merging of a far away cluster $C$ (even though $\Pr[f(C)=p]=0$). Instead, we run the clustering algorithm of~\cite{MillerPX2013} with the portals as centers. Every cluster aggregation cluster $C$ which is deeply inside an MPX cluster is assigned to the center (i.e.\ portal) of this MPX cluster. We then contract all the assigned clusters and continue recursively. This process guarantees that the probability that a cluster $C$ remains unassigned after $i$ iterations is $2^{-\Omega(i)}$ and, unlike our general graphs algorithm, far away clusters \emph{are} completely independent, allowing us to invoke Lov\'asz Local Lemma and obtain the desired bound.
\end{itemize}

\section{Notation, Conventions and Preliminaries}\label{sec:notation}
We review the (mostly standard) notation we use throughout this work.

\paragraph*{General.} We use $\sqcup$ for disjoint union; i.e.\ $U \sqcup V$ is the same set as $U \cup V$ but indicates that $U \cap V = \emptyset$.  $\text{Geom}(p)$ is the geometric distribution where the probability for value $i$ is $(1-p)^{i-1}\cdot p$, and the expectation is $\frac1p$. $\text{Bin}(n,p)$ stands for a binomial distribution with $n$ samples, each with success probability $p$.

\paragraph*{Graphs.} Given edge-weighted graph $G = (V,E, w)$ and vertex subset $U \subseteq V$, we let $G[U] = (U, \{e : e \subseteq U\}, w)$ be the induced graph of $U$. Given two edge-weight functions $w$ and $w'$ on disjoint edge sets $E$ and $E'$, we let $w \sqcup w'$ be the edge-weight function that gives $w(e)$ to each $e \in E$ and $w'(e')$ to each $e' \in E'$. We let $d_G(u,v)$ be the weight of the shortest path between $u$ and $v$ according to $w$ in $G$ and for $S \subseteq V$ we let $d_G(v, S) = \min_{u \in S} d_G(v,u)$. The diameter of $G$ is the maximum distance between a pair of vertices, i.e.\ $\max_{u,v \in U} d_{G}(u,v)$. The strong diameter of $S\subseteq V$ is the diameter of the induced graph $G[S]$, as opposed to the weak diameter $\max_{u,v \in S} d_{G}(u,v)$ (which is the maximum distance w.r.t. $d_G$). A partition $\mcC$ of $V$ has strong (resp.\ weak) diameter $\Delta$ if $G[C_i]$ has strong (resp.\ weak) diameter for every $C_i \in \mcC$. The (closed) ball $B_G(v, r) := \{u : d_G(u,v) \leq r\}$ is all vertices within distance $r$ from $v$ in $G$. We drop the $G$ subscript when the graph is clear from context. We let $n := |V|$ be the number of nodes in $G$ throughout this paper. A metric space $(X,d_X)$ induces a complete graph $G$ with $X$ as a vertex set, where the weight of the edge $\{u,v\}$ equals to the metric distance $d_X(u,v)$.

\paragraph*{Path Decomposition and Pathwidth.} 
Given a graph $G=(V,E)$, a \emph{path decomposition} of $G$ is a path
$P$ with nodes $X_1,\dots,X_s$ (called \emph{bags}) where each $X_i$ is
a subset of $V$ such that the following properties hold:
\begin{OneLiners}
	\item For every edge $\{u,v\}\in E$, there is a bag $X_i$ containing
	both $u$ and $v$.
	\item For every vertex $v\in V$, the set of bags containing $v$ form a
	connected subpath of $P$.
\end{OneLiners}
The \emph{width} of a path decomposition is $\max_i\{|X_i|-1\}$. The \emph{pathwidth} $\pw$ of $G$ is the minimum
width of a path decomposition of $G$.  For readers familiar with tree decompositions, a path decomposition is just a tree decomposition whose tree is a path. 
%

\paragraph*{Doubling dimension.}  The doubling dimension of a metric space is a measure of its local ``growth rate''. 
A metric space $(X,d)$ has doubling constant $\lambda$ if for every $x\in X$ and radius
$r>0$, the ball $B(x,2r)$ can be covered by $\lambda$ balls of radius $r$; that is, there exist $\lambda$ such balls whose union contains $B(x,2r)$. The doubling dimension is defined as $\ddim=\log\lambda$.
We say that a weighted graph $G=(V,E,w)$ has doubling dimension $\ddim$, if the corresponding shortest path metric $(V,d_G)$ has doubling dimension $\ddim$.
A $d$-dimensional $\ell_p$ space has $\ddim=\Theta(d)$, every $n$-vertex graph has $\ddim=O(\log n)$, and every weighted path has $\ddim=1$.
The following lemma gives the standard packing property of doubling metrics (see, e.g., \cite{GKL03}).
\begin{lemma}[Packing Property] \label{lem:doubling_packing}
	Let $(X,d)$ be a metric space  with doubling dimension $\ddim$.
	If $S \subseteq X$ is a subset of points with minimum inter-point distance $r$ that is contained in a ball of radius $R$, then 
	$|S| \le \left(\frac{2R}{r}\right)^{\sddim}$ .
\end{lemma}


\paragraph*{Dangling Net Constructions.} We summarize results regarding dangling nets.
\begin{theorem}[\cite{filtser20}]\label{thm:MPXbasedClusteringGeneral}
	Every weighted graph $G=(V,E,w)$ has a poly-time computable $\Delta$-covering $(O(\log n), O(\log n))$-sparse dangling net for every $\Delta>0$.
\end{theorem}
\begin{theorem}[\cite{filtser20}]\label{thm:MPXbasedClusteringPathwidth}
	Every weighted graph $G=(V,E,w)$ with pathwidth $pw$ has a poly-time computable $\Delta$-covering $(O(\pw),O(\pw^2))$-sparse dangling net for every $\Delta>0$.
\end{theorem}
\begin{theorem}[\cite{filtser20}]\label{thm:MPXbasedClusteringDoubling}
	Every weighted graph $G=(V,E,w)$ with doubling dimension $\ddim$ has a poly-time computable $\Lambda$-covering $(O(\ddim),\tilde{O}(\ddim))$-sparse dangling net for every $\Lambda>0$. Furthermore, let $N_V\subseteq V$ be the endpoints of the matching with the dangling net (i.e. $M$ is a matching from $N$ to $N_V$). Then $\min_{u,v\in N_v}d_G(u,v)\ge c_\Lambda\cdot \Lambda$ for some universal constant $c_\Lambda$.
\end{theorem}

Theorems \ref{thm:MPXbasedClusteringGeneral}, \ref{thm:MPXbasedClusteringPathwidth} and \ref{thm:MPXbasedClusteringDoubling} are proven in \cite{filtser20} in the context of ``MPX partitions''. There we sample shifts $\{\delta_t\}_{t\in N}$ and each vertex $v$ joins the cluster of the center $t$ maximizing $\delta_t-d_G(v,t)$. This is equivalent to our framework here, where in our dangling net we add $t$ at distance $\Delta-\delta_t$ from its corresponding vertex in $G$. The statements corresponding to Theorems \ref{thm:MPXbasedClusteringGeneral}, \ref{thm:MPXbasedClusteringPathwidth} and \ref{thm:MPXbasedClusteringDoubling} in \cite{filtser20} are Theorems 4, 13, and 10 (respectively).

\section{Hierarchies via Cluster Aggregation and Dangling Nets}\label{sec:hierarchiesFromCA}

In this section we reduce the existence of strong sparse partition hierarchies to the existence of good dangling nets and cluster aggregation solutions. Our algorithm for doing so is \Cref{alg:Hierarchical}. It may be useful for the reader to recall the relevant definitions: strong sparse hierarchies (\Cref{dfn:hierarchies}), cluster aggregation (\Cref{dfn:clusterAgg}) and dangling nets (\Cref{dfn:danglingNet}). 

\begin{algorithm}[]
	\caption{\texttt{Hierarchical-Strong-Sparse-Partition}}\label{alg:Hierarchical}
	\DontPrintSemicolon
	\SetKwInOut{Input}{input}\SetKwInOut{Output}{output}
	\Input{Weighted graph $G = (V,E,w)$ (edge weights at least $1$), 
		algorithm for $\Delta$-covering $(\alpha, \tau)$-sparse dangling net, algorithm for $\beta$-distortion cluster aggregation.}
	\Output{A $2\beta\cdot\left(2\alpha+1\right)$-hierarchy of strong $(8\alpha+4,\tau,\gamma)$-sparse partitions.}
	\BlankLine
	
	$i=0$ and $\gamma=2\beta(2\alpha+1)$.\;
	Set $\calC_0=\{\{v\} : v\in V\}$.\;
	\While{$\calC_i\ne\{\{V\}\}$}{
		Set $\Delta=2\alpha\beta\cdot\gamma^{i}$.\;
		Compute a $\Delta$-covering $(\alpha, \tau)$-sparse dangling net $N$.\;
		Compute a $\beta$-distortion cluster aggregation solution $f$ on $G+N$ with portals $N$ and clusters $\calC_{i} \cup \{\{t\}\}_{t \in N}$ with corresponding coarsened partition $\mcC' := \{f^{-1}(t) : t \in N\}$.\;
		Let ${\cal C}_{i+1}=\{C\setminus N\mid C\in \calC'\}$.\; 
  		$i\leftarrow i+1$.\;
	}
	\Return $\calC_0,\calC_1, \mcC_2 \ldots$
\end{algorithm}

Formally, we show the following theorem whose proof is illustrated in \Cref{fig:alg}.

\HSSPmainp*

\begin{proof}
We begin by describing our algorithm for strong sparse partition hierarchies in words; see \Cref{alg:Hierarchical} for pseudo-code. Our algorithm proceeds in rounds in a bottom up fashion, with round $0$ being the trivial partition ${\cal C}_0$ to singletons, and round $i$ constructing the coarsening of strong sparse partition ${\cal C}_{i}$ to obtain strong sparse partition ${\cal C}_{i+1}$.  

In the remainder, we elaborate on the coarsening step of
round $i$. Here, we receive as input a strong $\gamma^{i}$-diameter  $(8\alpha+4,\tau)$-sparse partition $\calC_{i}$. Let $\Delta=2\alpha\beta\cdot\gamma^{i}$.
Using the assumption of our theorem, we create a $\Delta$-covering $(\alpha,\tau)$-sparse dangling net $N$.


Next, we apply the cluster aggregation algorithm in the graph $G+N$ using $N$ as the portals and $\mcC_i$ with a singleton cluster for each element of $N$ as the input clusters.
As a result we obtain assignment function $f$ and corresponding coarsening $\mcC' := \{f^{-1}(t)\}_{t \in N}$. We obtain $\mcC_{i+1}$ by removing any vertex in $N$ from any cluster in $\calC'$.

	We now establish that for every $i$, ${\cal C}_i$ forms a strong $\gamma^i$-diameter $(\alpha, \tau)$-sparse partition.  The claim holds
	for $i = 0$ since ${\cal C}_0$ is a strong $\gamma^0$-diameter $(4(\alpha+1),
	\gamma)$-sparse partition. Consider arbitrary $i > 0$. We assume by induction that $\calC_{i}$ is a strong $\gamma^i$-diameter $(4(\alpha+1),\tau)$-sparse partition.
	Recall that $\Delta=2\alpha\beta\cdot\gamma^{i}$.
	
 We begin by bounding the diameter of every cluster in $\mcC_i$. For any vertex $v\in V$, we know that $d_{G+N}(v,N)\le\Delta$. Next we obtain a solution to the cluster aggregation problem $f$ such that
        \begin{align*}
            d_{G+N[f^{-1}(v)]}(v,f(v))\le d_{G+N}(v,N)+\beta\cdot\gamma^{i} \le \Delta+\beta\cdot\gamma^{i}.
        \end{align*}
	It follows that $f^{-1}(v)$ has strong diameter at most 
        \begin{align*}
            2\cdot\left(\Delta+\beta\cdot\gamma^{i}\right)=2\cdot\left(2\alpha\beta+\beta\right)\cdot\gamma^{i}=2\beta\cdot\left(2\alpha+1\right)\cdot\gamma^{i}=\gamma^{i+1}.
        \end{align*}
Finally, in the actual partition that we use, $\calC_{i+1}$, we only remove vertices of degree $1$ and this can only decrease the diameter.
	We conclude that $\calC_{i+1}$ has  strong diameter at most $\gamma^{i+1}$ as required. It is also clear that $\calC_{i+1}$ coarsens $\calC_{i}$.
	
	Next, we prove the ball preservation property. Fix a vertex $v\in V$. 
	Consider a ball $B_G(v, R)$ around $v$ of radius $R=\frac{\Delta}{4\alpha}$.
	For every $u\in B_{G}(v,R)$, the cluster aggregation solution assigns $u$ to a portal $t_u\in N$.
	By the guarantees of cluster aggregation we have
	\begin{align*}
		d_{G+N}(v,t_u) & \le d_{G+N}(v,u)+d_{G+N}(u,t_u)\\
		& \le d_{G}(v,u)+d_{G+N}(u,N)+\beta\cdot\gamma_{i}\\
		& \le d_{G+N}(v,N)+2d_{G}(v,u)+\beta\cdot\gamma_{i}\\
		& \le d_{G+N}(v,N)+\frac{\Delta}{2\alpha}+\frac{\Delta}{2\alpha}\\
            &=d_{G+N}(v,N)+\frac{\Delta}{\alpha}~.
	\end{align*}
	As $N$ is a $\Delta$-covering $(\alpha, \tau)$-sparse dangling net, it holds that $$\left|\left\{ t\in N\mid d_{G+N}(v,t)\le d_{G+N}(v,N)+\frac{\Delta}{\alpha}\right\} \right|\le\tau.$$ It follows that the vertices in $B_{G}(v,R)$ are assigned to at most
	$\tau$ different portals, as required.
	
	Finally, to conclude that $\mcC_{i+1}$ is $(8\alpha + 4, \tau)$-sparse, we observe that
	\[
	\frac{\gamma^{i+1}}{R}=\frac{\gamma^{i+1}}{\frac{\Delta}{4\alpha}}=\frac{4\alpha\cdot\gamma^{i+1}}{2\alpha\beta\cdot\gamma^{i}}=\frac{2\gamma}{\beta}=\frac{2\cdot2\beta\cdot\left(2\alpha+1\right)}{\beta}=8\alpha+4~,
	\]
	concluding our analysis and proof.
\end{proof}

\section{Improved Cluster Aggregation}\label{sec:improvedCA}
Having reduced strong sparse partition hierarchies to dangling nets and cluster aggregation in the previous section, we now give our new algorithms for cluster aggregation in general graphs, trees, doubling dimension-bounded and pathwidth-bounded graphs.

The reader may want to review the definition of cluster aggregation (\Cref{dfn:clusterAgg}). Throughout this section, given cluster aggregation solution $f$, we will make use of the notion of the detour of a cluster aggregation solution; informally, how much extra distance a vertex travels in the solution.
\begin{definition}[Cluster Aggregation Detour]
Given cluster aggregation solution $f$ in graph $G$ on portals $P$, we let the detour of vertex $v$ be
    \begin{align*}
    \dtr_f(v) := d_{G[f^{-1}(f(v))]}(v, f(v)) - d_G(v, P).
\end{align*}
\end{definition}
\noindent Observe that cluster aggregation solution $f$ has distortion $\beta$ if $\dtr_f(v) \leq \beta \cdot \Delta$ for every vertex $v$.

\subsection{Cluster Aggregation in General Graphs}\label{sec:CAgeneral}

We begin by demonstrating how to achieve $O(\log \kappa)$-distortion cluster aggregation solutions in general graphs when we are given $\kappa \leq n$ input clusters.

\begin{restatable}{theorem}{CAGeneral}\label{thm:caGen}
	Every instance of cluster aggregation with input partition $\calC=\{C_1,\ldots ,C_\kappa\}$ has an $O(\log \kappa)$-distortion solution that can be computed in polynomial time.
\end{restatable}

\noindent Our main approach is to grow the cluster of each portal in a round-robin and geometric fashion but treat each vertex's path to its nearest cluster with a portal as having length $0$; this idea is generally in the spirit of the star decompositions of \cite{dhamdhere2006improved} (see also the related Relaxed Voronoi algorithm \cite{Fil19sicomp}).

To formalize this, for each cluster $C_i \in \mc{C}$, arbitrarily choose a representative vertex $v_i \in C_i$, and let $\pi_i$ denote a shortest path in $G$ from $v_i$ to its closest portal in $P$. At all times in the algorithm, we refer to the \textit{maximal internally disjoint} (MID) prefix of $\pi_i$ as $\pi_i'$. It is the maximal prefix of $\pi_i$ such that its final node is the only node of the prefix belonging to a cluster already assigned to some portal. We denote the final node of the prefix by $\final(\pi_i')$. 
Initially no clusters are assigned, and thus $\pi_i'=\pi_i$, and $\final(\pi_i')$ is the closest portal to $v_i$.

\begin{figure}
    \centering
    \begin{subfigure}[b]{0.24\textwidth}
        \centering
        \includegraphics[width=\textwidth,trim=0mm 0mm 0mm 0mm, clip]{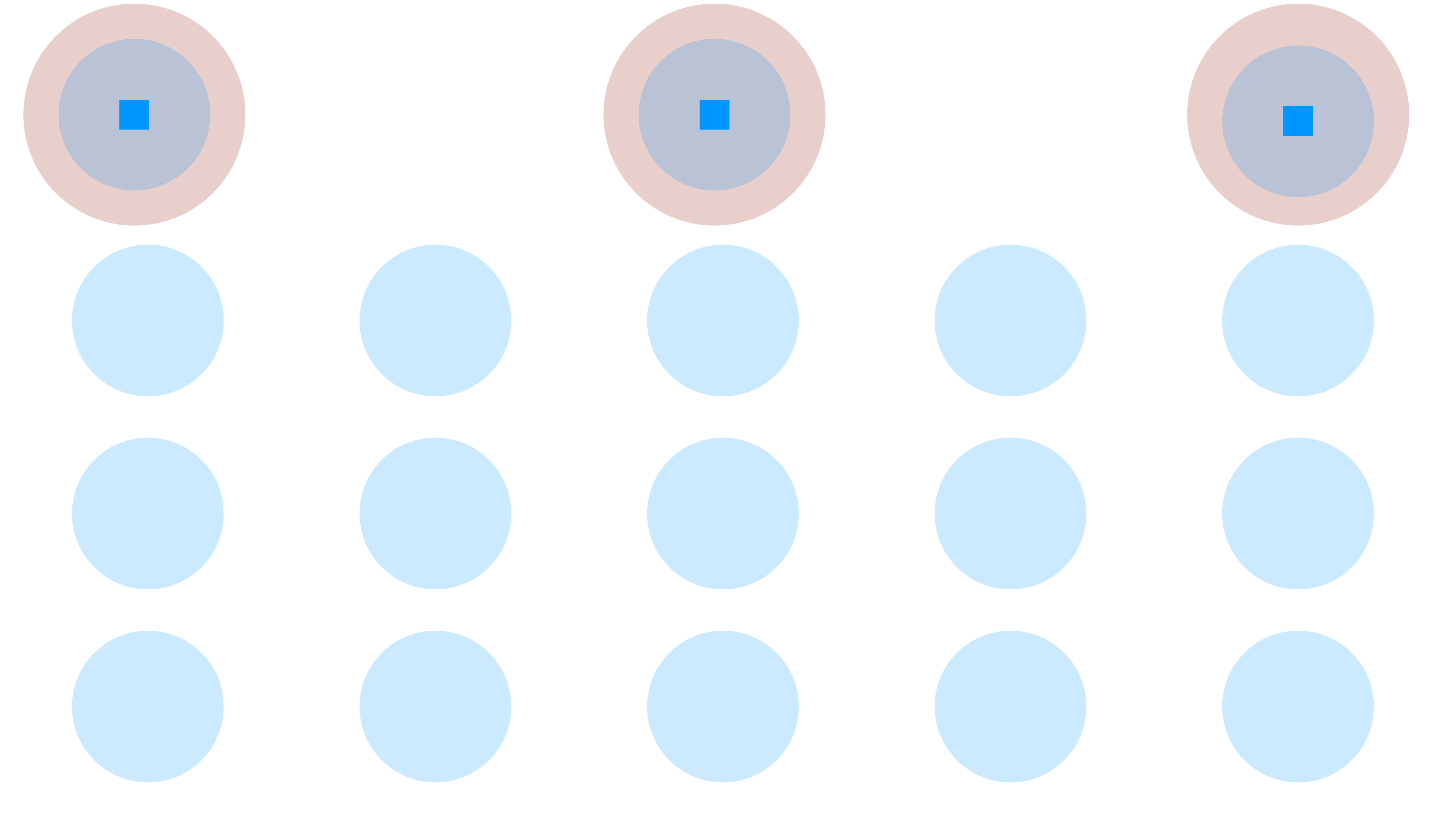}
        \caption{Input partition.}\label{sfig:CAAlg1}
    \end{subfigure}    \hfill
    \begin{subfigure}[b]{0.24\textwidth}
        \centering
        \includegraphics[width=\textwidth,trim=0mm 0mm 0mm 0mm, clip]{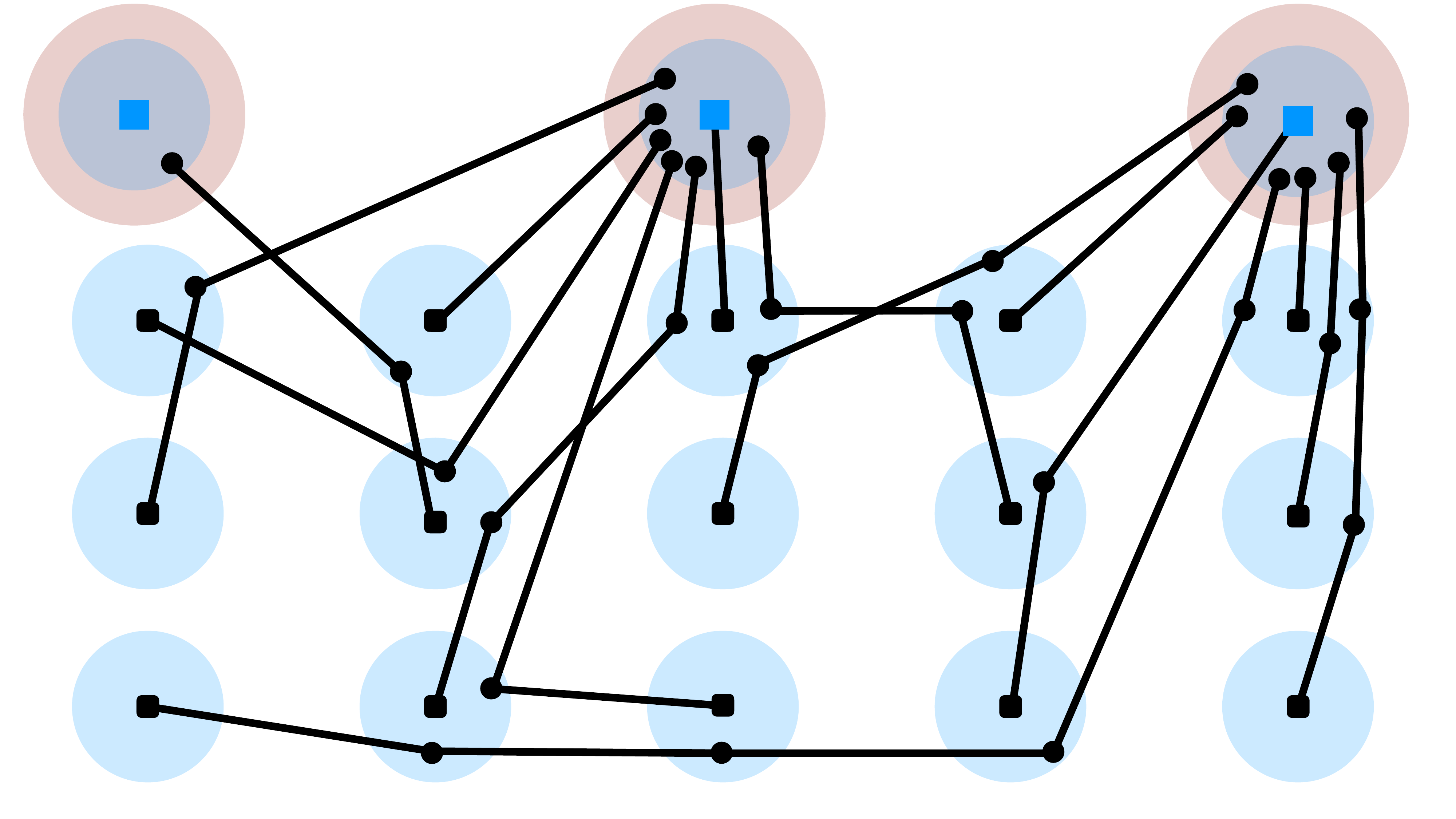}
        \caption{Initial MID paths.}\label{sfig:CAAlg2}
    \end{subfigure}    \hfill
    \begin{subfigure}[b]{0.24\textwidth}
        \centering
        \includegraphics[width=\textwidth,trim=0mm 0mm 0mm 0mm, clip]{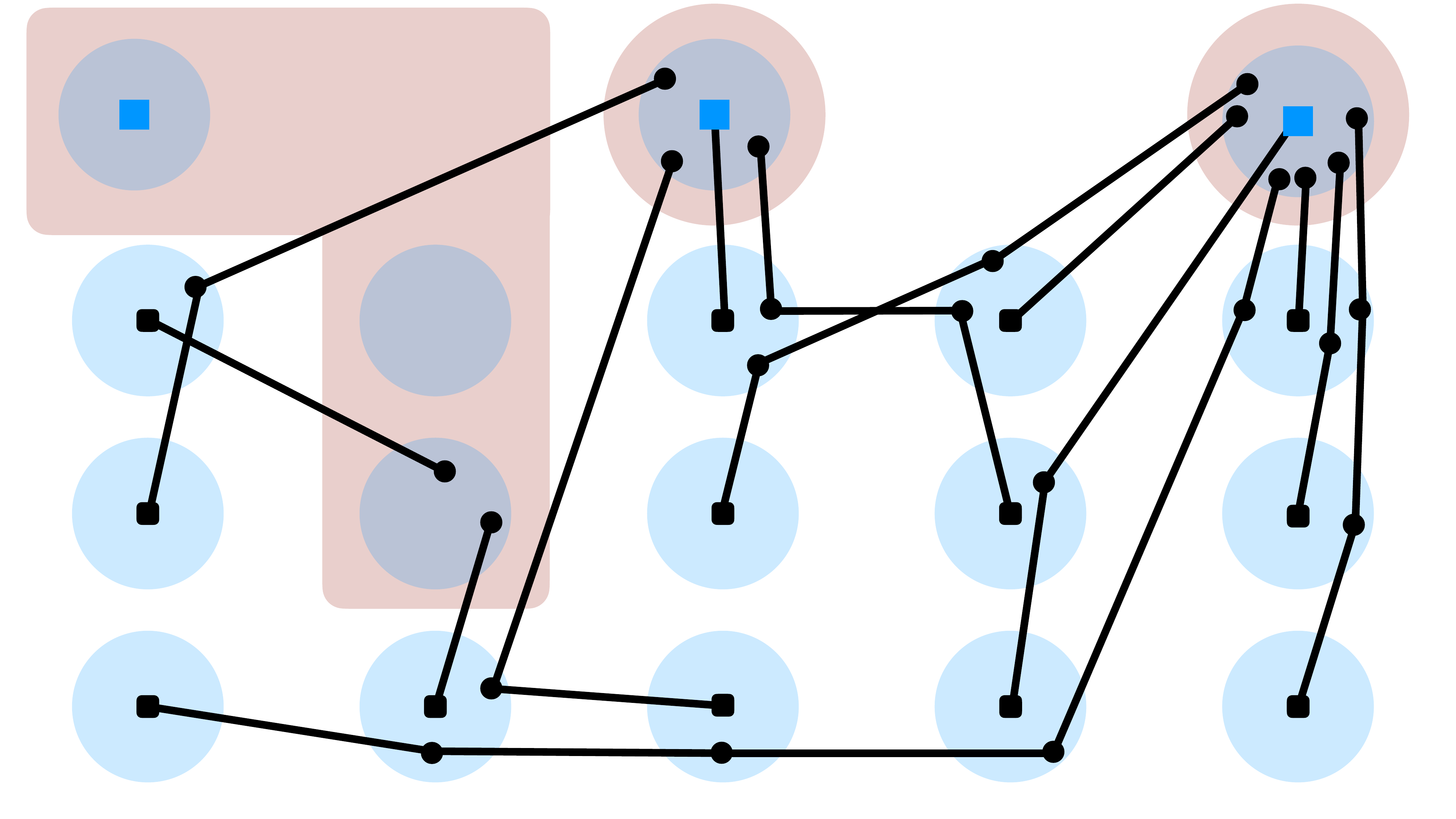}
        \caption{Update 1.}\label{sfig:CAAlg3}
    \end{subfigure}
    \begin{subfigure}[b]{0.24\textwidth}
        \centering
        \includegraphics[width=\textwidth,trim=0mm 0mm 0mm 0mm, clip]{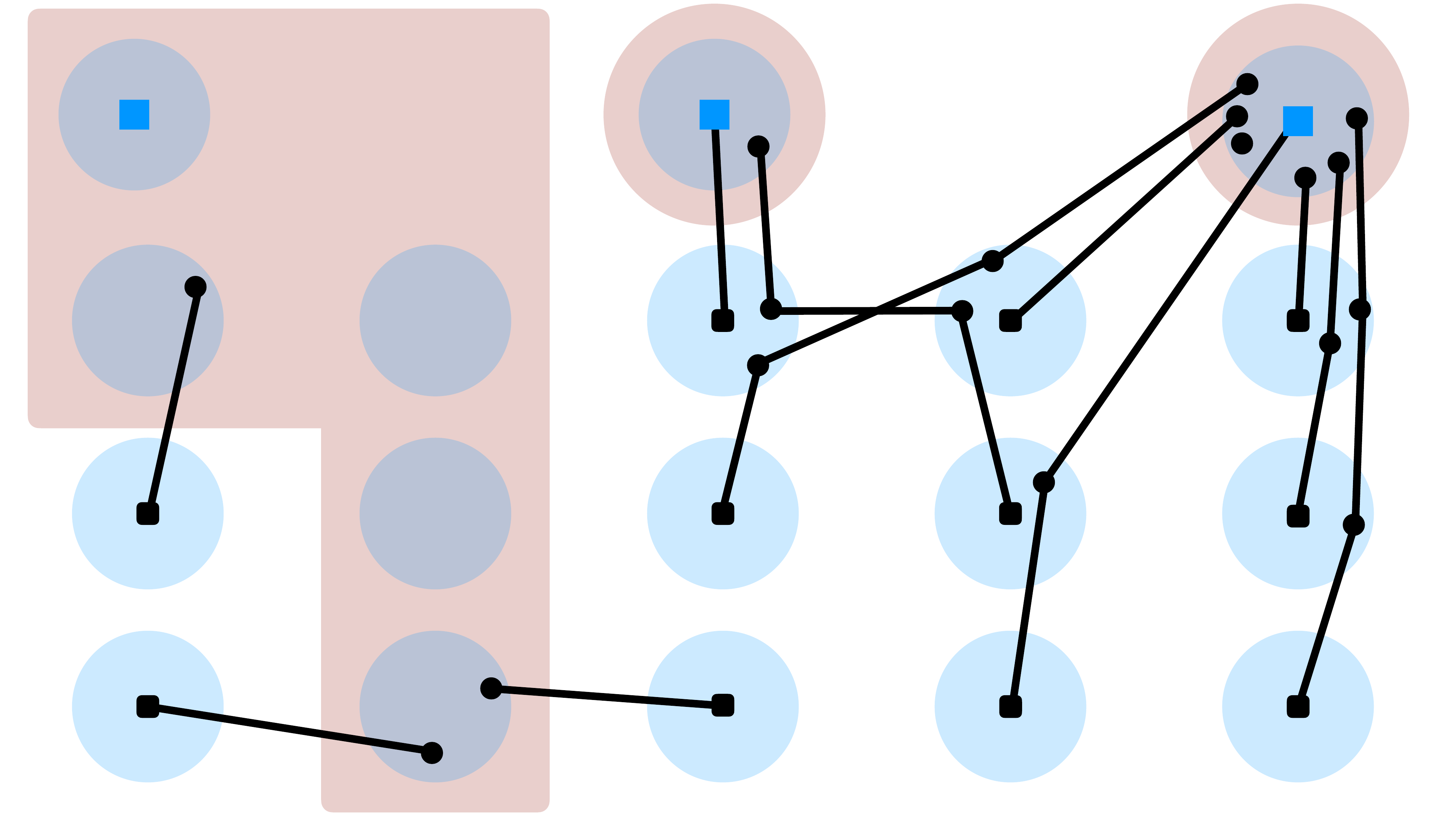}
        \caption{Update 2.}\label{sfig:CAAlg4}
    \end{subfigure}
        \begin{subfigure}[b]{0.24\textwidth}
        \centering
        \includegraphics[width=\textwidth,trim=0mm 0mm 0mm 0mm, clip]{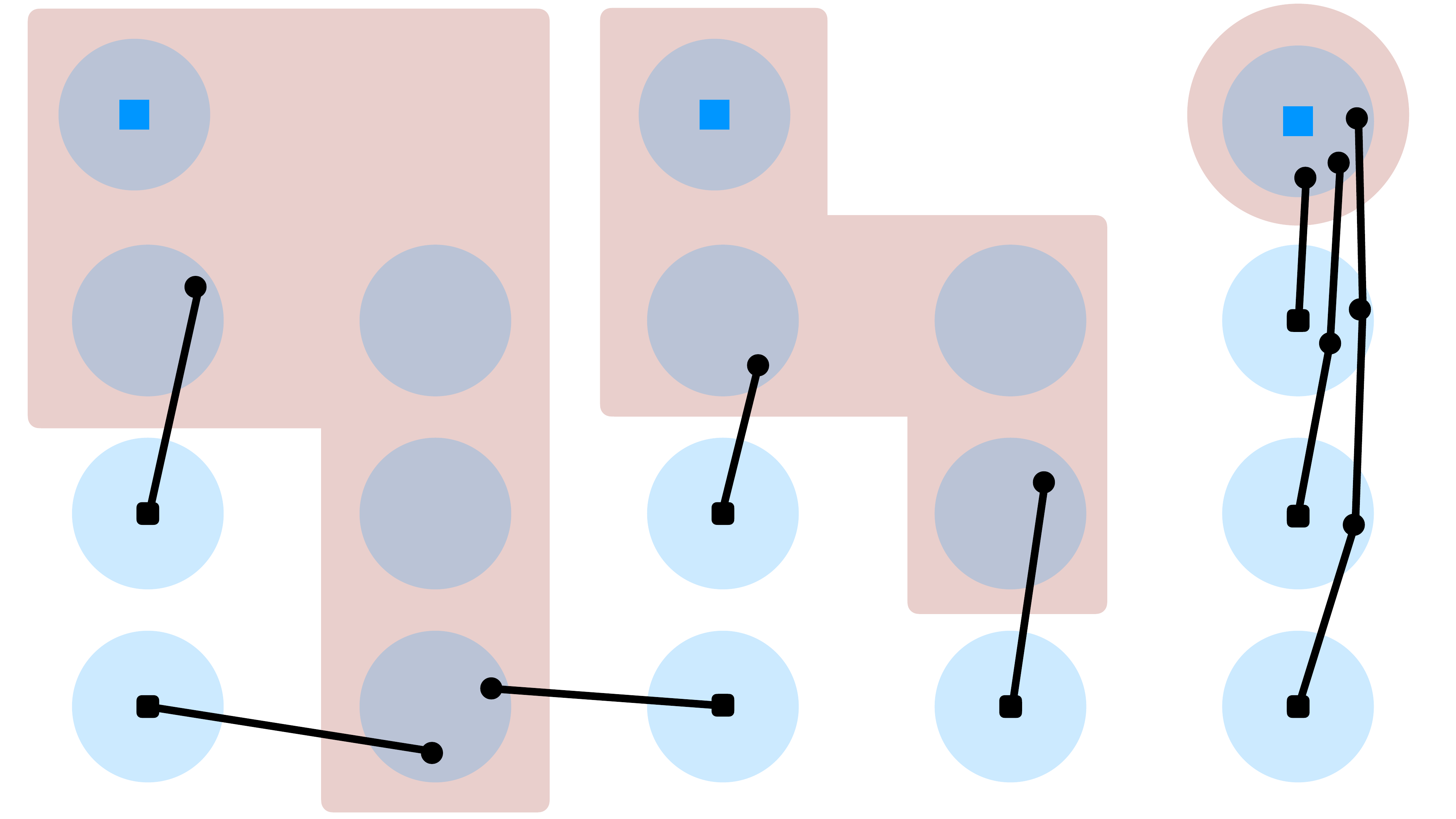}
        \caption{Update 3.}\label{sfig:CAAlg5}
    \end{subfigure}    \hfill
    \begin{subfigure}[b]{0.24\textwidth}
        \centering
        \includegraphics[width=\textwidth,trim=0mm 0mm 0mm 0mm, clip]{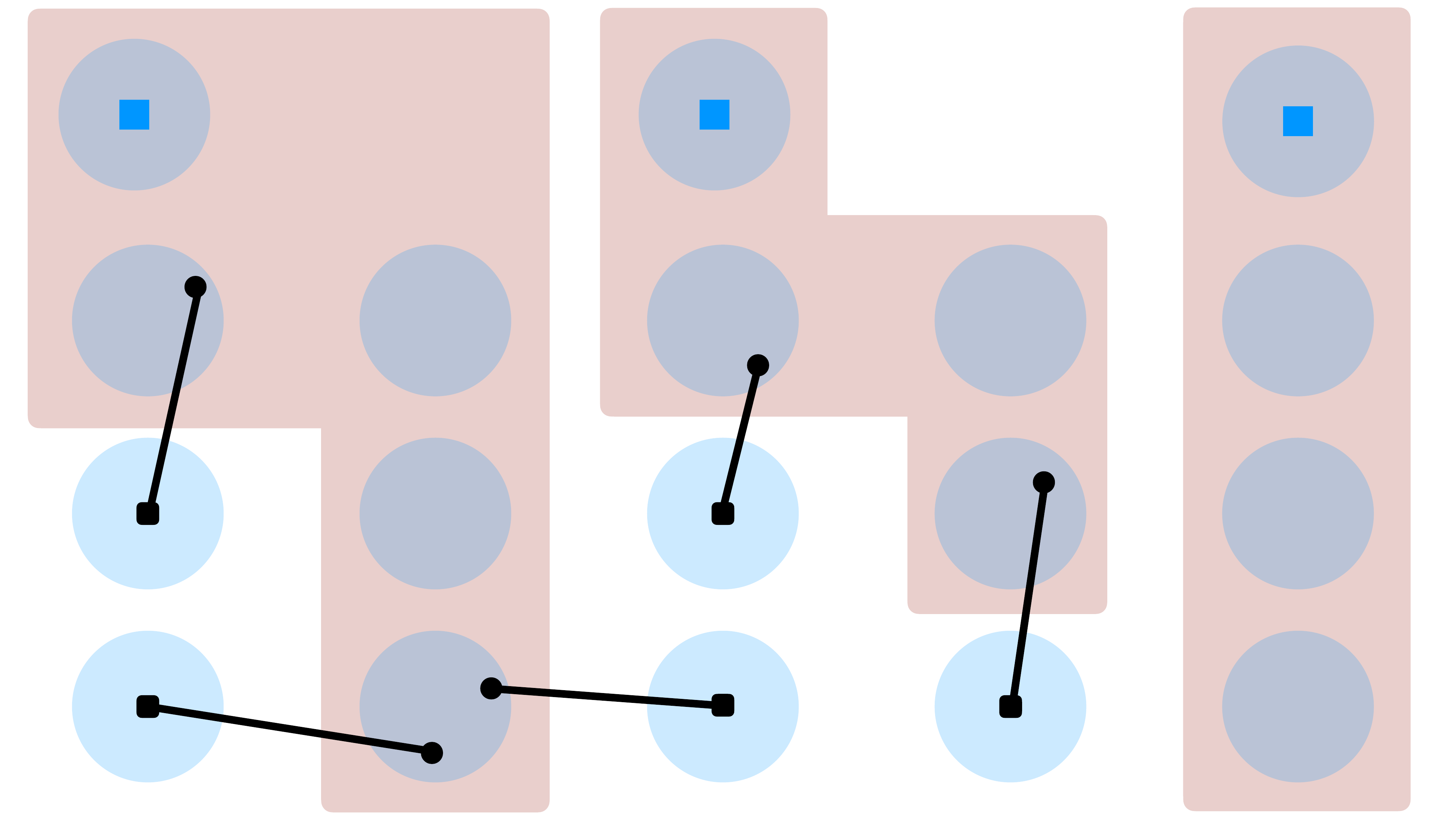}
        \caption{Update 4.}\label{sfig:CAAlg6}
    \end{subfigure}    \hfill
    \begin{subfigure}[b]{0.24\textwidth}
        \centering
        \includegraphics[width=\textwidth,trim=0mm 0mm 0mm 0mm, clip]{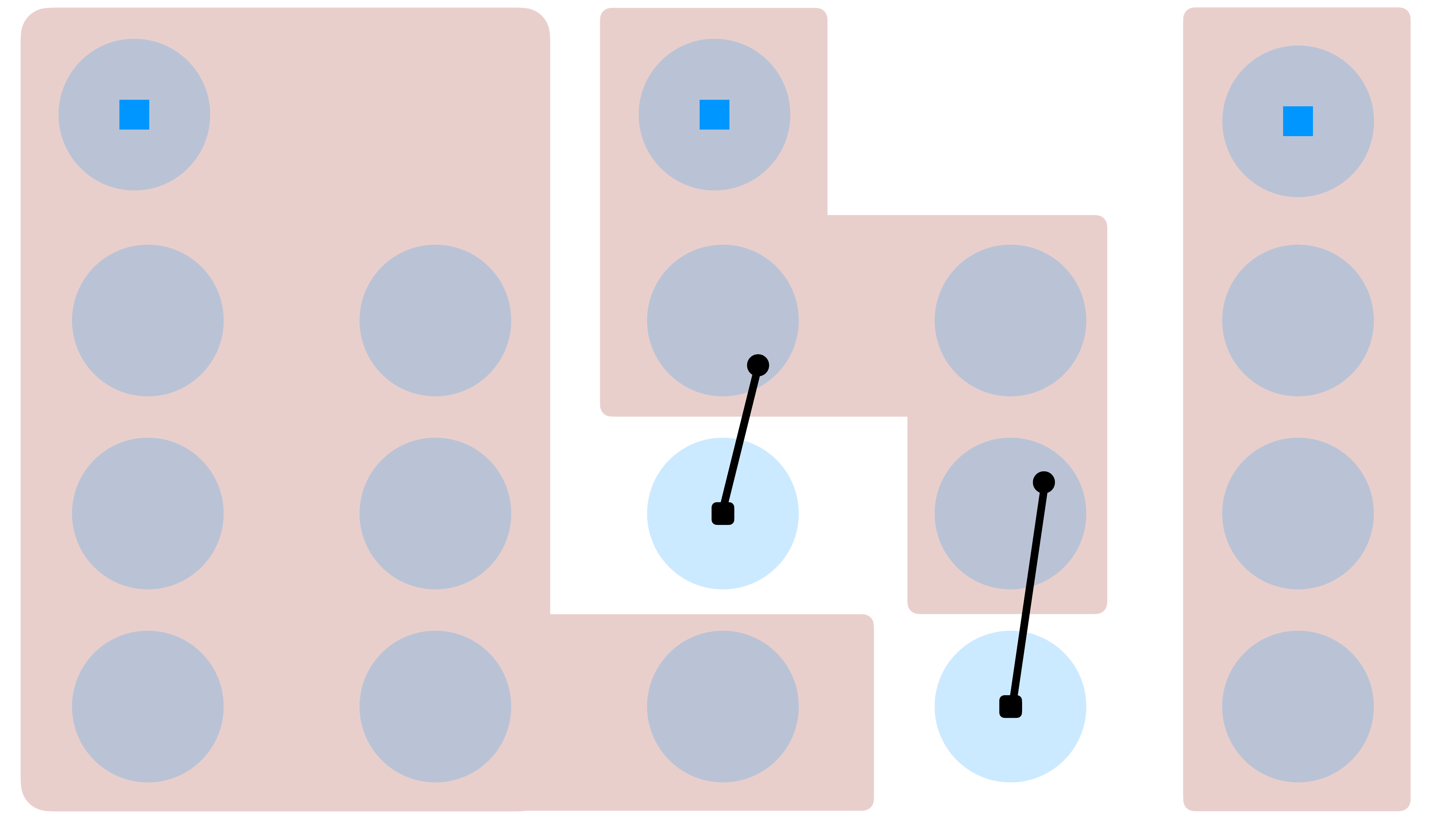}
        \caption{Update 5.}\label{sfig:CAAlg7}
    \end{subfigure}
    \begin{subfigure}[b]{0.24\textwidth}
        \centering
        \includegraphics[width=\textwidth,trim=0mm 0mm 0mm 0mm, clip]{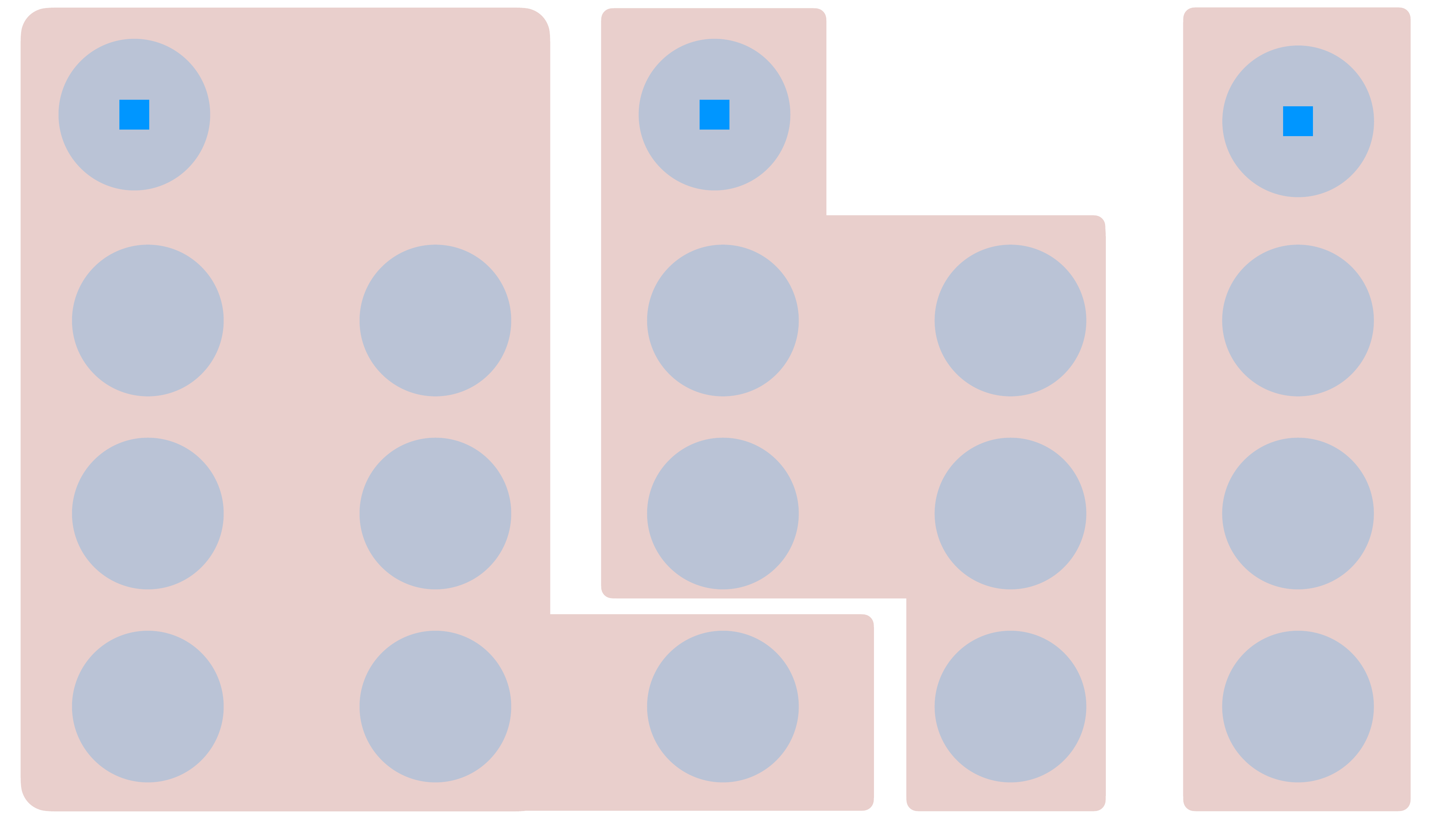}
        \caption{Update 6 (output).}\label{sfig:CAAlg8}
    \end{subfigure}
    \caption{An illustration of our cluster aggregation algorithm. \ref{sfig:CAAlg1} gives the initial partition $\mcC$ in blue and initialized output $\mcC'$ in red. \ref{sfig:CAAlg2} gives the initial MID paths.  We assume that the geometric random variables (left to right) is $2, 1,1$ in the first round and $1,1,0$ in the second round. \ref{sfig:CAAlg3}, \ref{sfig:CAAlg4}, \ref{sfig:CAAlg5}, \ref{sfig:CAAlg6}, \ref{sfig:CAAlg7} and \ref{sfig:CAAlg8} give the updated $\mcC'$ and MID paths after each heads.}\label{fig:CAalg}
\end{figure}

Also observe that we may assume without loss of generality that no cluster contains more than one portal: any assignment that uses more than one portal contained in a given cluster $C_i$ has infinite detour (that is, the portal not assigned cluster $C_i$ is not reachable from any of its assigned clusters), and the use of one portal over another can increase the detour by at most $\Delta$. 

\paragraph*{Algorithm Overview} Order the portals arbitrarily $p_1, \ldots, p_L$, and proceed in rounds. In every round $j$, each portal $p_\ell$ in sequence expands $f^{-1}(p_\ell)$, the set of clusters assigned to it, by claiming \textit{all} clusters $C_i\in\mc{C}$ such that $f(\final(\pi_i'))=p_\ell$, and all of the clusters along the paths $\pi_i'$. In other words, $p_\ell$ claims all the clusters $C_k$ such that for some cluster $C_i$ with $f(\final(\pi_i'))=p_\ell$, $C_k$ intersects $\pi_i'$.

Portal $p_\ell$ repeats this expansion a geometric variable $g_\ell^{(j)}$-many times, then we move on to the next portal. Clearly $f^{-1}(p_\ell)$ remains connected. We will show that only $O(\log(\abs{\mc{C}}))$ rounds are needed to assign every cluster to a portal, and that the total detour of a node assigned to any portal $p_\ell$ is at most $2\Delta \cdot \sum_j g_\ell^{(j)}$, which will suffice to prove the theorem. The algorithm is presented formally in \Cref{alg:cluster-agg-general} and illustrated in \Cref{fig:CAalg}.

\begin{algorithm}[h]
	\caption{\texttt{Cluster aggregation for general graphs}}
        \label{alg:cluster-agg-general}
	\DontPrintSemicolon
	\SetKwInOut{Input}{input}\SetKwInOut{Output}{output}
	\Input{Weighted graph $G = (V, E,w)$, portal set $P \subseteq V$, partition $\mc{C} = \{C_i\}_i$ into clusters of strong diameter at most $\Delta$.}
	\Output{Assignment $f: \mc{C} \to P$ of additive	 distortion $\beta = O(\log\abs{\mc{C}})$.}
	\BlankLine
	
	Name the portals $P = \{p_1, \ldots, p_L\}$\; 
	   
	For each $p_\ell$ contained in cluster $C_i$, set $f(C_i)=p_\ell$\;
	
	\For{rounds $j = 1, 2, \ldots, 10\log\abs{\mc{C}}$}{
		\For{portals $p_\ell = p_1, \ldots, p_L$}{
			Draw $g_{\ell}^{(j)} \sim \text{Geom}(\frac{1}{2})$\;
			\For{$\text{h} = 1, \ldots, g_\ell^{(j)}$ (expansion iterations)}{
            Set $\mc{U}_{1}=\left\{ C_{i}\in\mc{C}\mid C_{i}\text{ is unassigned and }f(\final(\pi_{i}'))=p_{\ell}\right\}$\;
            
            Set $\mc{U}_{2}=\left\{ C_{j}\in\mc{C}\mid\exists C_{i}\in\mc{U}_{1}\text{ such that }C_{j}\cap\pi_{i}'\ne\emptyset\right\}$\hfill\tcp*[h]{Note $\mc{U}_{1}\subseteq \mc{U}_{2}$}\;
            For every cluster $C_i\in \mc{U}_{2}$ set $f(C_i)=p_\ell$\;
			}
		}
	}
	\Return $f$\;
\end{algorithm}

\begin{proof}[Proof of \Cref{thm:caGen}]
	We begin by showing that with high probability, after $10\log|\mcC|$ rounds $f$ is defined on the entire set $C$.
\begin{lemma}
\label{lem:assignment-whp}
The algorithm assigns every cluster, with high probability.
\end{lemma}
\begin{proof}
We claim that in each round $j$, an unassigned cluster $C_i$ has probability at least $\frac{1}{2}$ of being assigned to a portal. Consider the MID prefix $\pi_i'$ of $\pi_i$ at the start of round $j$, and let $p_{\ell^*}=f(\final(\pi_i'))$. 
If no vertex along $\pi_i'$ was assigned to another cluster before iteration $\ell^*$, then in iteration $\ell^*$ we will set $f(C_i)=p_{\ell^*}$, and be done with $C_i$.
Otherwise, let $\ell'$ be the first iteration where some node $u$ lying on $\pi_i'$ is assigned to some other cluster.
It may be that $C_i$ was assigned, and we are done. Otherwise, let $h$ be the expansion iteration for $p_{\ell'}$ at which this first occurs. At this point the MID prefix $\pi_i'$ is updated accordingly and $f(\final(\pi_i'))=p_{\ell'}$. If $p_{\ell'}$ performs one additional expansion iteration then the cluster $C_i$ will be assigned to $p_{\ell'}$. By the memorylessness of geometric distributions, this occurs with probability $\frac{1}{2}$. It follows that indeed $C_i$ is clustered in round $j$ with probability at least $\frac{1}{2}$.	

While the rounds are not necessarily independent, it is clear from the above argument that this bound holds for any unassigned cluster in round $j$ conditioned on any events that depend only on previous rounds. Therefore, denoting by $B_{C_i, j}$ the event that cluster $C_i$ is \textit{not} assigned in round $j$, we have
\[
    \Pr[C_i \text{ not assigned}] = \Pr\Bigg[\bigcap_{j} B_{C_i, j}\Bigg] = \prod_{j} \Pr\left[B_{C_i, j} \mid B_{C_i, 1}, \ldots, B_{C_i, j-1}\right] \leq \left(\frac{1}{2}\right)^{10\log\abs{\mc{C}}} = \frac{1}{\abs{\mc{C}}^{10}}.
\]
Taking a union bound over $\mc{C}$ yields the desired result. 
\end{proof}

\begin{lemma}
The algorithm produces an assignment with detour $\dtr_f \leq O(\log \abs{\mc{C}}) \cdot \Delta$, with high probability. 
\end{lemma}
\begin{proof}
We claim that in any round $j$, a cluster $C_i$ assigned to some $p_\ell$ in its $h$'th iteration of expansion satisfies 
\begin{equation}\label{eq:cluster-agg-induction}
  \forall v\in C_i\qquad d_{G[f^{-1}(p_{\ell})]}(v,p_{\ell})\leq d_{G}(v,P)+2\left(\sum_{j'=1}^{j-1}g_{\ell}^{(j')}+h\right)\cdot\Delta 
\end{equation}
We prove this by induction. When $j = 0$, 
only the portals are sent to themselves, and hence the distortion is $0$. 
For any $j \geq 1$, consider some cluster $C_i$ claimed by portal $p_\ell$ during round $j$ after $h$ expansion iterations.
This means there was some cluster $C_{i'}$ such that it's MID prefix $\pi'_{i'}$ intersects $C_i$, and $f(\final(\pi_{i'}'))=p_{\ell}$ (note that it is possible that $i=i'$). Let $w\in C_i\cap \pi_{i'}'$. Note that a shortest path from $w$ to $P$ follows $\pi'_{i}$. 
As $\final(\pi_{i'}')$ is already assigned to $p_\ell$ at this time, by the induction hypothesis it holds that $d_{G[f^{-1}(p_{\ell})]}(\final(\pi_{i'}'),p_{\ell})\le d_{G}(\final(\pi_{i'}'),P)+2\left(\sum_{j'=1}^{j-1}g_{\ell}^{(j')}+(h-1)\right)\cdot\Delta$. We conclude that for every vertex $v\in C_i$ it holds that (see \Cref{fig:genCA} for an illustration).
\begin{align*}
d_{G[f^{-1}(p_{\ell})]}(v,p_{\ell}) & \le d_{G[f^{-1}(p_{\ell})]}(v,w)+d_{G[f^{-1}(p_{\ell})]}(w,\final(\pi_{i'}'))+d_{G[f^{-1}(p_{\ell})]}(\final(\pi_{i'}'),p_{\ell})\\
 & \le d_{G[f^{-1}(p_{\ell})]}(v,w)+d_{G}(w,\final(\pi_{i'}'))+d_{G}(\final(\pi_{i'}'),P)+2\left(\sum_{j'=1}^{j-1}g_{\ell}^{(j')}+(h-1)\right)\cdot\Delta\\
 & =d_{G[f^{-1}(p_{\ell})]}(v,w)+d_{G}(w,P)+2\left(\sum_{j'=1}^{j-1}g_{\ell}^{(j')}+(h-1)\right)\cdot\Delta\\
 & \le d_{G}(v,P)+2d_{G[f^{-1}(p_{\ell})]}(v,w)+2\left(\sum_{j'=1}^{j-1}g_{\ell}^{(j')}+(h-1)\right)\text{}\cdot\Delta\\
 & \le d_{G}(v,P)+2\left(\sum_{j'=1}^{j-1}g_{\ell}^{(j')}+h\right)\cdot\Delta~.
\end{align*}

\begin{figure}[h]
    \centering
        \centering        \includegraphics[width=0.5\columnwidth,trim=0mm 140mm 0mm 0mm, clip]{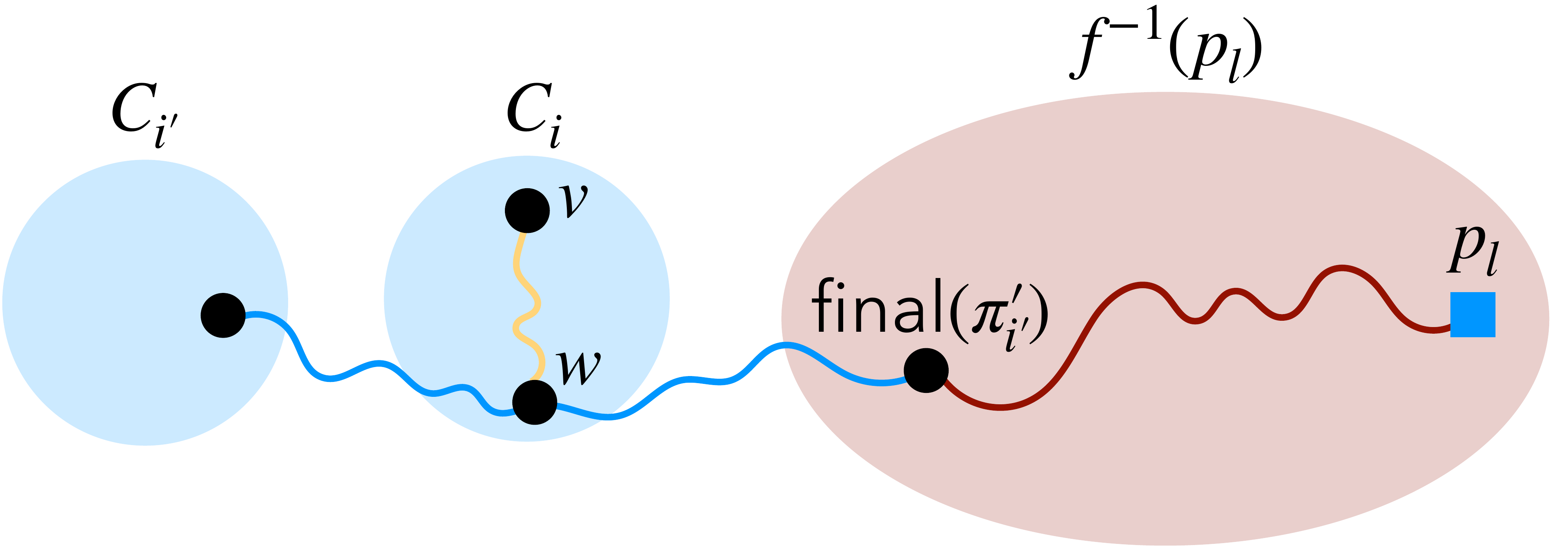}
    \caption{Paths involved in bounding $d_{G[f^{-1}(p_{\ell})]}(v,p_{\ell})$. Blue path $\pi'_{i'}$ is the MID prefix that caused $C_i$ to be assigned to $p_\ell$; yellow path is a shortest path in $C_i$, so has length at most $\Delta$; red path is a shortest path within coarsened cluster $f^{-1}(p_{\ell})$, so its length is bounded by the induction hypothesis.}\label{fig:genCA}
\end{figure}

With the claim proved, we get that at the end of the algorithm every $p_\ell$ and $v\in f^{-1}(p_\ell)$, satisfies $d_{G[f^{-1}(p_\ell)]}(v,p_{\ell})\leq d_{G}(v,P)+2\cdot(\sum_{j'=1}^{10\log\abs{\mc{C}}}g_{\ell}^{(j')})\cdot\Delta$.
Let $X \sim \text{Bin}(40\log\abs{\mc{C}}, \frac{1}{2})$, and observe that the probability of needing more than $40\log\abs{\mc{C}}$ coin tosses to get $10\log\abs{\mc{C}}$ tails is equal to the probability that \textit{exactly} $40\log\abs{\mc{C}}$ coin tosses results in \textit{fewer than} $10\log\abs{\mc{C}}$ tails. That is, we have for the sum of IID geometric random variables
\[
\Pr\left[\sum_{j'=1}^{10\log|{\mc C}|}g_{\ell}^{(j')}>40\log|{\mc C}|\right]=\Pr\left[X<10\log|{\mc C}|\right]=\Pr\left[X<\frac{1}{2}
\mathbb{E}
\left[X\right]\right]
<e^{-\mathbb{E}
    \left[X\right]/8}\leq\frac{1}{|{\mc C}|^{2}}
\]
by a standard Chernoff bound. 
Then, a union bound over $\mc{C}$ shows that the algorithm produces an assignment with detour $\leq 80\log \abs{\mc C} \cdot \Delta$.
Finally, letting $B_{U}$ denote the event that there is an unassigned cluster, and $B_R$ the event that some assigned node has detour $> 80\log \abs{\mc C}\cdot \Delta$, we have by \cref{lem:assignment-whp} and a union bound:
 \[
\Pr[\dtr_{f}>80\log\abs{\mc C}\cdot\Delta]\leq\Pr[B_{U}]+\Pr[B_{R}]\leq\frac{1}{\abs{\mc{C}}^{9}}+\frac{1}{\abs{\mc{C}}}\leq\frac{2}{\abs{\mc{C}}}~.\qedhere
\]
\end{proof}
\noindent As $\kappa = \abs{\mc C}$, the theorem follows. 
\end{proof}

\subsection{Cluster Aggregation in Trees}\label{sec:CAtree}
In this section we give our $O(1)$-distortion cluster aggregation solutions for trees.
\begin{restatable}{theorem}{CATrees}\label{thm:caTree}
	Every instance of cluster aggregation on a tree has a $4$-distortion solution that can be computed in polynomial time.
\end{restatable}
The fact that trees admit constant distortion solutions is not trivial at first glance. Indeed, related structures such as strong sparse partitions in trees exhibit $\tilde{\Omega}(\log n)$ lower bounds \cite{filtser20}. 
The basic idea is to exploit the fact that all simple paths in a rooted tree have an ``up'' part and then a ``down'' part which allows us to limit the number of conflicts when trying to coarsen our partition.

We begin by describing our algorithm. Consider an instance of cluster aggregation on $G = (V,E, w)$ with portals $P \subseteq V$ and initial partition $\mcC$ of diameter $\Delta$ where $G$ is a tree.  
For simplicity we will assume that all portals are leafs, and contained in a singleton clusters.
A general instance can be reduced to this case by simply adding an auxiliary portal $p'$ for each portal $p\in P$ with a single edge $(p,p')$ of weight $0$, and the singleton clusters $\{p'\}_{p\in P}$ to $\mcC$. It is easy to see that a solution to the auxiliary case immediately implies a solution to the original tree of the same distortion. 
 
Root $G$ arbitrarily at root $r \in V$. As we may assume that each cluster of $\mcC$ is connected,\footnote{Since otherwise distortion $0$ cluster aggregation is trivial.} each cluster $C_i \in \mcC$ has a unique vertex $r_i$ which is closest to $r$.
Let $p_i$ be the closest portal to $r_i$, and let $\pi_i$ be the shortest path in $G$ from $r_i$ to $p_i$. 
For simplicity of presentation we assume that $p_i$ is unique (in general, consistent tie breaking suffices). 

We partition the clusters $\mcC$ into two sets:
\begin{itemize}
    \item \textbf{Monotone:} Say that $C_i$ is 
    \emph{monotone} if $p_i$ is a descendant of $r_i$ in $G$. Note that it is impossible for $p_i$ to be an ancestor of $r_i$ (as $p_i$ is a leaf).
    \item \textbf{Bitone:} Say that $C_i$ is \emph{bitone} if $p_i$ is not a descendant of $r_i$ in $G$. For a bitone cluster, we let $t_i$ be the unique vertex in $\pi_i$ where $\pi_i$ switches from ``going up'' to ``going down''; that is, $t_i$ is an ancestor of all vertices of $\pi_i$ in $G$.
\end{itemize}
\noindent We also say that a cluster $C_i$ is \emph{reflective} if there is a bitone cluster $C_j$ such that $C_i$ contains $t_j$. 

We now construct our cluster aggregation solution $\mcC'$ as follows. See \Cref{fig:treeAlg} for an illustration.

\begin{enumerate}
    \item \textbf{Assign Monotone Clusters:} For each monotone cluster $C_i$, assign $C_i$ to its preferred portal; that is,  $f(C_i) = p_i$ for each monotone $C_i$. Note that each portal is assigned to itself.
    \item \textbf{Assign Reflective Bitone Clusters:} For each reflective bitone cluster $C_i$ containing a $t_j$ we assign $C_i$ to $C_j$'s preferred portal, namely $f(C_i) = p_j$. If there are multiple such $C_j$ we pick an arbitrary one. 
    \item  \textbf{Assign Non-Reflective Bitone Clusters:} For each so far unassigned bitone cluster $C_i \in \mcC$, let $C_j$ be the unique cluster containing $t_i$. In this case, $C_i$ joins the cluster to which $C_j$ was already assigned; that is, we let $f(C_i) = f(C_j)$.
\end{enumerate}

\begin{figure}[h]
    \centering
    \begin{subfigure}[b]{0.48\textwidth}
        \centering
        \includegraphics[width=\textwidth,trim=0mm 0mm 0mm 0mm, clip]{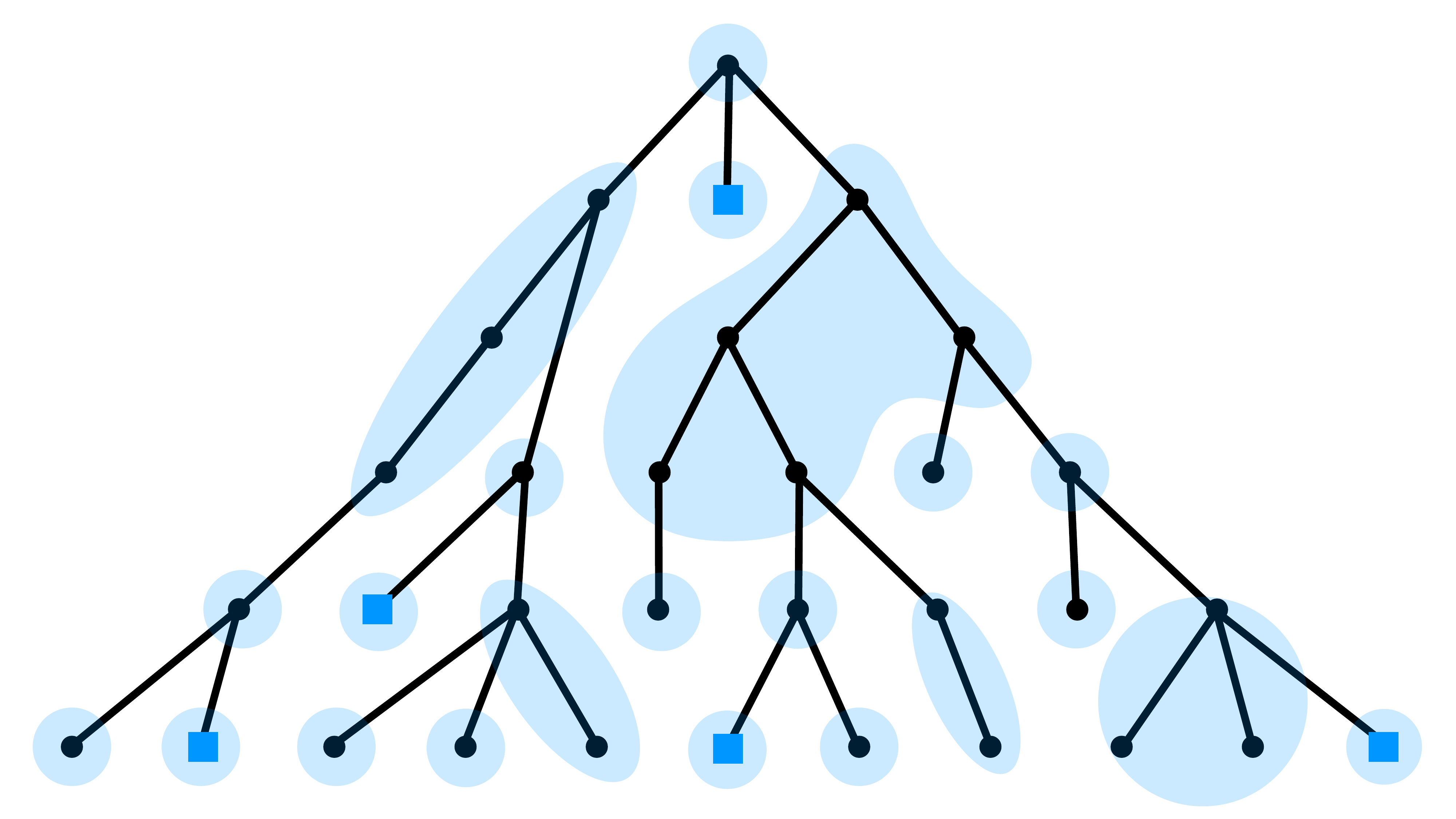}
        \caption{Input clusters and portals.}\label{sfig:treeAlg1}
    \end{subfigure}    \hfill
    \begin{subfigure}[b]{0.48\textwidth}
        \centering
        \includegraphics[width=\textwidth,trim=0mm 0mm 0mm 0mm, clip]{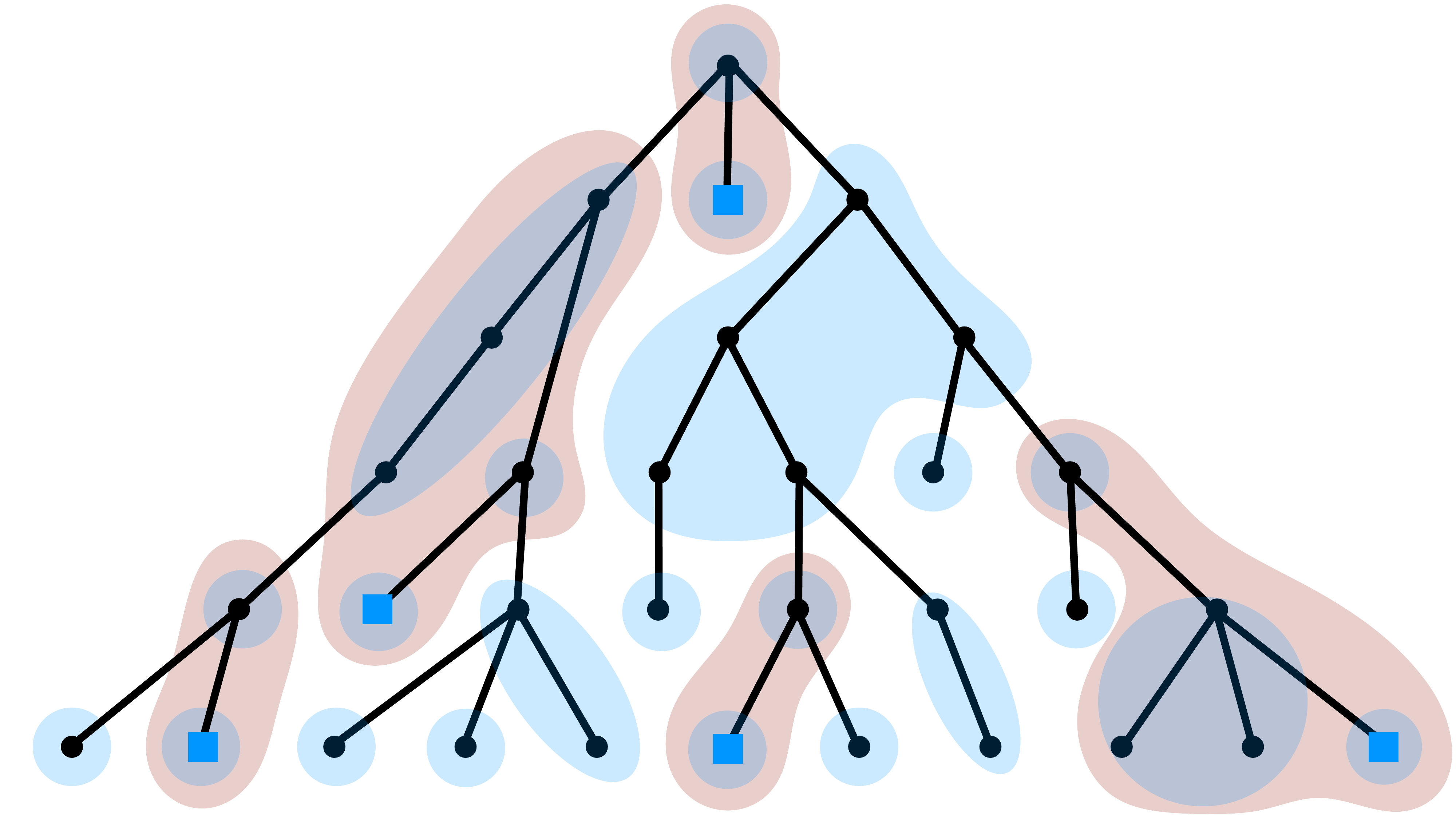}
        \caption{Assigning monotone clusters.}\label{sfig:treeAlg2}
    \end{subfigure}    \hfill
    \begin{subfigure}[b]{0.48\textwidth}
        \centering
        \includegraphics[width=\textwidth,trim=0mm 0mm 0mm 0mm, clip]{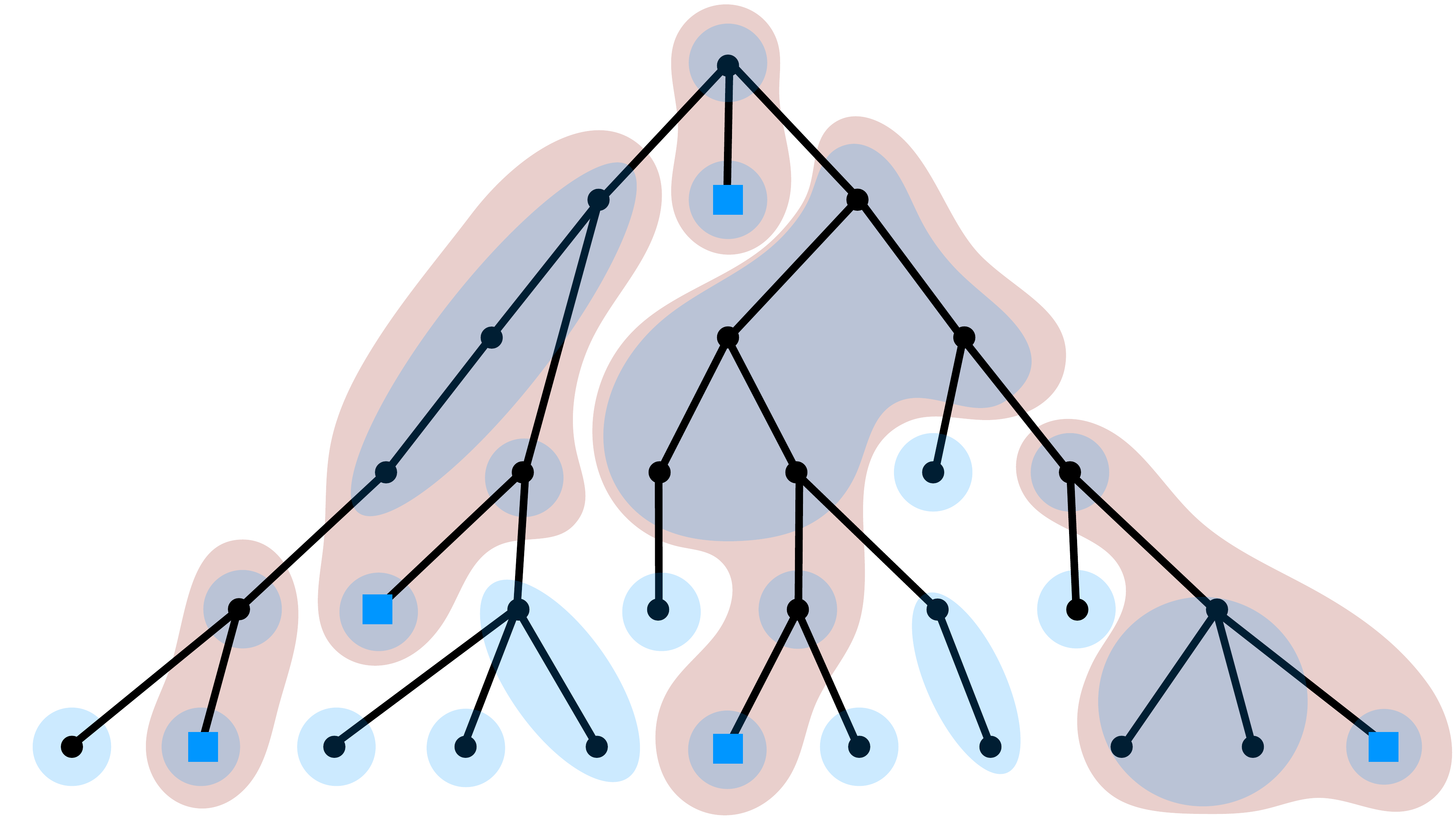}
        \caption{Assigning reflective bitone clusters.}\label{sfig:treeAlg3}
    \end{subfigure}
    \begin{subfigure}[b]{0.48\textwidth}
        \centering
        \includegraphics[width=\textwidth,trim=0mm 0mm 0mm 0mm, clip]{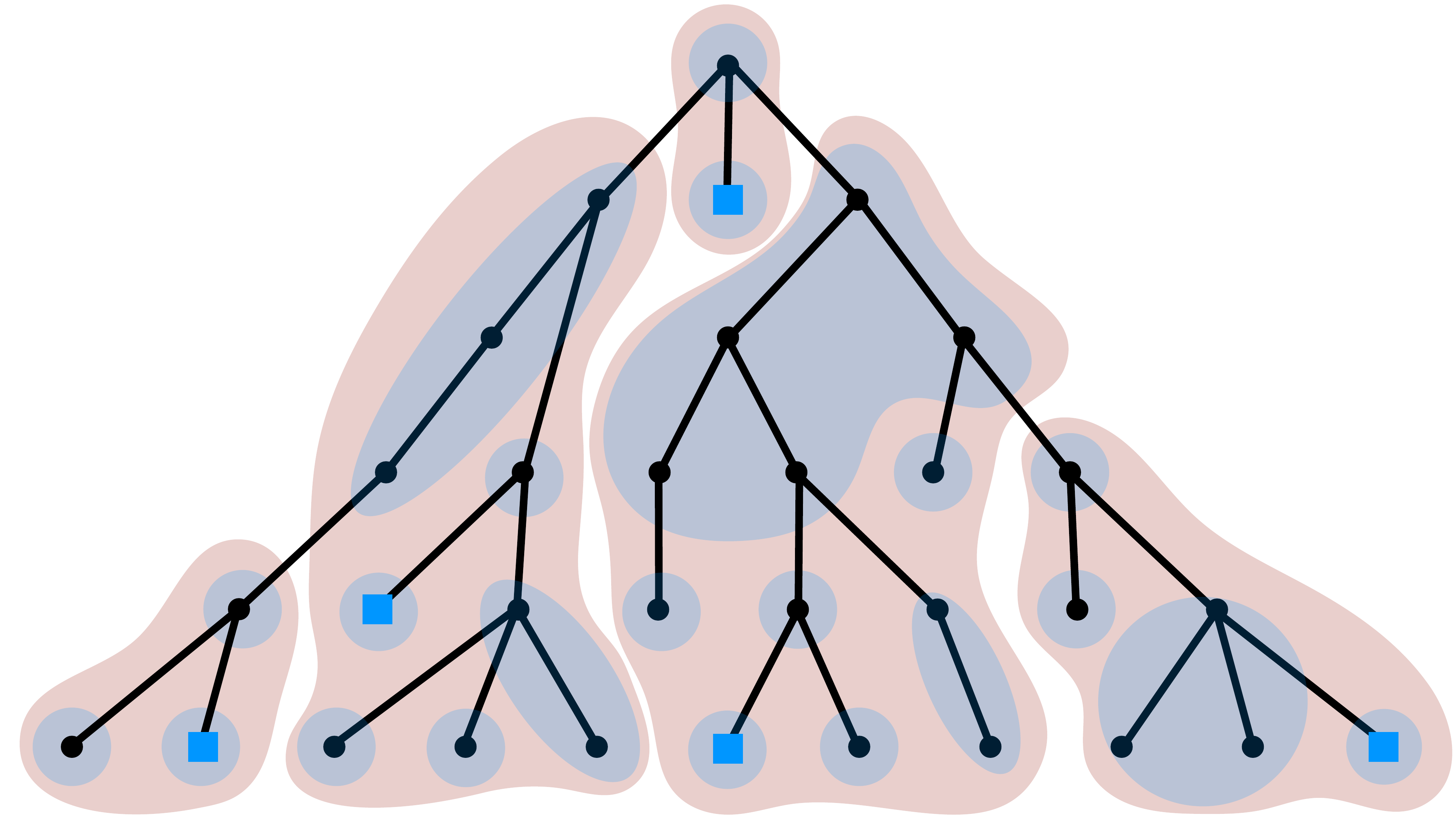}
        \caption{Assigning other bitone clusters (output).}\label{sfig:treeAlg4}
    \end{subfigure}
    \caption{Our cluster aggregation algorithm on a unit-weight tree. The input partition $\mcC$ is blue and our output partition $\mcC'$ is red. \ref{sfig:treeAlg1} gives our input. \ref{sfig:treeAlg2} is the result from processing all monotone clusters. \ref{sfig:treeAlg3} is the result of processing all reflective bitone clusters. \ref{sfig:treeAlg4} is the result of processing all remaining bitone clusters.}\label{fig:treeAlg}
\end{figure}

The following summarizes the properties of our algorithm. Note, also, that our analysis is tight for our algorithm as shown by \Cref{fig:TreeCAAnalysisTight}.
\begin{proof}[Proof of \Cref{thm:caTree}]
We use the algorithm described above. The algorithm trivially runs in polynomial time. We begin by observing that this process is indeed a cluster aggregation solution in that $f(C_i)$ is a portal for every $C_i \in \mcC$. In particular, observe that each trivial, monotone and reflective bitone cluster is trivially assigned. The only remaining clusters of $\mcC$ are those that are bitone but not reflective but any such $C_i$ must be such that the cluster containing $t_i$ is either monotone or is a reflective bitone cluster and so already assigned.

Next, we analyze the distortion of the resulting cluster aggregation solution. Recall the definition of the detour of $v$: 
\begin{align*}
        \dtr_f(v) := d_{G[f^{-1}(f(v))]}(v, f(v)) - d_G(v, P).
\end{align*}
Our goal is to show $\dtr_f(v) \leq 4\Delta$ for every $v$.
The intuition is as follows. Roots of monotone clusters have detour value of $0$ by construction. 
On the other hand, non-root vertices of monotone clusters must traverse their clusters, potentially paying $\Delta$ for the distance traveled and another $\Delta$ for potentially increasing their distance to $P$ by $\Delta$ for a cumulative detour of $2\Delta$. Vertices in reflective bitone cluster $C_i$ contain some $t_j$ which is sent to its closest portal (i.e. $\dtr_f(t_j)=0$). By the same argument as for monotone clusters, all other vertices in $C_i$ will have detour value of at most $2\Delta$.
Lastly, vertices in non-reflective bitone clusters must pay $2 \Delta$ to eventually get to a reflective bitone cluster or monotone cluster at which point they pay the $2 \Delta$ for which this cluster already paid for $4 \Delta$ total.

We now make these bounds formal. For the remainder of this proof we consider a fixed vertex $v\in V$ where $v \in C_i$. We take cases on what type of cluster $C_i$ is.
Clearly, for every $v\in P$, we have $\dtr_f(v) = 0$.

Suppose $C_i$ is monotone. That is $p_i$ is a descendant of $r_i$. Every cluster $C_j$ that $\pi_i$ intersects is also monotone, with $p_j=p_i$.
It follows that to get to $p_i$ in $G[f^{-1}(p_i)]$, $v$ may have to traverse its own cluster until it reaches $r_i$ at which point it can follow $\pi_i$ to $p_i$. We conclude:
\begin{align*}
	d_{G[f^{-1}(p_{i})]}(v,p_{i}) & \le d_{G[f^{-1}(p_{i})]}(v,r_{i})+d_{G[f^{-1}(p_{i})]}(r_{i},p_{i})\\
	& =d_{G}(v,r_{i})+d_{G}(r_{i},P)\\
	& \le d_{G}(v,P)+2\cdot d_{G}(v,r_{i})\\
        &\le d_{G}(v,P)+2\Delta~.
\end{align*}


Next, suppose $C_i$ is a reflective bitone cluster. 
We claim that $\dtr_f(v) \leq 2\Delta$. Specifically, there is a bitone cluster $C_j$ which is descendant of $r_i$, such that $f(C_i)=p_j$, and $t_j\in C_i$ is the ancestor of all the vertices in $\pi_j$. Note that all the clusters on the path from $t_j$ to $p_j$ are monotone clusters sent to $p_j$. In particular, $d_{G[f^{-1}(p_{j})]}(t_j,p_{j})=d_{G}(t_j,p_{j})=d_{G}(t_j,P)$. As $v$ and $t_j$ belong to the same cluster we conclude:
\begin{align*}
	d_{G[f^{-1}(p_{j})]}(v,p_{j}) & \le d_{G[f^{-1}(p_{j})]}(v,t_{j})+d_{G[f^{-1}(p_{j})]}(t_{j},p_{j})\\
	& =d_{G}(v,t_{j})+d_{G}(t_{j},P)\\
	& \le d_{G}(v,P)+2\cdot d_{G}(v,t_{j})\\
        &\le d_{G}(v,P)+2\Delta~.
\end{align*}
%
%
%

 Lastly, suppose $C_i$ is a bitone cluster that is not reflective. We claim $\dtr_f(v) \leq 4 \Delta$.
Let $C_j$ be the reflective cluster containing $t_i$. Suppose that $f(C_j)=p_k$. Note that all the vertices on the path from $r_i$ to $t_i$ belong to bitone clusters with $C_j$ being their corresponding reflective cluster. In particular, all these vertices are assigned to $p_k$. It follows that $d_{G[f^{-1}(p_{k})]}(r_i,t_{i})=d_{G}(r_i,t_{i})$. As the shortest path from $r_i$ to $P$ goes through $t_i$, and we can bound $d_{G[f^{-1}(p_{k})]}(t_i,p_{k})$ using the previous two cases, we conclude
\begin{align*}
	d_{G[f^{-1}(p_{k})]}(v,p_{k}) & \le d_{G[f^{-1}(p_{k})]}(v,r_{i})+d_{G[f^{-1}(p_{k})]}(r_{i},t_{i})+d_{G[f^{-1}(p_{k})]}(t_{i},p_{k})\\
	& \le d_{G}(v,r_{i})+d_{G}(r_{i},t_{i})+d_{G}(t_{i},P)+2\Delta\\
	& =d_{G}(v,r_{i})+d_{G}(r_{i},P)+2\Delta\\
	& \le d_{G}(v,P)+2\cdot d_{G}(v,r_{i})+2\Delta\\
        &\le d_{G}(v,P)+4\Delta~.
\end{align*}
The theorem now follows.
%
\end{proof}

\begin{figure}[h] 
\includegraphics[width=.4\textwidth,trim=0mm 0mm 175mm 0mm, clip]{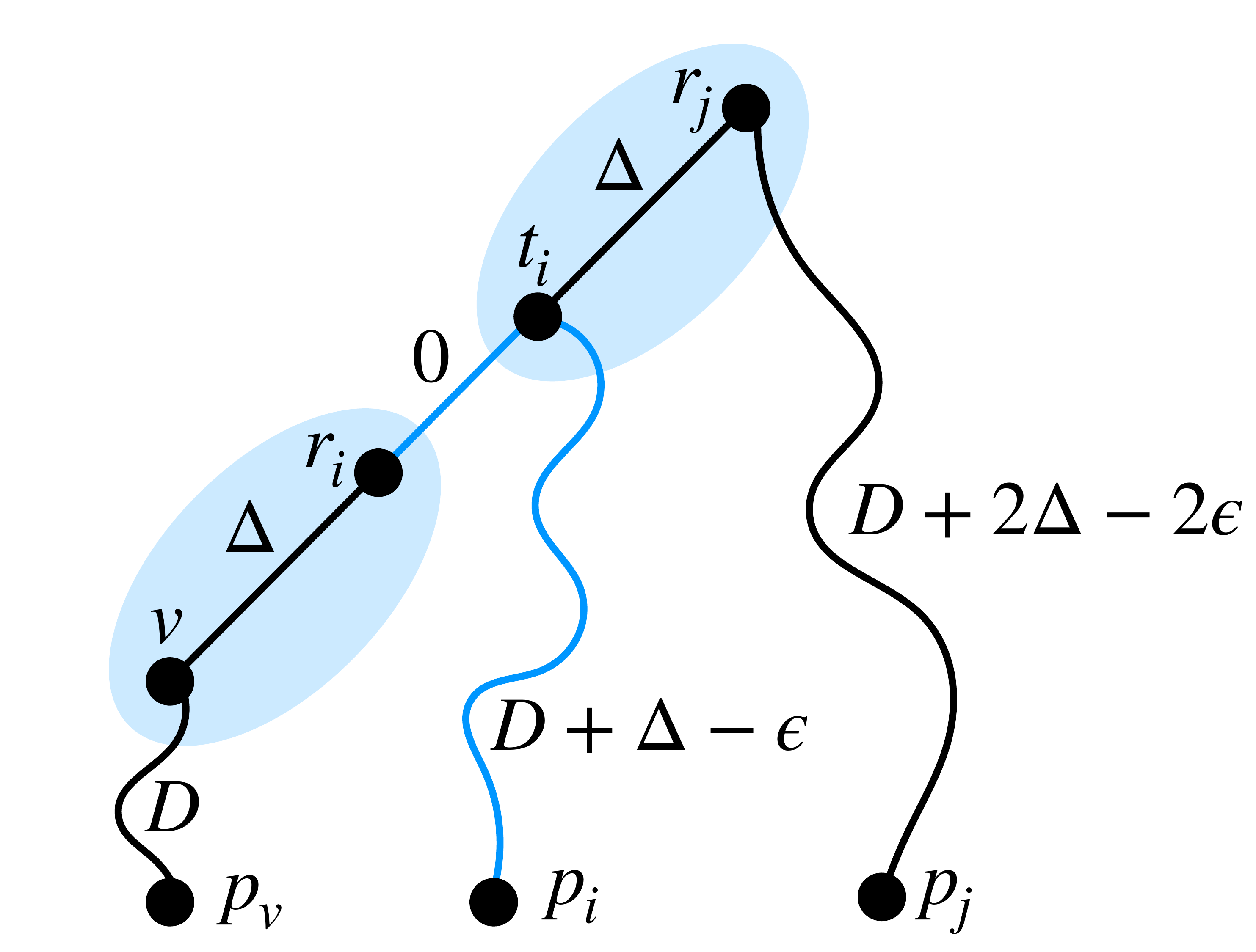}
\caption{Example showing tightness of the analysis of our algorithm for cluster aggregation on trees. The tree is rooted at $r_j$ which belongs to monotone cluster $C_j$ where the closest portal $p_j$ is at distance $D+2\Delta-2\epsilon$. $C_i$ is a bitone cluster, where the shortest path $\pi_i$ (colored in blue) from $r_i$ to $P$ is of length $D+\Delta-\epsilon$, and goes through $t_i\in C_j$. The algorithm sends all the vertices of the cluster $C_i$ and $C_j$ to $p_j$. The vertex $v\in C_i$ has a portal at distance $D$, but is sent to $p_j$ at distance $D+4\Delta-2\epsilon$. As $\epsilon$ goes to $0$, the additive distortion of $v$ goes to $4\Delta$.}\label{fig:TreeCAAnalysisTight}
\end{figure}

\subsection{Cluster Aggregation in Bounded Doubling Dimension Graphs}\label{sec:CAdoubling}

\newcommand{\junk}[1]{}
In this section, we prove the following improved bound for cluster
aggregation in graphs with bounded doubling dimension.

\begin{restatable}{theorem}{CADD}\label{thm:caDD}
	For every large enough constant $\lambda$, there is a constant $c_s$
	such that any instance $\calC$ of cluster aggregation on a graph $G$ of doubling dimension $\ddim$ with portals $P$ and an input partition of strong diameter $\Delta$ where:
    \begin{enumerate}
		    \item $P$ is a $\Lambda$-net for $\Lambda =\lambda\cdot\ddim^3\log\ddim \cdot \Delta$. That is, $\max_{v\in V}d_G(v,P)\le\Lambda$ and the minimum pairwise distance between portals is at least $\min_{p,p'\in P}d_G(p,p')\ge c_{\Lambda}\cdot\Lambda$. (Here, $c_\Lambda$ is the constant $c_\Delta$ from \Cref{thm:MPXbasedClusteringDoubling}); and
	    \item Every cluster $C\in\calC$ has a center $v_C$ such that the minimum pairwise distance between centers is $\min_{C,C'\in\calC}d_G(v_C,v_{C'})\ge \frac{c_\Lambda}{3}\cdot\Delta$;
 
	\end{enumerate}
 has a poly-time $4c_{s}\cdot\ln\frac{4}{c_{\Lambda}}\cdot\ddim^{2}\log
	\ddim\cdot \Delta$-distortion solution.
\end{restatable}

\junk{
\begin{theorem}\label{thm:CAdoubling}
	Consider a graph $G=(V,E,w)$ with doubling dimension $d$, and let
	${\cal C}$ be a partition into clusters with strong diameter at most
	$\Delta$, and let $P$ be a set of portals such that:
	\begin{enumerate}
		\item The portal $P$ constitute a $\Lambda=c_{\Lambda}\cdot d^{2}\log d\cdot \Delta$
		net for some large enough constant $c_{\Lambda}$.
		\item Every cluster $C\in{\cal C}$ contains a center $v_{c}\in C$, and
		the minimal pairwise distance between two centers is $\Omega(\frac{\Delta}{d})$.
	\end{enumerate}
	Then there is a solution to the cluster aggregation problem with stretch
	$O(d^{2}\log d)\cdot \Delta$.
\end{theorem}
}
 
%
%


\subsubsection{Preliminary: Clustering Using Exponential Shifts (\cite{MillerPX2013})}
We will be using a modified version of the clustering algorithm by Miller \etal \cite{MillerPX2013} due to \cite{Fil19Padded,Fil20}.
%
Consider a weighted graph $G=(V,E,w)$. Let $N\subseteq V$ be some set of centers, where each $t\in N$ has a parameter $\delta_t$. 
For each vertex $v$ set a function $f_v:N\rightarrow \mathbb{R}$ as follows: for a center $t$, $f_v(t)=\delta_t-d_G(t,v)$. The MPX algorithm has the vertex $v$ join the cluster $C_t$ of the center $t\in N$ maximizing $f_v(t)$.  
Note that apriori (and in some applications even necessarily),  a center $t\in N$ might join the cluster of a different center $t'\in N$.
An intuitive way to think about the clustering process is as follows: each center $t$ wakes up at time $-\delta_t$ and begins to ``spread'' in a continuous manner. The spread of all centers is performed simultaneously. A vertex $v$ joins the cluster of the first center that reaches it. 

The partition created by the MPX algorithm will not necessarily agree with $\calC$ in the sense that an MPX clusters may not fully contain every cluster of $\calC$ that it intersects. Our algorithm will therefore work as follows: we will sample shifts (i.e.\ $\delta_t$ values) to create an MPX partition. We fix and contract areas of the MPX partition which are ``consistent'' with respect to $\calC$ and then recurse on the resulting graph.

\subsubsection{The Doubling Dimension Cluster Aggregation Algorithm}
\label{sec:ca_dd_algorithm}
The cluster aggregation algorithm is described in \Cref{alg:cluster-agg-doubling}.
The cluster aggregation algorithm runs in a sequence of rounds.  In
each round, a subset of the clusters get assigned to portals using a
slight modification of the clustering algorithm of~\cite{MillerPX2013}
due to~\cite{filtser20}.  All clusters that are assigned to a given portal
are then contracted to yield a new graph partitioned into unassigned clusters.  This process continues for a number
of rounds, which is large enough to guarantee that all clusters are
aggregated to their portals as desired.

More formally, our algorithm is as follows; see, also, \Cref{alg:dd}. Fix $c_\top=\ln\left\lceil \frac{4}{c_{\Lambda}}\right\rceil +3$.
We will sample the shifts in the algorithm using a \emph{betailed exponential distribution} with parameters $(1,c_\top)$ (see \cite{Fil20}). 
Specifically, in each round $i$, for each portal $p\in P$ we first sample $\tilde{\delta}_p^{(i)}$ according to exponential distribution with parameter $1$ (the distribution with density function $e^{-x}$ for $x\ge0$), and then the shift is $\delta_p^{(i)}=\min\{\tilde{\delta}_p^{(i)},c_\top\cdot d\}\cdot\Delta$.
%
For every vertex $v\in V$, these shifts induce a function from the portals: $g^{(i)}_v(p) = \delta_p^{(i)} - d_{G_i}(v,p)$.
A cluster $C$ is said to \emph{choose} $p$ if for every $u\in
C$, $g^{(i)}_{u}(p)$ is the maximum of $g^{(i)}_{u}$. A cluster $C$ is said to
be \emph{satisfied by $p$} if there is a shortest path $Q_{u,p}$
from some vertex $u\in C$ to $p$ such that all the clusters $C'$
intersecting $Q_{u,p}$ choose $p$.
For every cluster $C$ satisfied by $p$, the cluster aggregation assignment of $C$ will be $p$.
For each portal $p$, let $A_{p}^{(i)}$ be the cluster formed by the
union of all the clusters satisfied by $p$.		
We construct a graph $G_{i+1}$ from $G_i$ by contracting all the sets $\{A_{p}^{(i)}\}_{p\in P}$. Specifically, each set 
$A_p^{(i)}$ is contracted into a single vertex, which we relabel as portal $p$. If there is an edge in $G_i$ from $v$ to a vertex in
$A_p^{(i)}$, then add an edge from $v$ to $p$ in $G_{i+1}$ with
weight $d_{G[A_{p}^{(i)}\cup\{v\}]}(p,v)$.

\begin{algorithm}[h]\label{alg:dd}
	\caption{\texttt{Cluster aggregation for general graphs}}
	\label{alg:cluster-agg-doubling}
	\DontPrintSemicolon
	\SetKwInOut{Input}{input}\SetKwInOut{Output}{output}
	\Input{Weighted graph $G = (V, E,w)$ with doubling dimension $\ddim$, partition $\mc{C} = \{C_i\}_i$ into clusters of strong diameter at most $\Delta$ with centers $\{v_C\}_{C\in\calC}$ at minimal pairwise distance $\Omega(\Delta)$, portal set $P \subseteq V$ which is a $\Lambda=\Theta(\ddim^2\cdot\log\ddim\cdot\Delta)$-net.}
	\Output{Assignment $f: \mc{C} \to P$ of additive distortion $\beta = O(\ddim^2\cdot\log\ddim)$.}
	\BlankLine
	
	Set $V_0 = V$, $E_0 = E$, $G_0 = G = (V, E)$, $f=\emptyset$, and $s = c_s\cdot d \log d$.\;
	\For{$i$ from 1 to $s$}{
		For every portal $p\in P$, sample $\delta_{p}^{(i)}$ according to a betailed exponential distribution with parameters $(1,c_\top)$. \hfill
		\tcp*[h]{Here $c_{\top}=\ln\left\lceil \frac{4}{c_{\Lambda}}\right\rceil +3$}\;
		
		
		
		For each portal $p$, for each cluster $C$ satisfied (w.r.t. $g^{(i)}$) by $p$, set
		$f(C) = p$.\;
		
		For each portal $p$, let $A_{p}^{(i)}$ be the cluster formed by the
		union of all the clusters satisfied by $p$.		
		Construct $G_{i+1}$ from $G_i$
		as follows: (a) contract each
		$A_p^{(i)}$ into a single vertex, which we relabel as portal $p$;
		(b) if there is an edge in $G_i$ from $v$ to a vertex in
		$A_p^{(i)}$, then add an edge from $v$ to $p$ in $G_{i+1}$ with
		weight $d_{G[A_{p}^{(i)}\cup\{v\}]}(p,v)$\label{line:ddimCreageGi+1}\;
	}
	\Return $f$\;
\end{algorithm}

%
%

\junk{
Following {[}Fil20Scattering{]} (and using the notation from there),
we run MPX with the portals being centers. Every portal $p\in P$
samples a shift $\delta_{v}^{(1)}$ according to \textbf{betailed}
exponential distribution. That is first sample
$\tilde{\delta}_{v}^{(1)}$ according to exponential distribution with
parameter $1$. Set $\delta_{v}^{(1)}=\min\left\{
\tilde{\delta}_{v}^{(1)},c_{\top}\cdot d\right\} \cdot D$, for a
constant $c_{T}$ to be determined later.

Each vertex $v\in V$ has a function from portals to reals $g_{v}(p)=\delta_{p}-d_{G}(v,p)$.

Note that as the distance between
portals is greater than the maximal shift, no portal can ever get
clustered by a different portal. 

This is a single iteration of our algorithm.

Each cluster at the end of a single iteration is denoted
$A_{p}^{(1)}$. We then contract each $A_{p}^{(1)}$ into a single
vertex still called $p$ (as the portal), the newly created graph is
called $G_{1}$. Denote also $G_{0}=G$.  Note that $G_{1}$ contains an
edge from $v$ to $A_{p}^{(1)}$ if there is an edge in $G_{0}$ from $v$
to a vertex in $A_{p}^{(1)}$.  We set the weight of this edge to be
$d_{G[A_{p}^{(1)}\cup\{v\}]}(p,v)$.

We then repeat the process in $G_{1}$ to obtain samples $\left\{ \delta_{p}^{(2)}\right\} _{p\in P}$
, which defines functions $\left\{ g_{v}^{(2)}\right\} _{v\in G_{2}}$and
clusters $A_{p}^{(2)}$. We then contract them in the same way to
obtain $G_{2}$. Note that all pairwise distances in $G_{2}$ can
only grow. We continue in this manner until every vertex is clustered.
}

\subsubsection{Analysis of the Distance to Assigned Portal}
Throughout the execution of the algorithm, every edge of $G_i$ has a weight equal to some path in the original graph $G$. In particular, pairwise distances can never decrease.  
As the maximum possible shift is $c_\top\cdot\Delta$, while the minimum pairwise distance between portals is at least $\Lambda> c_\top\cdot\Delta$, it holds that a portal $p$ will never choose another portal $p'$. It follows that $P$ is a subset of the vertices for all graphs in $\{G_i\}_{i\ge0}$.
%

This subsection is devoted to proving \Cref{lem:DistanceToClosestPortal}, which stated that if a vertex $v$ is clustered during the $j$th iteration, its ``additive distortion'' is at most $O(j\cdot c_\top\cdot d)$. As we will have only $\tilde{O}(d)$ iterations, the overall distortion will be bounded by $\tilde{O}(d^2)$.

\begin{claim}
  \label{clm:DistanceForJoined}
  Fix an iteration $i$ and consider a cluster $C$ in $G_i$.  For
        every $u,v\in C$ and portal $p$,
        $\left|g^{(i)}_{u}(p)-g^{(i)}_{v}(p)\right|\le \Delta$.
\end{claim}
\begin{proof}
	The claim follows immediately from triangle
        inequality: $\left|g^{(i)}_{u}(p)-g^{(i)}_{v}(p)\right|$ equals
        $|d_{G_i}(u,p) - d_{G_i}(v,p)|$, which is at most $\Delta$.
\junk{
	For the second part of the claim, note that there is some
        vertex $u \in C$ such that all the vertices along $Q_{u,p}$ choose
        $p$. By definition,\atodo{This is not accurate. Not by definition but rather by shortest path properties} all these clusters will also be satisfied
        by $p$.  Hence $d_{G[A_{p}]}(u,p)=d_{G}(u,p)$. Clearly for
        every vertex $v\in C$, $d_{G[A_{p}]}(u,v)\le \Delta$. We conclude,
        $d_{G[A_{p}]}(v,p)\le d_{G[A_{p}]}(v,u)+d_{G[A_{p}]}(u,p)\le
        d_{G}(u,p)+\Delta\le d_{G}(v,p)+2\Delta$.}
\end{proof}

\begin{lemma}\label{lem:DistanceToClosestPortal}
	The two following properties hold
	in the graph $G_{j}$:
	\begin{enumerate}
		\item For every unclustered vertex $v$ in $G_{j}$, $d_{G_{j}}(P,v)\le d_{G}(P,v)+j\cdot(c_{\top}\cdot d+2)\cdot \Delta$.
		\item For every vertex $v\in A_{p}^{(j)}$, it holds that $d_{G[A_{p}^{(j)}]}(p,v)\le d_{G}(P,v)+j\cdot(c_{\top}\cdot d+2)\cdot \Delta$.
	\end{enumerate}
\end{lemma}

\begin{proof}
	The proof is by induction on $j$. For $j=0$, $G_{0}=G$, and hence
	for every vertex $v\in V$, $d_{G_{0}}(P,v)=d_{G}(P,v)=d_{G}(P,v)+0\cdot c_{\top}\cdot d\cdot \Delta$.
	Furthermore, every cluster $A_{p}^{(0)}$ is the singleton $\{p\}$,
	and trivially $d_{G[A_{p}^{(0)}]}(p,p)=0=d_{G}(p,p)+0\cdot c_{\top}\cdot d\cdot \Delta$.
	Assume the claim holds for $G_{j}$ and $\left\{ A_{p}^{(j)}\right\} _{p\in P}$,
	and we will prove it for $j+1$.
	
	Consider a vertex $v\in V$ in a cluster $C$, and let $p_{v}\in P$
	be the closest portal to $v$ in $G_{j}$. We begin by proving the
	second assertion. Suppose that $v$ joins the cluster $A_{p}^{(j+1)}$
	centered in a portal $p$ during the iteration $j+1$. It follows
	that there a vertex $u\in C$ with shortest path $Q_{u,p}$ from $u$
	to $p$ such that all the clusters intersecting $Q_{u,p}$ choose
	$p$. In particular $A_{p}^{(j+1)}$ contains a path from $u$ to
	$p$ of weight $d_{G_{j}}(u,p)$. We conclude
\begin{eqnarray}
	d_{G[A_{p}^{(j+1)}]}(p,v) & \le & d_{G[A_{p}^{(j+1)}]}(p,u)+d_{G[A_{p}^{(j+1)}]}(u,v)\nonumber \\
	& \le & d_{G_{j}}(p,u)+\Delta\nonumber \\
	& \le & d_{G_{j}}(p_{v},u)+c_{\top}\cdot d\cdot \Delta+\Delta\nonumber \\
	& \le & d_{G_{j}}(p_{v},v)+c_{\top}\cdot d\cdot \Delta+2\Delta \label{eq:SingleStep}\\
	& \le & d_{G}(P,v)+j\cdot(c_{\top}\cdot d+2)\cdot \Delta+(c_{\top}\cdot d+2)\cdot \Delta\nonumber \\
	& = & d_{G}(P,v)+(j+1)\cdot(c_{\top}\cdot d+2)\cdot \Delta~.\nonumber 
\end{eqnarray}
The first step follows from triangle inequality.  The second and
fourth steps follow from the diameter bound for every cluster.  The
third step follows from the upper bound on the shifts, as $g_{u}(p_{v})\le g_{u}(p)$
implies $d_{G_{j}}(p,u)\le d_{G_{j}}(p_{v},u)+\delta_{p}\le d_{G_{j}}(p_{v},u)+c_{\top}\cdot d\cdot \Delta$. The fifth step follows from the induction hypothesis.

	We now move to prove the first assertion. Let $Q_{v,p_{v}}=\left(v=x_{0},x_{1},\dots,x_{m}=p_{v}\right)$
	be the shortest path from $v$ to $p_{v}$ in $G_{j}$. Note that
	the only portal in $Q_{v,p_{v}}$ is
	$x_{m}=p_{v}$ (otherwise $p_v$ would not be the closest portal to $v$ in $G_j$). If no vertex
	in $Q_{v,p_{v}}$ joined any cluster (other than that of $p_{v}$), then
	$d_{G_{j+1}}(P,v)\le d_{G_{j+1}}(p_{v},v)=d_{G_{j}}(p_{v},v)=d_{G_{j}}(P,v)$,
	and we are done by induction. Otherwise, let $x_{q}$ be the vertex
	with minimal index that joined some cluster centered at a portal $p\in P$.  By the definition of $G_{j+1}$, there is an edge from
	$x_{q-1}$ to $p$ (as there is an edge from a vertex $x_{q}\in A_{p}^{(j+1)}$
	towards $x_{q-1}$). We conclude:
	\begin{align*}
	  d_{G_{j+1}}(P,v) & \le d_{G_{j+1}}(p,v)\\
                & \le d_{G_{j+1}}(v,x_{q-1})+d_{G_{j+1}}(x_{q-1},p)\\
		& \le d_{G_{j}}(v,x_{q-1})+d_{G[A_{p}^{(j+1)}\cup\{x_{q-1}\}]}(x_{q-1},p)\\
		& \le d_{G_{j}}(v,x_{q-1})+w_{G_{j}}(x_{q-1},x_{q})+d_{G[A_{p}^{(j+1)}]}(x_{q},p)\\
		& \le d_{G_{j}}(v,x_{q})+d_{G_{j}}(x_{q},p_{v})+(c_{\top}\cdot d+2)\cdot \Delta\\
		& =d_{G_{j}}(v,p_{v})+(c_{\top}\cdot d+2)\cdot \Delta\\
        & \le d_{G}(v,P)+(j+1)\cdot(c_{\top}\cdot d+2)\cdot \Delta\,.
	\end{align*}
        The second,
        third and fourth steps follow from triangle inequality.  The
        fifth step follows from Equation~\ref{eq:SingleStep}.  The
        sixth step follows from the definition of $x_q$ while the last
        step follows from the induction hypothesis.
\end{proof}

\subsubsection{Analysis of the Probability of Being Clustered}
In this subsection we analyze the clustering probability from the perspective of a single cluster. From \Cref{lem:clusterProb} it follow that the probability a cluster is unassigned by the end of the algorithm is $2^{-O(s)}$.
\begin{lemma}
  \label{lem:prob_cluster}
	Consider an iteration $j$, and an unclustered vertex $v\in C$, let
	$p_{_{v}}$ be the portal maximizing $g_{v}^{(j)}$. If $g_{v}^{(j)}(p_{v})\ge2\Delta+\max_{p\ne p_{v}}g_{v}^{(j)}(p)$,
	then $C$ is satisfied by $p_{v}$.
\end{lemma}

\begin{proof}
	The iteration $j$ is executed on the graph $G_{j-1}$. For simplicity
	we will assume $j=1$, as the other cases are the same. Let $Q_{v,p_{v}}$
	be the shortest path from $v$ to $p_{v}$ in $G$. Consider a cluster
	$C'$ intersecting $Q_{v,p_{v}}$, and let $z\in C'\cap Q_{v,p_{v}}$.
	We argue that $C'$ chooses $p_{v}$. Note that this implies that
	$C$ is satisfied by $p_{v}$. Consider a vertex $u\in C'$ and a
	portal $p'$. Then
	\begin{align*}
	  g_{u}(p_{v}) & \ge g_{z}(p_{v})-\Delta\\
          & =\delta_{p_{v}}-d_{G}(z,p_{v})-\Delta\\
          &=\delta_{p_{v}}-d_{G}(v,p_{v})+d_{G}(v,z)-\Delta\\
		& =g_{v}(p_{v})+d_{G}(v,z)-\Delta\\
		&\ge g_{v}(p')+d_{G}(v,z)+\Delta\\
		& =\delta_{p'}-d_{G}(v,p')+d_{G}(v,z)+\Delta\\
		&\ge\delta_{p'}-d_{G}(z,p')+\Delta\\
		& =g_{z}(p')+\Delta\\
		&\ge g_{u}(p')~.
	\end{align*}
        (The first and last steps follow from
        \Cref{clm:DistanceForJoined}.  The second, fourth, sixth,
        and eighth steps follow from the definition of $f$.  The third
        step follows from the choice of $z$.  The fifth step follows
        from the assumption in the lemma.  The penultimate step
        follows from triangle inequality.)  Thus, all the vertices in
        $C'$ indeed choose $p_{v}$ as required.
\end{proof}
Let $s=c_{s}\cdot d\log d$ for a constant $c_{s}$ to be determined
later.
\begin{lemma}
  \label{lem:clusterProb}
	For $j\le s$, every cluster $C$ in $G_{j}$ is satisfied by some
	portal with probability at least $\frac{e^{-2}}{2}$. In particular,
	the probability that a cluster $C$ remains unclustered after $s$
	iterations, is at most $2^{-\frac{s}{10}}$.
\end{lemma}
\begin{proof}
  Fix a vertex $v\in G_{j}$. Let $\tilde{\Lambda}$ denote the distance from $v$ to the closest portal.  Then, we have
\begin{equation}
	\tilde{\Lambda}=d_{G_{j}}(v,P)\le d_{G}(v,P)+j\cdot(c_{\top}\cdot d+2)\cdot\Delta\le\Lambda+s\cdot(c_{\top}\cdot d+2)\cdot\Delta\le2\cdot\Lambda~,\label{eq:LambdaTildeBound}
\end{equation}
  where the first inequality follows from the second part of
  \Cref{lem:DistanceToClosestPortal}, the second inequality holds
  since every vertex is within $\Lambda$ distance of a portal and there are at most $s$ rounds, and the last inequality holds for
  appropriate choice of the constants $c_s,\lambda$ (see \Cref{rem:constantsChoiseDoubling}).   
%
  We call a portal $p'$ an \emph{almost
    winner of $v$} if and only if $g_{v}^{(j)}(p')\ge\max_{p\in
    P}g_{v}^{(j)}(p)-2\Delta$ (note that a winner is also an almost
  winner).  Therefore, since the shifts are bounded by $c_{\top}\cdot
  d$, a portal $p'$ has a non-zero probability to become an almost
  winner only if $d_{G_{j}}(v,p')\le\tilde{\Lambda}+(c_{\top}\cdot
  d+2)\cdot \Delta$.  Let $P_{v}$ be the set of all the portals at distance at most
  $\tilde{\Lambda}+(c_{\top}\cdot d+2)\cdot \Delta\le2\Lambda$ from $v$ w.r.t.
  $G_{j}$.  Since the portal in $P_{v}$ are all in a ball of radius $2\cdot \Lambda$ around $v$, and the  
  pairwise distance between any pair of vertices of $P_{v}$ in $G$ is at least $c_\Lambda\cdot\Lambda$, the packing property (\Cref{lem:doubling_packing}) implies that
  \[
  |P_{v}|\le\left(\left\lceil \frac{4\Lambda}{c_{\Lambda}\cdot\Lambda}\right\rceil \right)^{d}=e^{d\cdot\ln\left\lceil \frac{4}{c_{\Lambda}}\right\rceil }~.
  \]

	We are now ready to analyze the clustering
        probability. Initially ignore the truncation in the
        shifts. 
%
%
%
        We show that the
        probability that, for any given vertex $v$, the portal that maximizes $g_v(\cdot)$ is at least $2\Delta$ ahead of other portals is at
        least $e^{-2}$.         
        Let $p_{1}$ and $p_{2}$ be the first and second portals maximizing
$g_{v}$ (i.e. $g_{v}(p_{1})\ge g_{v}(p_{2})\ge g_{v}(p)$ for $p\ne p_{1},p_{2}$).
Denote $a=\delta^{(j)}_{p_{2}}-d_{G_j}(v,p_{2})+d_{G_j}(v,p_{1})$.
As $g_{v}(p_{1})\ge g_{v}(p_{2})$ it holds that $\delta^{(j)}_{p_{1}}\ge a$.
Since $\tilde{\delta}^{(j)}_p$ has exponential distribution with parameter $1$ and $\delta^{(j)}_p$ equals $\tilde{\delta}^{(j)}_p \cdot \Delta$ assuming no truncation, it follows from the memoryless property of the exponential distribution that 
\[
\Pr\left[\delta^{(j)}_{p_{1}}\ge a+2 \Delta\mid\delta^{(j)}_{p_{1}}\ge a\right] = \Pr\left[\tilde{\delta}^{(j)}_{p_{1}}\ge a/\Delta +2 \mid\tilde{\delta}^{(j)}_{p_{1}}\ge a/\Delta\right] \ge e^{-2}\,,
\]
(there is no equality as it might be $a<0$). If this event indeed
occurs, then $g_{v}(p_{1})\ge 2\Delta+g_{v}(p_{2})$. By
the law of total probability (summing over all $p_{1},p_{2}$),
it holds that $\Pr\left[g_{v}(p_{v}) \ge 2\Delta+\max_{p\ne p_{v}}g_{v}(p)\right]$ is at least $e^{-2}$.
        
        We now consider truncation of the shift $\delta_p^{(j)}$ for any portal $p$.  The probability that some portal in $P_{v}$ is
        truncated, is, by a union bound, at most $e^{d\cdot\ln\left\lceil \frac{4}{c_{\Lambda}}\right\rceil }\cdot e^{-c_{\top}\cdot d}\le\frac{e^{-2}}{2}$ (recall that $c_{\top}=\ln\left\lceil \frac{4}{c_{\Lambda}}\right\rceil +3$).  By another union bound, the
        probability that either $C$ is not clustered ignoring
        truncation, or someone in $P_{v}$ is truncated is at most
        $1-e^{-2}+\frac{e^{-2}}{2}=1-\frac{e^{-2}}{2}$. We conclude
        that with probability $\frac{e^{-2}}{2}$ both events do
        not occur, leading to $C$ being clustered.
	
	The probability that $C$ remains unclustered after $s$ rounds is
	thus at most $(1-\frac{e^{-2}}{2})^{s}\le2^{-\frac{s}{10 }}$.
\end{proof}

\subsubsection{The Global Argument}

Recall that the algorithm runs in $s$ rounds, and in each round $i$,
each center $p\in P$ samples shifts $\delta_{p}^{(i)}$, according to
the betailed exponential distribution. For every cluster $C$, let
$\Psi_{C}$ be the event that $C$ remained unclustered after $s$
rounds. Then, by \Cref{lem:clusterProb}, $\Pr[\Psi_{C}]\le2^{-\frac{s}{10}}$. Denote by ${\cal
  A}=\left\{ \Psi_{C}\right\} _{C\in{\cal C}}$ the collection of this
events. Let $\Gamma[\Psi_{C}]\subset\left\{ \Psi_{C'}\in{\cal A}\mid
d_{G}(C,C')\le 8 s\cdot\Lambda \right\} $.
\begin{lemma}\label{lem:ddimInependentEvents}
	For every cluster $C$, $\Psi_{C}$ is independent from the collection
	of events ${\cal A}\setminus\Gamma(\Psi_{C})$.
\end{lemma}
\begin{proof}
To establish the independence of
  $\Psi_{C}$ from the events in ${\cal A}\setminus\Gamma(\Psi_{C})$,
  we consider the execution of 
  cluster aggregation as a distributed algorithm and show that no
  vertex in $C$ ever communicates with any vertex in ${\cal A}
  \setminus \Gamma(\Psi_{C})$.  For simplicity, we will assume that
  all the edges in $G$ have integer weight (we can scale the weights
  accordingly to achieve this). In the remaining, we assume that every edge of weight $w$ in the graph is subdivided into $w$ unit weight edges. Thus we will assume that the graph is unweighted. Note that all this changes are for the sake of argument only.

  To formalize the argument, consider the LOCAL model of computation. Here each vertex represent a server, and communication happens in synchronized rounds, where arbitrary messages can be send on every edge.
  A natural implementation of the \cite{MillerPX2013} algorithm in the LOCAL model is the following: each portal $p$
  wake up at time $-\delta_p$ and begin to ``spread'' along the edges
  in a synchronous manner.  The spread of all portals is
  performed simultaneously.  For any vertex $v$, the first
  portal that spreads to it is the $p$ that maximizes $g_v(p)$. The required number of rounds is bounded by the maximum shift+maximum distance to portal. See \cite{EN22} for additional details of this implementation.
  
  In our \Cref{alg:cluster-agg-doubling}, we have $s$ iterations.
  In the $j$th iteration of the algorithm, every vertex has a portal at distance at most $\Lambda+j\cdot(c_{\top}\cdot d+2)\cdot \Delta$ (by \Cref{lem:DistanceToClosestPortal}). 
  Since the maximum shift is $c_{\top}\cdot d\cdot \Delta$, each vertex $v$ can find that $p$ that maximizes $g_v(p)$ in $\Lambda+ (j+2)\cdot c_{\top}\cdot
  d\cdot \Delta$ synchronous rounds. 
%
%
%
%
%
  Once every vertex $v$ in a cluster has identified the portal that maximizes $g_v(\cdot)$ for that iteration, in another $2\Delta$ rounds of communication, the cluster can determine if it has chosen $p$.  If a cluster $C$ has chosen $p$, to determine whether $C$ is satisfied by $p$, each vertex in $C$ communicates along the shortest path from $v$ to $p$.  This requires an additional $2\cdot\left(\Lambda+ (j+2)\cdot c_{\top}\cdot
  d\cdot \Delta\right)$ rounds to complete. 
  Therefore, \Cref{line:ddimCreageGi+1}  of iteration $j$ can be completed in at most
  \begin{equation}
  	3\Lambda + (3(j+2)\cdot c_{\top}\cdot d+2)\cdot \Delta\le4\cdot\Lambda~,\label{eq:LambdaTildeBound2}
  \end{equation}
   rounds, where the inequality holds for appropriate choise of the constants $c_s$ and $\lambda$ (see \Cref{rem:constantsChoiseDoubling}).
%
  In total, as our algorithm runs for $s$ iterations, it can be
  executed in $4s\cdot\Lambda$ 
	synchronous rounds.
\junk{We actually need to reconstruct the graph and similar  issues.}

        Thus, the event of a cluster being assigned to a portal is a function only of the random choices of the shifts made by the portals at most $4s\cdot\Lambda$ away from the cluster.  Hence, if two clusters are at distance greater than
        $8s\cdot\Lambda$, their corresponding events are completely independent. The
        lemma follows.
\end{proof}
The proof of the main theorem in this section uses the constructive
version of the Lov\'asz Local Lemma by Moser and Tardos \cite{MT10}.
\begin{lemma}[Constructive Lov\'asz Local Lemma] \label{lem:lovasz}
	Let $\mathcal{P}$ be a finite set of mutually independent random variables in a probability space. Let $\mathcal{A}$ be a set of events determined by these variables. For $A\in\mathcal{A}$ let $\Gamma(A)$ be a subset of $\mathcal{A}$ satisfying that $A$ is independent from the collection of events $\mathcal{A}\setminus(\{A\}\cup\Gamma(A))$.
	If there exist an assignment of reals $x:\mathcal{A}\rightarrow(0,1)$  such that
	\[\forall A\in \mathcal{A}~:~~\Pr[A]\le x(A) \cdot \Pi_{B\in\Gamma(A)}(1-x(B))~,\]
	then there exists an assignment to the variables $\mathcal{P}$ not violating any of the events in $\mathcal{A}$. Moreover, there is an algorithm that finds such an assignment in expected time  $\sum_{A\in\mathcal{A}}\frac{x(A)}{1-x(A)}\cdot{\rm poly} \left(|\mathcal{A}|+|\mathcal{P}|\right)$.
\end{lemma}

\begin{LabeledProof}{\Cref{thm:caDD}}
  We are given a graph $G$ with doubling dimension $d$ and a partition
  ${\cal C}$ of clusters with strong diameter at most $\Delta$. The clusters also contains centers $\{v_C\}_{C\in\calC}$ at minimum pairwise distance  $\min_{C,C'\in\calC}d_G(v_C,v_{C'})=\frac{c_\Lambda}{3}\cdot\Delta$.
%
  By the packing property, it follows that
  \begin{align*}
  	\left|\Gamma[\Psi_{C}]\right| & =\left|\left\{ v_{C'}\mid d_{G}(C,C')\le8s\cdot\Lambda\right\} \right|\\
  	& \le\left|\left\{ v_{C'}\mid d_{G}(v_{C},v_{C'})\le8s\cdot\Lambda+2\Delta\right\} \right|\\
  	& =\left(\frac{2\cdot(8s\cdot\Lambda+2\Delta)}{\frac{c_{\Lambda}}{3}\cdot\Delta}\right)^{d}\\
  	& \le\left(\frac{60\cdot c_{s}\cdot\lambda}{c_{\Lambda}}\cdot\ddim^{4}\log^{2}\ddim\right)^{d}~.
  \end{align*}
For every cluster $C\in\calC$, let $x(\Psi_{C})=p=e\cdot2^{-\frac{s}{10}}$.
Then for every
$\Psi_{C}\in{\cal A}$ it holds that
\begin{align}
	x(\Psi_{C})\cdot\Pi_{\Psi_{C'}\in\Gamma(\Psi_{C})}(1-x(\Psi_{C'})) & =p\cdot(1-p)^{|\Gamma(\Psi_{C})|}\nonumber \\
	& \ge p\cdot e^{-2p\cdot|\Gamma(\Psi_{C})|}\nonumber \\
	& \ge e\cdot2^{-\frac{s}{10}}\cdot\exp\left(-2e\cdot2^{-\frac{s}{10}}\cdot\left(\frac{60\cdot c_{s}\cdot\lambda}{c_{\Lambda}}\cdot\ddim^{4}\log^{2}\ddim\right)^{d}\right)\nonumber \\
	& \ge2^{-\frac{s}{10}}\label{eq:cSLambdaChoise}\\
	& \ge\Pr[\Psi_{C}]~,\nonumber 
\end{align}
where the last inequality holds by \Cref{lem:clusterProb}, and inequality (\ref{eq:cSLambdaChoise}) hold for appropriate choice of the constants $\lambda,c_s$ (see \Cref{rem:constantsChoiseDoubling}).
Hence by \Cref{lem:lovasz}, there is an algorithm running in polynomial
time that assigns all the shifts so that after $s$ rounds all the clusters will be assigned. By
\Cref{lem:DistanceToClosestPortal}, each vertex $v$ will be
assigned to a portal at distance (in the induced graph) at most
$d_{G}(P,v)+s\cdot(c_{\top}\cdot d+2)\cdot \Delta\le d_{G}(P,v)+4c_{s}\cdot\ln\frac{4}{c_{\Lambda}}\cdot\ddim^{2}\log
\ddim\cdot \Delta$.  The theorem follows.
\end{LabeledProof}
\begin{remark}\label{rem:constantsChoiseDoubling}
	We did not explicitly state the values of the constants $\lambda$ and $c_s$ during the proof. However, their value came into play only in Equations (\ref{eq:LambdaTildeBound}), (\ref{eq:LambdaTildeBound2}), (\ref{eq:cSLambdaChoise}), and also Equation (\ref{eq:cslambdaGlobal}) from the proof of \Cref{thm:SSPDD}. One can verify that for every large enough constant $\Lambda$, an appropriate constant $c_s$ exist. Indeed,
	\begin{itemize}
		\item In Equation (\ref{eq:LambdaTildeBound}) it is enough that  $c_{s}\cdot\ddim\cdot\log\ddim=s\le\frac{\Lambda}{\Delta\cdot(c_{\top}\cdot d+2)}$
		or $c_{s}\le\frac{\lambda\cdot\ddim^{2}}{c_{\top}\cdot\ddim+2}$.
		\item  In Equation (\ref{eq:LambdaTildeBound2}) it is enough that $(3(s+2)\cdot c_{\top}\cdot d+2)\cdot\Delta\le\Lambda$ or $c_{s}=\frac{s}{\ddim\cdot\log\ddim}\le\frac{1}{\ddim\cdot\log\ddim}\cdot((\frac{\Lambda}{\Delta}-2)\cdot\frac{1}{3\cdot c_{\top}\cdot d}-2)=\frac{1}{\ddim\cdot\log\ddim}\cdot((\lambda\cdot\ddim^{3}\log\ddim-2)\cdot\frac{1}{3\cdot c_{\top}\cdot d}-2)$.
		\item In Equation (\ref{eq:cslambdaGlobal}) it is enough that $c_{s}\le\lambda\cdot\frac{\ddim'}{\alpha}\cdot\frac{1}{8\cdot\ln\frac{4}{c_{\Lambda}}}=\lambda\cdot\frac{\ddim}{\alpha}\cdot\frac{1+\log\frac{8}{c_{\Lambda}}}{8\cdot\ln\frac{4}{c_{\Lambda}}}$, where $\alpha=O(\ddim)$ is the parameter from \Cref{thm:MPXbasedClusteringDoubling} (determining the maximum number of clusters a ball can intersect).
		\item While in Equation (\ref{eq:cSLambdaChoise}) it is enough that $\frac{s}{10d}=\frac{c_{s}\cdot\log d}{10}\ge\log\left(\frac{60\cdot c_{s}\cdot\lambda}{c_{\Lambda}}\cdot\ddim^{4}\log^{2}\ddim\right)$.
	\end{itemize}
\end{remark}

\subsection{Cluster Aggregation in Bounded Pathwidth Graphs}\label{sec:CApathwidth}
This section is devoted to proving the following theorem.
\begin{restatable}{theorem}{CAPW}\label{thm:caPW}
	Every instance of cluster aggregation on a graph of pathwidth $\pw$ has a $8( \pw + 1 )$-distortion solution that can be computed in polynomial time.
\end{restatable}

\subsubsection{Overview of Cluster Aggregation Algorithm}

Consider a graph $G = (V, E, w)$ with pathwidth $\pw \geq 1$
with a path decomposition $T$
that has bags
$X_1, \ldots, X_s$ and width $\pw$.
Thus, 
$X_i \subseteq V$ and $|X_i| \leq \pw+1$, for all $1 \leq  i \leq s$.
Assume an orientation of the bags in $T$ from left to right such that
$X_1$ is the leftmost bag and $X_s$ is the rightmost bag of $T$,
while $X_i$ is adjacent to $X_{i-1}$ on the left and $X_{i+1}$ on the right, 
where $1 < i < s$.
Let $\mcC$ be a partition of $V$ 
with strong diameter at most $\Delta$.
Let $P \subseteq V$ be a set of portal nodes.
We present an aggregation algorithm for graph $G$ with path decompositions $T$ that 
produces a set $\mcC'$ of aggregated clusters
through the aggregation mapping $f$.

\begin{figure}[t]
    \centering
        \centering
        \includegraphics[width=0.5\columnwidth]{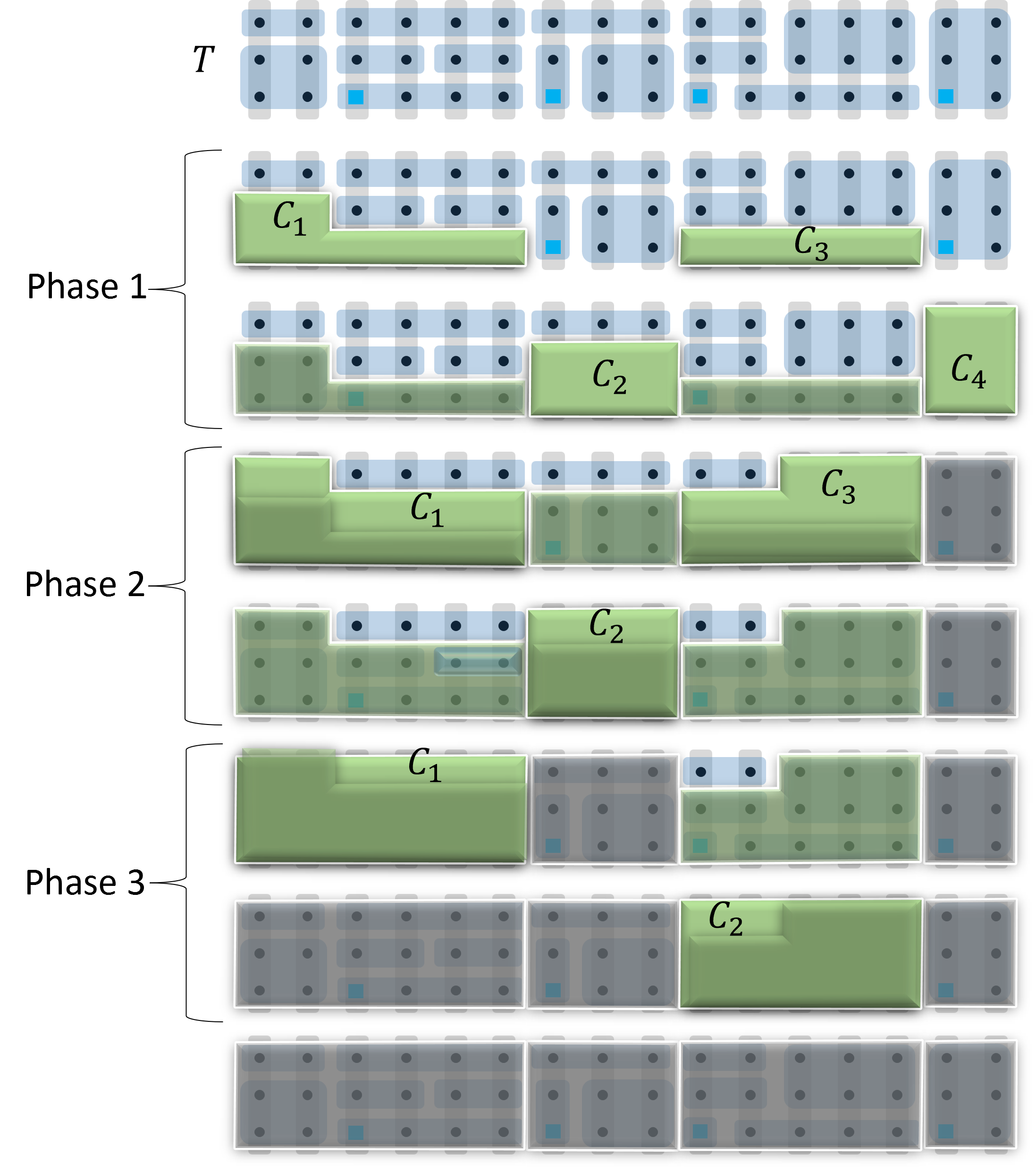}
    \caption{
    An example execution of the pathwidth aggregation algorithm on a path $T$ with width $\pw =2$.
    Each bag of $T$ is a vertical set of 3 nodes. The clusters and four portals are highlighted in the top part of the figure. The algorithm executes in $\pw + 1 = 3$ phases.
    Phase $i$ creates groups of clusters $\mcC_{1},\mcC_{2}, \ldots$.
    Phase $i$ consists of two subphases so that in the first (resp.\ second) subphase 
    the groups of the odd (resp.\ even) subsequence get assigned to a portal. The groups are shown with green-shaded areas. At the bottom is the final clustering where gray-shaded areas indicate finalized clusters that cannot grow further. 
    }\label{fig:pathwidth}
\end{figure}

The algorithm works in $\pw+1$ phases.
At each phase, each bag $X_i$ that contains nodes unassigned to a portal
has at least one of its member nodes assigned to a portal.
Therefore, by the end of the last phase $\pw+1$, 
each bag $X_i$ has all of its member nodes assigned to portals.
Whenever a node $u \in X_i$ is assigned to a portal,
all the nodes in the cluster $C_u \in \mcC$ that contains $u$
are also assigned to the same portal.
Hence, by the end of the last phase $\pw+1$, 
each cluster $C \in \mcC$ is assigned to a portal $f(C) \in P$.

The total detour in our algorithm is $O(\pw \cdot \Delta)$.
A simplified explanation is as follows.
At each phase,
an unassigned cluster 
could merge with a cluster that was previously assigned to a portal at an earlier phase.
Such merging incurs $O(\Delta)$ additional detour to the newly merged cluster.
Hence, in the worst case, the total detour is the product  
of the number of phases $\pw + 1$ times the detour $O(\Delta)$ contributed by a phase, giving a total detour of $O(\pw \cdot \Delta)$.

A more detailed explanation is as follows.
At each phase,
the algorithm organizes some of the unassigned clusters
into groups. Clusters 
in the same group will aggregate toward the same portal.
The groups are formed in a way that allows their clusters to merge with $O(\Delta)$ detour at each phase.
Moreover, each bag of $T$ with unassigned nodes will contribute at least one of its unassigned nodes in a group,
which guarantees the necessary progress of node assignments for the bags at each phase.

In a phase,
let the cluster groups be
$\mcC_1, \ldots, \mcC_{\zeta}$ 
(see Figure \ref{fig:pathwidth}).
Clusters in the same group $\mcC_j$ 
span consecutive bags of $T$ with unassigned nodes.
The groups are ordered from left to right, so that 
$\mcC_1$ contains the leftmost bags of $T$ with unassigned clusters,
while $\mcC_\zeta$ has the rightmost bags.

The groups are processed in two subphases,
where the first subphase is dedicated to the odd subsequence groups,
while the second subphase is dedicated to the even subsequence groups.
Odd subsequence
groups $\mcC_1, \mcC_3, \ldots$,
are conflict-free in the sense that paths to portals from one group
do not use the same cluster as the paths to portals from another group.
This allows the odd subsequence groups to be processed in one subphase,
adding at most $2\Delta$ detour to the merged clusters.

In a similar way, 
the even subsequence groups $\mcC_2, \mcC_4, \ldots$
are processed in a second subphase,
adding at most $6 \Delta$ detour.
The reason for the factor $3$ (compared to  first subphase) 
is that there could be conflicts 
between adjacent even and odd subsequence groups. 
In particular, group $\mcC_j$ of the even subsequence  
may conflict with the odd subsequence 
group $\mcC_{j-1}$ on the left 
and group $\mcC_{j+1}$ on the right. 
The left and right groups may each define an additional possible assignment portal.
Hence, there are three candidate assignment portals 
for group $\mcC_j$.
Merging of $\mcC_j$ can be accomplished in three main steps,
one step for each portal, where each step contributes at most $2\Delta$ detour.
Thus, the second subphase contributes at most $6\Delta$ detour in total.

In fact, the second subphase may not merge all the clusters in a group $\mcC_j$ of the even subsequence.
But for each cluster $C \in \mcC_j$ that is not merged,
the algorithm guarantees that
at least one cluster (and therefore at least one node) from each bag spanned by $C$
is assigned to a portal.
This guarantees the required assignment progress for the phase.

Therefore, each phase adds a total $2\Delta + 6\Delta = 8\Delta$ detour.
Consequently, considering all the $\pw + 1$ phases, 
the overall detour is at most $8 (\pw + 1) \Delta$.
Hence, the distortion of the cluster aggregation algorithm is $8(\pw + 1)$.

\subsubsection{Detailed Pathwidth Algorithm}

Assume that every vertex $v\in V$ has a
unique closest portal $p_v \in P$ 
(if $v$ has multiple closest portals, then pick one arbitrarily).
Let $\pi_v$ be a shortest path from $v$ to $p_v$.
We also assume that each node $x \in \pi_v$ has the same portal $p_x = p_v$,
where the shortest path $\pi_x$ is the {\em suffix} (the remaining part) of $\pi_v$ that starts at $x$ and ends at $v$.
For any set of nodes $Y \subseteq V$,
let $\mcC(Y) = \{C \in \mcC \mid C \cap Y \neq \emptyset \}$ 
denote the set of clusters that intersect $Y$.

We say that a vertex $v\in V$ is assigned if $f(v)$ is already defined.
In our cluster merging algorithm,
for a cluster $C \in \mcC$, 
when a node $v \in C$ is assigned to a portal $p$ 
then 
the whole cluster is assigned to $p$.
In this case,
we write
$f(C) = p$ to denote that each $v \in C$ is assigned to the same portal $p$ and 
we say that $C$ is assigned to $p$. 
For a portal $p$ we denote by $f^{-1}(p)$
the set of clusters that have been assigned to $p$.

An unassigned node $v$ has a {\em preferred} portal.
By default, the preferred portal of $v$ is $p_v$.
However, during the execution of the algorithm,
nodes in path $\pi_v$ may be assigned and the
preferred portal of $v$ may change.
Let $x$ be the first node assigned on path $\pi_v$
starting from $v$.
Then, the preferred portal of $v$ is $f(x)$, the portal to which $x$ was assigned.
However, if $\pi_v$ has no assigned nodes,
then the preferred portal of $v$ remains $p_v$.

We can extend the notion of preferred portals to a cluster $C \in \mcC$.
The set of preferred portals of $C$, denoted $P(C)$, includes the 
preferred portals of all the nodes in $C$.

\paragraph*{Cluster Groups.}
Consider phase $\phi$.
For a bag $X_i$ the {\em partial bag} $Y \subseteq X_i$
contains the nodes
that are unassigned at the beginning of phase $\phi$. 
Denote by $Y_1, Y_2, \ldots$ the partial bags of $T$ which are non-empty at the beginning of phase $\phi$ and are ordered from left to right in $T$.

Let $P(Y_i) = \bigcup_{C \in \mcC(Y_i)} P(C)$ 
denote the preferred portals of all the clusters in $\mcC(Y_i)$.
For each portal $p \in P(Y_i)$ 
we define the respective list of partial bags 
$L_{i}(p) = Y_i, Y_{i+1}, \ldots$,
to be the longest sequence of consecutive non-empty partial bags 
that starts at $Y_{i}$ such that each partial bag $Y_j$
in the list has $p$ as one of its preferred portals, 
that is, $p \in P(Y_j)$.

Let ${\widehat p}_{i} \in P(Y_i)$
be the portal that defines the longest list for $Y_{i}$,
namely, 
$|L_{i}({\widehat p}_{i})| = \max_{p \in P(Y_{i})} |L_{i}(p)|$
(if there is more than one portal in $P(Y_i)$ that defines the list of maximum length for $Y_i$,
then pick one of these portals arbitrarily to be ${\widehat p}_{i}$).
Denote this longest list for $Y_i$ as ${\widehat L}_{i} = L_{i}({\widehat p}_{i})$.

Let $\mcC({\widehat L}_{i}) \subseteq \mcC$ 
denote the {\em group} (set of) of clusters 
that intersect with the partial bags of ${\widehat L}_{i}$,
where each $C \in \mcC({\widehat L}_i)$ has ${\widehat p}_{i}$ 
as a preferred portal.
We also say that ${\widehat p}_{i}$ is the preferred portal of group $\mcC({\widehat L}_{i})$.
We recursively define a sequence of cluster groups 
$\mcC_1, \mcC_2, \ldots$,
as follows:
\begin{itemize}
\item The first group is $\mcC_1 = \mcC({\widehat L}_1)$.
\item Suppose that we have defined group $\mcC_k = \mcC({\widehat L}_{i})$,
based on a list ${\widehat L}_{i}$, where $k,i \geq 1$.
\item
We define group $\mcC_{k+1} = \mcC({\widehat L}_{j})$
such that the first (leftmost) partial bag $Y_{j} \in {\widehat L}_{j}$ 
immediately follows the last (rightmost) partial bag $Y_\ell \in {\widehat L}_{i}$,
namely, $j = \ell + 1$, where $\ell \geq i$.
\end{itemize}
Let $\mcC_{\zeta}$ be the last group defined in the above sequence.
Note that each partial bag $Y_i$ is included in the list ${\widehat L}_{j}$ of exactly one group in $\mcC_1, \ldots, \mcC_{\zeta}$.

\paragraph*{Subphases.}

Consider a phase $\phi$/
Let $S = \mcC_1, \mcC_2, \ldots, \mcC_{\zeta}$ 
be the groups of phase $\phi$  
with respective preferred portals 
$p_1, p_2, \ldots, p_{\zeta}$.

Consider a group $\mcC_i \in S$.
at the beginning of phase $\phi$, 
each cluster $C \in \mcC_i$ has $p_i$
as a preferred portal.
We define a path $\pi_C$ 
toward portal $p_i$
which consists only of unassigned nodes as follows.
Pick arbitrarily a node $v \in C$ with preferred portal $p_i$.
If $\pi_v$ does not have any assigned nodes,
then $\pi_C = \pi_v$.
If $\pi_v$ has assigned nodes, 
then let $x$ be the first assigned node on $\pi_v$ starting from $v$.
Clearly, $f(x) = p_i$.
Let $\pi_C$ be the prefix (earlier part) of $\pi_v$ from $v$ up to
the node just before $x$.

Let $S_{\text{odd}} = \mcC_1, \mcC_3,\ldots$ 
and $S_{\text{even}} = \mcC_2, \mcC_4,\ldots$ 
be the respective odd and even subsequences of $S$.
Groups in $S$ are processed
in two subphases of $\phi$, where the first subphase handles the
groups in $S_{\text{odd}}$,
while the second subphase handles the groups in $S_{\text{even}}$.

\begin{itemize}
\item{\em First subphase of $\phi$ (groups $S_{\text{odd}}$):}

Consider a group $\mcC_i \in S_{\text{odd}}$ (with preferred portal $p_i$).
For each $C \in \mcC_i$,
all unassigned clusters of $\mcC$ that intersect ${\pi}_C$
are assigned to the portal $p_i$.

\item{\em Second subphase of $\phi$ (groups $S_{\text{even}}$):}

Consider a group $\mcC_i \in S_{\text{even}}$ (with preferred portal $p_i$).
For each cluster $C \in \mcC_i$,
the respective path $\pi_C$ 
may now contain nodes assigned from the first subphase.

We define a prefix path $\pi'_C$ of $\pi_C$ consisting only of unassigned nodes
as follows.
If $\pi_C$ does not have any assigned nodes then $\pi'_C = \pi_C$. 
If $\pi_C$ has assigned nodes (from the first subphase), 
let $x_C \in \pi_C$ be its first assigned node with respect to the start node $v$ of $\pi_C$.
In this case, $\pi'_C$ is the prefix of $\pi_C$ from $v$ to the node just before $x_C$.

From Lemma~\ref{lemma:left-right-portal},
$f(x_C) = p_j$ where $j \in \{i-1, i+1\}$.
Thus, we can write $\mcC_i = \mcC_i^1 \sqcup \mcC_i^2 \sqcup \mcC_i^3$,
such that for each cluster $C \in \mcC_i$ there are three possibilities: 
(i) $C \in \mcC_i^1$ if $\pi_C$ does not have any assigned nodes, that is, $\pi_C = \pi'_C$;
(ii) $C \in \mcC_i^2$ if $\pi_C$ has some assigned node and $f(x_C) = p_{i-1}$;
(iii) $C \in \mcC_i^3$ if $\pi_C$ has some assigned node and $f(x_C) = p_{i+1}$.
We perform the merging for the clusters of $\mcC_i$ in three respective steps.
Let $\text{Assigned}(\mcC_i^k)$ denote the set of clusters that are assigned in step $k$.

\begin{itemize}
    \item {\em Step 1 (clusters $\mcC_i^1$):} 
    For each cluster $C \in \mcC_i^1$
    all the unassigned clusters of $\mcC$ that intersect ${\pi}'_C$
    are assigned to portal $p_i$.
    
    \item {\em Step 2 (clusters $\mcC_i^2$):}    
    Let ${\widehat L}_{l}$ be the respective list that defines $\mcC_i$.
    Let $Y_R \in {\widehat L}_{l}$ be the rightmost partial bag of $T$
    that intersects with some cluster $C \in \mcC_i^2$.

    We can write $\pi'_C$ as a concatenation of subpaths of maximal length $\pi''_1, \pi''_2, \ldots$ such that odd subpaths $\pi''_1, \pi''_3, \ldots$ do not intersect any clusters of $\text{Assigned}(\mcC_i^1)$,
    while even subpaths $\pi''_2, \pi''_4, \ldots$ intersect only with clusters of $\text{Assigned}(\mcC_i^1)$.
    If $C \in \text{Assigned}(\mcC_i^1)$, then we assume for simplicity that $\pi''_1$ is empty, so that $\pi''_2$ starts with nodes of $C$;
    otherwise (that is, $C \notin \text{Assigned}(\mcC_i^1)$), $\pi''_1$ starts with nodes of $C$.
    Note that odd subpaths have all their nodes unassigned.

    The clusters overlapping the even subpaths have already been assigned in Step 1.
    Therefore, consider the odd subpaths.
    Take a subpath $\pi''_z$, where $z$ is odd (hence, none of the nodes in $\pi''_z$ are assigned).
    Let $C'$ be the assigned cluster that immediately follows $\pi''_z$.
    Each cluster $C''$ intersecting $\pi''_z$ is assigned to portal $f(C'') = f(C')$.
    We treat similarly the other odd subpaths.
    
    \item {\em Step 3 (clusters $\mcC_i^3$):}
    This is similar to Step 2. 
    Symmetrically, instead of $Y_R$, we use $Y_L \in {\widehat L}_{l}$ to be the leftmost partial bag of $T$
    that intersects with a cluster $C \in \mcC_i^3$.
    Similarly to Step 2, we write $\pi'_C$ as a concatenation of subpaths.
    The odd subpaths do not intersect with any assigned clusters in $\text{Assigned}(\mcC_i^1) \cup \text{Assigned}(\mcC_i^2)$,
    while the even subpaths intersect only with clusters in $\text{Assigned}(\mcC_i^1) \cup \text{Assigned}(\mcC_i^2)$.
    The assignment of clusters intersecting odd subpaths is performed as in Step 2.
\end{itemize}
Hence, the newly assigned clusters are $\text{Assigned}(\mcC_i) = \bigcup_{j=1}^3\text{Assigned}(\mcC_i^j)$.
Note that the second subphase may not assign all clusters of group $\mcC_i \in S_{\text{even}}$.
That is, there may exist a cluster $C' \in \mcC_i^2 \cup \mcC_i^3$
that is not assigned (i.e. $C' \notin \text{Assigned}(\mcC_i)$),
because the chosen cluster $C$ based on which the aggregation is performed in Step 2 or 3 may have a respective path $\pi'_C$ 
that does not cross $C'$ (and also $C'$ was not assigned in Step 1).
However, this is not an issue since, as we show in the analysis,
each partial bag of ${\widehat L}_{l}$ (which defines $\mcC_i$)
will have a node assigned in the current phase $\phi$ (Lemma \ref{lemma:group-assign}).
\end{itemize}

\subsubsection{Analysis}

Consider a phase $\phi$.
We start with results that prove the independence
between groups of the odd (resp.\ even) subsequence.
Thus, when analyzing the detour and the assignment process
we can focus on one group at a time.

\begin{lemma}
\label{lemma:disjoint-groups}
For any $i \neq j$, $\mcC_i \cap \mcC_j = \emptyset$.
\end{lemma}

\begin{proof}
    Without loss of generality suppose that $i < j$.
    Let ${\widehat L}_{i'}$ and ${\widehat L}_{j'}$
    be the respective lists that defined groups $\mcC_i$ and $\mcC_j$.
    Suppose that there is a cluster $C \in \mcC_i \cap \mcC_j$.
    Clearly, cluster $C$ is unassigned.
    Thus, $C$ overlaps a sequence of partial bags including the last partial bag of ${\widehat L}_{i'}$ up to at least the first partial bag of ${\widehat L}_{j'}$.
    Since $p_i = {\widehat p}_{i'}$ is a preferred portal by $C$ (due to $C \in \mcC_i$),
    the partial bags spanned by $C$ qualify to be members of ${\widehat L}_{i'}$.
    Thus, ${\widehat L}_{i'}$ is not the longest list with respect to portal ${\widehat p}_{i'}$.
    A contradiction.
\end{proof}

\begin{lemma} 
\label{lemma:path-preferred-portal}
For any cluster $C \in \mcC_{i}$,
the clusters that path $\pi_C$ crosses
must have $p_i$ as a preferred portal. 
\end{lemma}

\begin{proof}
    From Lemma \ref{lemma:disjoint-groups}, $C$ belongs only to $\mcC_{i}$,
    which implies that there is a unique path $\pi_C$ for $C$.
    Let $v \in C$ be the origin node of path $\pi_{C}$.
    We have that $\pi_{C}$ is a prefix of path $\pi_v$ toward portal $p_v$.
    All nodes along path $\pi_v$
    also have the same closest portal $p_v$.
    If $\pi_{C} = \pi_v$, then it must be $p_i = p_v$, and all nodes along
    $\pi_{C}$
    also have trivially preferred portal $p_i = p_v$.
    If $\pi_{C} \neq \pi_v$, then the first assigned node $x$ along $\pi_v$
    defines the preferred portal $p_i = f(x)$, which 
    makes all the nodes along $\pi_{C}$ to have $p_i = f(x)$ as their preferred portal as well.
    Therefore, each cluster $C'$ that intersects $\pi_C$ must also have $p_i$
    as a preferred portal.
\end{proof}

\begin{lemma}
\label{lemma:paths-same-cluster}
    For any two clusters $C_1 \in \mcC_{i}$ and $C_2 \in \mcC_{j}$,
    where $j \geq i + 2$,
    the respective paths $\pi_{C_1}$ and $\pi_{C_2}$
    do not go through any of the same clusters of $\mcC$.
\end{lemma}

\begin{proof}
    Suppose that $\pi_{C_1}$ and $\pi_{C_2}$
    go through the same cluster $C' \in \mcC$.
    Let $\mcC(\pi_{C_1})$ (resp. $\mcC(\pi_{C_2})$)
    be the set of clusters that intersect $\pi_{C_1}$ (resp. $\pi_{C_2}$).
    Clearly, $C' \in \mcC(\pi_{C_1}) \cap \mcC(\pi_{C_2})$.
    From Lemma~\ref{lemma:path-preferred-portal}, since $C_1 \in \mcC_i$, 
    all clusters in $\mcC(\pi_{C_1})$ have $p_i$ as a preferred portal.
    Similarly, all clusters in $\mcC(\pi_{C_2})$ have $p_j$ as a preferred portal.

    Let ${\widehat L}_{i'}$ and ${\widehat L}_{j'}$
    be the respective lists that defined groups $\mcC_i$ and $\mcC_j$
    (hence, the respective preferred portals 
    are $p_i = {\widehat p}_{i'}$ and $p_j = {\widehat p}_{j'}$).
    Let $Y_\ell$ be the rightmost partial bag of ${\widehat L}_{i'}$.
    Let $Y_{\xi}$ be the rightmost partial bag that 
    intersects a cluster of $\mcC(\pi_{C_1})$.
    There are two possible cases:
    \begin{itemize}
        \item $\xi > \ell$:
        All partial bags of ${\widehat L}_{i'} = Y_{i'}, \ldots, Y_{\ell}$ have $p_i = {\widehat p}_{i'}$ as their preferred portal by definition.
        Since $v_1 \in C_1$ and $C_1 \in \mcC_i$, the origin $v_1$ of path $\pi_{C_1}$ must be in a bag $Y_l$, where $l \leq \ell$
        (otherwise the list ${\widehat L}_{i'}$ would be longer).
        Hence, all partial bags 
        $Y_{\ell + 1}, \ldots, Y_{\xi}$ 
        must have ${\widehat p}_{i'}$ as a preferred portal,
        since each such bag intersects some cluster of $\mcC(\pi_{C_1})$
        which, as shown above, has $p_i = {\widehat p}_{i'}$ as a preferred portal.
        Thus, all partial bags $Y_{i'}, \ldots, Y_{\xi}$
        have ${\widehat p}_{i'}$ as preferred portal.
        Since $\xi > \ell$, the list ${\widehat L}_{i'} = L_{i'}({\widehat p}_{i'})$ is not the longest with respect to portal ${\widehat p}_{i'}$.
        A contradiction.

        \item $\xi \leq \ell$:
        In this case, cluster $C' \in \mcC(\pi_{C_1})$ must intersect some partial bag $Y_l$, where $l \leq \ell$.
        The origin $v_2$ of path $\pi_{C_2}$ is in a cluster $C'' \in \mcC_{j}$ that intersects some bag $Y_r$, where $r \geq {j'}$  (by definition of $\mcC_{j}$ and its respective list ${\widehat L}_{j'}$).
        Hence, all partial bags 
        $Y_{\xi}, \ldots, Y_{r}$ 
        must have $p_j$ as a preferred portal,
        since each such bag intersects some cluster of $\mcC(\pi_{C_2})$
        which, as shown above, has $p_j$ as a preferred portal.
        
        Consider the in-between list ${\widehat L}_{\ell+1} = L_{\ell+1}({\widehat p}_{\ell + 1})$ that defines respective group $\mcC_{i+1}$ between groups $\mcC_{i}$ and $\mcC_{j}$ (recall $j \geq i+2$).
        List ${\widehat L}_{\ell+1}$ spans partial bags starting from $Y_{\ell+1}$
        up to at most $Y_{j'-1}$.
        Consider also list $L_{\ell+1}(p_{j})$
        which must span all partial bags $Y_{\ell+1}, \ldots, Y_{r}$,
        since $\xi \leq \ell < \ell + 1$ and all these bags have $p_j$ as a preferred portal.
        Since $r \geq j'$, we have that $|L_{\ell+1}(p_{j})| > |L_{\ell+1}({\widehat p}_{\ell + 1})|$. Thus, ${\widehat p}_{\ell + 1}$ does not define the maximum length list for $Y_{\ell+1}$. A contradiction.
        \qedhere 
    \end{itemize}
\end{proof}

\begin{lemma}
\label{lemma:left-right-portal}
For $C \in \mcC_i$, where $\mcC_i \in S_{\text{even}}$,
any node $x \in \pi_C$ that was assigned during the first subphase 
has $f(x) = p_j$, where $j \in \{i-1,i+1\}$.
\end{lemma}

\begin{proof}
    Lemma \ref{lemma:paths-same-cluster} implies that 
    cluster $C_x$, where $x \in C_x$, was assigned in the first subphase during the merging process 
    of a group $\mcC_j$ where $j \in \{j-1, j+1\}$.
    In the first subphase, the processing of group $\mcC_j$ 
    assigns clusters to portal $p_j$,
    hence, $C_x$ is assigned to $p_j$.
    Thus, $x$ is assigned to $p_{j}$ as well, that is, $f(x) = p_j$.
\end{proof}

We say that two groups $\mcC_{i}$ and $\mcC_{j}$ 
are {\em conflict-free in the first subphase} if for any two clusters
$C_1 \in \mcC_{i}$ and $C_2 \in \mcC_{j}$
the respective paths $\pi_{C_1}$ and $\pi_{C_2}$ 
do not go through the same cluster in $\mcC$.
Lemma \ref{lemma:paths-same-cluster}
implies that the clusters in $S_{\text{odd}}$ 
are conflict-free in the first subphase.

We say that two groups $\mcC_{i}, \mcC_{j} \in S_{\text{even}}$ 
are {\em conflict-free in the second subphase} if for any two clusters
$C_1 \in \mcC_{i}$ and $C_2 \in \mcC_{j}$
the respective paths $\pi'_{C_1}$ and $\pi'_{C_2}$
do not go through the same cluster in $\mcC$.
Since for a cluster $C$ of a group in $S_{\text{even}}$ 
path $\pi'_C$ is a subpath of $\pi_C$,
Lemma \ref{lemma:paths-same-cluster}
implies that the clusters in $S_{\text{even}}$
are conflict-free in the second subphase.
    
\begin{lemma}
    \label{lemma:conflict-free}
    Any two groups in $S_{\text{odd}}$
    are conflict-free in the first subphase,
    and any two groups in $S_{\text{even}}$
    are conflict-free in the second subphase.
\end{lemma}

Lemma \ref{lemma:conflict-free}
implies that we can analyze the detour of each group individually.
Let $\delta$ denote the detour of each assigned node at the end of the previous phase.
We give a bound on the detour caused by the first subphase.

\begin{lemma}
\label{lemma:odd-detour}
    Each node assigned while processing group $\mcC_i \in S_{\text{odd}}$
    has detour at most $\delta + 2\Delta$.
\end{lemma}

\begin{proof}
Consider a node $y \in C_y$ which is assigned while processing $\mcC_i \in S_{\text{odd}}$.
There is a cluster $C \in \mcC_i$ such that $\pi_C$ intersects $C_y$. 
Each cluster that is intersected by path $\pi_C$, including $C_y$, 
is assigned to $p_i$, hence, $f(C_y) = f(y) = p_i$.

Let $v \in C$ be the origin of path $\pi_C$ and take node $z \in C_y \cap \pi_C$.
Thus, $d_{G[C_y]}(y, z) \leq \Delta$.
We also have $p_z = p_v$, since $z \in p_v$.
Moreover, 
$$d_{G}(z,p_v) = d_{G}(z,p_z) \leq d_G(y,z) + d_G(y,p_y) \leq d_{G[C_y]}(y,z) + d_G(y,p_y) \ ,$$
since otherwise $z$ would have $p_y$ as its closest portal.
Furthermore,
$$
\dtr_f(y) 
= d_{G[f^{-1}(p_i)]}(y, p_i) - d_G(y, p_y)
\leq (d_{G[C_y]}(y, z) + d_{G[f^{-1}(p_i)]}(z, p_i)) - d_G(y, p_y)\ .
$$
We examine two cases:
\begin{itemize}
    \item $\pi_C = \pi_v$:
    In this case, the end node of $\pi_C$ is $p_i$.
    Thus, $p_v = p_i$, $p_z = p_i$ and $d_{G[f^{-1}(p_i)]}(z, p_z) = d_G(z,p_z)$,
    since all nodes of $\pi_z$ are assigned to $p_z$.
    Thus,
    $$
    \dtr_f(y) 
    \leq (d_{G[C_y]}(y, z) + d_{G}(z, p_z)) - d_G(y, p_y)
    \leq 2 d_{G[C_y]}(y, z)
    \leq 2 \Delta \ .
    $$

    \item $\pi_C \neq \pi_v$:
    In this case, the first encountered node (starting from v) $x \in \pi_C$ that was assigned has $f(x) = p_i$.
    We have that $p_x = p_v$, since $x \in \pi_v$.
    By the detour $\delta$ caused in previous phases, we have 
    $d_{G[f^{-1}(p_i)]}(x, p_i) \leq \delta + d_{G}(x, p_v)$.
    Hence,
    $$d_{G[f^{-1}(p_i)]}(z, p_i) 
    \leq d_{G}(z,x) + d_{G[f^{-1}(p_i)]}(x, p_i) 
    \leq d_{G}(z,x) + \delta + d_{G}(x, p_v) = \delta + d_{G}(z,p_v) \ .$$
    Thus,
    $$
    \dtr_f(y) 
    \leq (d_{G[C_y]}(y, z) + \delta + d_{G}(z,p_v)) - d_G(y, p_y)
    \leq \delta + 2 d_{G[C_y]}(y, z)
    \leq \delta + 2 \Delta \ .
    $$
\end{itemize}
Consequently, in all cases, the detour for $y$ is bounded by $\delta + 2 \Delta$.
\end{proof}

Let $\delta'$ denote the detour of each assigned node at the beginning of 
the second subphase.
We give a bound on the detour caused by the second subphase.

\begin{lemma}
\label{lemma:even-detour}
    Each node assigned while processing group $\mcC_i \in S_{\text{even}}$,
    has detour at most $\delta' + 6\Delta$.
\end{lemma}

\begin{proof}
Consider a node $y \in C_y$ which is assigned to a portal while processing
group $\mcC_i \in S_{\text{even}}$.
From the algorithm we have that $y$ was assigned in one of the three steps
of the second subphase.
In Step $k$, $y$ was assigned while processing some cluster $C \in \mcC_i^k$
because path $\pi'_C$ goes through cluster $C'_y$.
Let $\delta_k$ denote the detour of Step $k$.
We examine each step separately,

\begin{itemize}
    \item{\em Step 1:}
    In this step $f(y) = p_i$.
    The detour analysis is similar as in \ref{lemma:odd-detour},
    where instead of $\delta$, we have $\delta'$.
    Hence, the detour for $y$ is at most $\delta_1 \leq \delta' + 2 \Delta$. 
    
    \item{\em Step 2:}
    In this step $f(y) = p_j$, where $j \in \{i-1, i\}$.
    There is subpath $\pi''_z$ of $\pi'_C$, where $z$ is odd, that intersects $C_y$.
    Let $x$ be the node that immediately follows $\pi''_z$ along $\pi'_C$.
    Then, all clusters intersecting $\pi''_z$, including $C_y$, are assigned to $f(x)$.
    There are two subcases:
    \begin{itemize}
        \item 
        Node $x$ is assigned in the first subphase. Hence $f(y) = f(x) = f(x_C) = p_{i-1}$ (recall $x_C$ is the first assigned node along $\pi_C$).
        The detour analysis for $y$ is similar as in the proof of Lemma \ref{lemma:odd-detour}, where instead of $\delta$ we have $\delta'$.
        Thus, $\delta_2 \leq \delta' + 2 \Delta$.
        
        \item
        Node $x$ is assigned in Step 1. Thus, $f(y) = f(x) = p_i$.
        Again the detour analysis is similar to Lemma \ref{lemma:odd-detour},
        where instead of $\delta$ we have $\delta_1$,
        giving $\delta_2 \leq \delta_1 + 2 \Delta \leq \delta' + 4 \Delta$.
    \end{itemize}
    Combining the two subcases,
    the maximum detour for $y$ is
    $\delta_2 \leq \delta' + 4 \Delta$.
    
    \item{\em Step 3:}
    This is similar to Step 2. However, now $f(y) = p_j$, where $j \in \{i-1, i, i+1\}$. There are three subcases:
    \begin{itemize}
        \item 
        Node $x$ is assigned in the first subphase. Hence, $f(y) = f(x) = f(x_C) = p_{i+1}$.
        As in Step 2, first subcase, $\delta_3 \leq \delta' + 2 \Delta$.
        
        \item
        Node $x$ is assigned in Step 1. Hence, $f(y) = f(x) = p_{i}$.
        As in Step 2, second subcase, $\delta_3 \leq \delta' + 4 \Delta$.

        \item 
        Node $x$ is assigned in Step 2. Hence, $f(y) = f(x) = p_j$, where $j \in \{i-1, i\}$.
        The analysis of this subcase is similar to Step 2, second subcase,
        where instead of $\delta_1$ we use $\delta_2$,
        which gives detour $\delta_3 \leq \delta_2 + 2 \Delta \leq \delta' + 6 \Delta$.
    \end{itemize}
    Combining the three subcases,
    the maximum detour for $y$ is
    $\delta_3 \leq \delta' + 6 \Delta$.
\end{itemize}
Hence, the overall detour for $y$ is at most $\delta' + 6 \Delta$.
\end{proof}

\begin{lemma}
\label{lemma:group-assign}
    For group $\mcC_i = \mcC({\widehat L}_j)$,
    at least one node from each partial bag of ${\widehat L}_j$
    gets assigned to a portal during phase $\phi$.
\end{lemma}

\begin{proof}
    Let $Y_k$ be a partial bag of ${\widehat L}_j$.
    We will show that at least one node of $Y_k$ gets assigned
    during phase $\phi$.
    By definition of $\mcC_i$, $Y_k$ intersects a cluster $C' \in \mcC_{i}$.
    
    First, consider the case where $\mcC_i$ is handled
    at the first subphase ($i$ is odd).
    Cluster $C'$ is assigned to the portal $p_i$ in the first subphase.
    Hence, at least one node of $Y_k$ gets assigned, as needed. 

    Now, consider the case where $\mcC_i$ is handled
    at the second subphase ($i$ is even).
    If some node of $Y_k$ has already been assigned during the first subphase
    (when processing adjacent groups in $S_{\text{odd}}$),
    then the result is proven.
    
    Suppose none of the nodes in $Y_k$
    are assigned in the first subphase.
    If some cluster that intersects $Y_k$
    gets assigned in Step 1 of the second subphase, 
    then clearly, at least one node of $Y_k$
    gets assigned, as needed.
    
    Hence, consider the case where no cluster intersecting $Y_k$
    is assigned in Step 1 of the second subphase.
    Thus, it has to be that $C' \in \mcC^2_i \cup \mcC^3_i$,
    since if $C' \in \mcC^1_i$, $C'$ would be assigned in Step 1.
    We examine the following cases:
    \begin{itemize}
        \item $C' \in \mcC^2_i$:
        In this case, we show that some node of $Y_k$ will get assigned in Step 3
        of the subphase.
        Let $C \in \mcC^2_i$ be the cluster intersecting $Y_R$
        that the algorithm uses to determine $x'_C$
        (we know that $C$ exists since $\mcC^2_i$ is not empty).     
        If $Y_k = Y_R$, then cluster $C$ gets assigned as well,
        which implies that some node of $Y_k$ is assigned, as needed.
        
        Now consider, $Y_k \neq Y_L$.
        It has to be that $k < R$, by the definition of $Y_R$.
        Node $x_C$ must be in a partial bag on the left of $Y_k$,
        since none of the nodes of $Y_k$ were assigned in the first subphase.
        Therefore,
        path $\pi'_{C}$ must go through partial bag $Y_k$.
        Thus, some subpath $\pi''_z$ of unassigned nodes 
        ($z$ is odd) intersects an unassigned cluster $C''$ 
        that intersects $Y_k$.
        Cluster $C''$ will get assigned, and hence, 
        a node of $Y_k$ in the intersection with $C''$ is also assigned, as needed.
        
        \item $C' \in \mcC^3_i$:
        If a node in $Y_k$ is assigned in Step 2 the result is proven.
        Thus, suppose that no node in $Y_k$ is assigned in Step 2.
        This case is symmetric to the previous case,
        where instead of $Y_R$ we have $Y_L$.
        Thus, a node of $Y_k$ will be assigned in Step 3.
        \qedhere
    \end{itemize}
\end{proof}

The next result follows immediately by Lemma \ref{lemma:group-assign}
and by induction on the number of phases.
Each cluster in $\mcC'$ 
is connected, since whenever a cluster $C \in \mcC$ 
that contains the origin of one of the paths $x_C$, $x'_C$, or $x''_z$ (for some $z \geq 1$)
is assigned to a portal $p$,
together with $C$ all the unassigned clusters along the respective path 
($x_C$, $x'_C$, or $x''_z$) 
are assigned to the same portal $p$.

\begin{lemma}
\label{lemma:pathwidth-correctness}
By the end of phase $\pw +1$, each node $v \in V$ is assigned to a portal
such that the respective cluster $G[f^{-1}(v)] \in \mcC'$ is connected. 
\end{lemma}

Next, we continue to prove a bound on the maximum detour by the end of phase $\phi$.

\begin{lemma}
\label{lemma:subphase-induction}
At the end of phase $\phi$, where $1 \leq \phi \leq \pw + 1$,
each $v \in V$ that has been assigned to a portal has a detour $\dtr_f(v) \leq 8 \phi \Delta$.
\end{lemma}

\begin{proof}
We prove the claim by induction on $\phi$.
For the base case $\phi=1$,
from Lemma \ref{lemma:odd-detour},
since $\delta = 0$,
the detour contributed to newly assigned nodes 
in the first subphase is at most $\delta' \leq 2\Delta$.
From Lemma \ref{lemma:even-detour},
the additional detour contributed to newly assigned nodes 
in the second subphase is at most $\delta' + 6\Delta \leq 8 \Delta$.
Therefore,
for $\phi = 1$,
$\dtr_f(v) \leq 8 \Delta$ for any assigned node $v \in V$.

Suppose that the claim holds for $\phi \geq 1$.
Now consider phase $\phi+1$. 
Similar to the base case analysis,
the additional detour due to the two subphases 
is at most $8\Delta$.
Thus, the total detour of any assigned node $v \in V$
is $\dtr_f(v) \leq 8 \phi \Delta + 8 \Delta = 8(\phi+1) \Delta$, 
as needed.
\end{proof}

Clearly, the pathwidth aggregation algorithm executes 
in a polynomial number of steps 
with respect to the number of nodes in $G$, bags in $T$, and $\pw$.
Theorem \ref{thm:caPW} follows 
from Lemma \ref{lemma:pathwidth-correctness} and 
from Lemma \ref{lemma:subphase-induction} 
by setting $\phi = \pw+1$.

\section{Combining Reduction, Cluster Aggregation and Dangling Nets}\label{sec:putTogether}
In this section, we combine our reduction of strong sparse partitions to dangling nets and cluster aggregation (\Cref{thm:MetaHierarchical}) with known dangling net constructions and our cluster aggregation solutions from \Cref{sec:improvedCA}. The result is our strong sparse partition hierarchies (\Cref{dfn:hierarchies}) which when fed into \Cref{thm:decompToUST} gives our UST constructions.

\subsection{Hierarchies and USTs in General Graphs}

We begin by proving the existence of good hierarchies of strong sparse partitions for general graphs.
\begin{theorem}\label{thm:SSPGen}
    Every edge-weighted graph $G = (V,E, w)$ admits a $\gamma$-hierarchy of $(\alpha, \tau)$-sparse strong partitions for $\alpha = \tau = O(\log n)$ and $\gamma = O(\log ^ 2 n)$.
\end{theorem}
\begin{proof}
	\sloppy 
By \Cref{thm:MPXbasedClusteringGeneral} every general graph has a poly-time computable $\Delta$-covering $(O(\log n), O(\log n))$-sparse dangling net $N$ for every $\Delta > 0$. Furthermore, for any such $N$ we have that $G+N$ has a poly-time computable $O(\log n)$-distortion cluster aggregation by \Cref{thm:caGen}. Applying our reduction theorem (\Cref{thm:MetaHierarchical}) gives the result.
\end{proof}
\noindent We note that our cluster aggregation for trees (\Cref{thm:caTree}) allows us to improve $\gamma$ to $O(\log n)$ in the above theorem (for the case where $G$ is a tree).

Combining \Cref{thm:decompToUST} with \Cref{thm:SSPGen} gives our UST theorem for general graphs.
\begin{theorem}\label{thm:USTGen}
    Every edge-weighted graph $G = (V,E, w)$ admits an $O(\log ^ 7 n)$-approximate universal Steiner tree. Furthermore, this tree can be computed in polynomial time.
\end{theorem}    

\subsection{Hierarchies and USTs in Bounded Pathwidth Graphs}

We begin by giving our strong sparse partition hierarchies for pathwidth $\pw$ graphs.
\begin{theorem}\label{thm:SSPPW}
	Every edge-weighted graph $G = (V,E, w)$ with pathwidth $\pw$ admits a $\gamma$-hierarchy of $(\alpha, \tau)$-sparse strong partitions for $\alpha = O(\pw)$, $\tau = O(\pw^2)$ and $\gamma = O(\pw^2)$.
\end{theorem}
\begin{proof}
	\sloppy 
	By \Cref{thm:MPXbasedClusteringPathwidth} every pathwidth $\pw$ graph has a poly-time computable $\Delta$-covering $(O(\pw), O(\pw^2))$-sparse dangling net $N$ for every $\Delta > 0$. Furthermore, observe that adding a dangling net to such a graph increases its pathwidth by at most $1$. It follows that for any such $N$ we have that $G+N$ has a poly-time computable $O(\pw)$-distortion cluster aggregation by \Cref{thm:caPW}. Applying our reduction theorem then (\Cref{thm:MetaHierarchical}) gives the result.
\end{proof}

Combining \Cref{thm:SSPPW} with \Cref{thm:decompToUST} we conclude our UST for pathwidth $\pw$ graphs.
\begin{theorem}\label{thm:USTPW}
	Every edge-weighted graph $G = (V,E, w)$ with pathwidth $\pw$ admits an $O(\pw^8 \cdot \log n)$-approximate universal Steiner tree. Furthermore, this tree can be computed in polynomial time.
\end{theorem}
\subsection{Hierarchies and USTs in Bounded Doubling Dimension Graphs}

We next give our strong sparse partition hierarchies for doubling dimension $\ddim$ graphs. 
The proof of \Cref{thm:SSPDD} requires more work than \Cref{thm:SSPGen,thm:SSPPW}.
\begin{theorem}\label{thm:SSPDD}
	Every edge-weighted graph $G = (V,E, w)$ with doubling dimension $\ddim$ admits a $\gamma$-hierarchy of $(\alpha, \tau)$-sparse strong partitions for $\alpha = O(\ddim)$, $\tau = \tilde{O}(\ddim)$ and $\gamma = \tilde{O}(\ddim^3)$.
\end{theorem}
\begin{proof}
	Let $c_\Lambda=\Theta(1),\alpha=O(\ddim),\tau=\tilde{O}(\ddim)$ be the parameters from \Cref{thm:MPXbasedClusteringDoubling}. That is given a graph with doubling dimension $\ddim$ and parameter $\Lambda>0$, \Cref{thm:MPXbasedClusteringDoubling} returns a $\Lambda$-covering $(\alpha,\tau)$-sparse dangling net $N$, such that $\min_{u,v\in N_V}d_G(u,v)\ge c_\Lambda\cdot\Lambda$ ($c_{\Lambda}\in
	(0,1)$).
	
	We will construct $\gamma$-hierarchy of strong $(12\alpha,\tau)$-sparse partitions, (for $\gamma=\tilde{O}(\ddim^3)$ to be determined later) with the additional property that every cluster $C\in\calC_i$ has a center $v_C$ such that the minimal distance between the cluster centers is $\min_{C,C'\in\calC_i}d_G(v_C,v_{C'})\ge \frac{c_{\Lambda}}{3}\cdot \gamma^i$.
	Specifically, we will follow the lines of the  proof of \Cref{thm:MetaHierarchical}, and argue that the cluster aggregation procedure can always be applied.
	The construction of the hierarchical partition is bottom up. The lowest level of singletons is trivially constructed. Suppose by induction that we obtained a partition $\calC_i$ with the cluster centers as above.
	
	Set 
	$\ddim'=\ddim\cdot(1+\log\frac{8}{c_\Lambda})$.
	For the inductive step, 
	we first apply \Cref{thm:MPXbasedClusteringDoubling} with $\Lambda=\lambda\cdot\ddim'^3\cdot\log\ddim'\cdot\gamma^i=O(\ddim^3\cdot\log\ddim)\cdot\gamma^i$, (for a large enough constant $\lambda$ to be determined later)
	to obtain a $\Lambda$-covering $(\alpha=O(\ddim),\tau=\tilde{O}(\ddim))$-sparse dangling net $N$.
	For every $p\in N$, let $p_v\in V$ be the vertex $p$ is matched to in $V$. Then by 
	\Cref{thm:MPXbasedClusteringDoubling} it holds that the minimum pairwise distance between $N_V=\{p_v\mid p\in N\}$ is $c_\Lambda\cdot\Lambda$.
	We argue that the doubling dimension of $G+N$ is $\ddim'$.
	\begin{claim}
		$G+N$ has doubling dimension at most $\ddim'=\ddim\cdot(1+\log\frac{8}{c_\Lambda})$.
	\end{claim}
	\begin{proof}
		Consider a ball $B_{G+N}(v,r)$, we will show that it can be covered by $2^{\ddim'}$ balls of radius $\frac r2$. 
		We can assume that $v\in V$, as otherwise, if $v=p\in N$, than either $B_{G+N}(p,r)$ is the singleton $\{p\}$, or $B_{G+N}(p,r)\subseteq B_{G+N}(p_v,r)$.
		We proceed by case analysis.
		\begin{itemize}
			
			\item If $r\ge4\Lambda$, the ball $B_G(v,r)$ can be covered by $2^{\ddim}$ balls of radius $\frac r2$, each of which can be covered by $2^{\ddim}$ balls of radius $\frac r4$. In particular, there are  $2^{2\ddim}$ centers $u_1,u_2,\dots$ such that  $B_G(v,r)\subseteq\bigcup_{i=1}^{2^{2\ddim}}B_G(u_i,\frac r4)$. As every vertex in $N$ is at distance at most $\Lambda$ from some vertex in $V$, it follows that $B_{G+N}(v,r)\subseteq\bigcup_{i=1}^{2^{2\ddim}}B_{G}(u_{i},\frac{r}{4}+\Lambda)\subseteq\bigcup_{i=1}^{2^{2\ddim}}B_{G}(u_{i},\frac{r}{2})$, as required.

			\item Else, $r<4\Lambda$. Then by the packing property (\Cref{lem:doubling_packing}), $\left|N_V\cap B_G(v,r)\right|\le\left(\frac{2r}{c_\Lambda\cdot\Lambda}\right)^{\ddim}< \left(\frac{8}{c_\Lambda}\right)^{\ddim}$. $B_{G+N}(v,r)$ can be covered by $2^{\ddim}$ balls of radius $\frac r2$ in $G$, and in addition $\left(\frac{8}{c_\Lambda}\right)^{\ddim}$ balls centered in $N_V\cap B_G(v,r)$ vertices, $2^{\ddim\cdot(1+\log\frac{8}{c_\Lambda})}$ balls in total.\qedhere
		\end{itemize}
		
	\end{proof}
	
	We next apply \Cref{thm:caDD} of $G+N$ on the clusters $\calC_i$ with portals $N$. Note that we can use \Cref{thm:caDD} as indeed $\calC_i$ have strong diameter $\gamma^i$, $N$ is a $\lambda\cdot\ddim'^3\log\ddim'\cdot\gamma^i$-net with minimum distance between portals $c_\Lambda\cdot\Lambda$, and the clusters $\calC_i$ have cluster centers at minimum pairwise distance $\frac{c_\Lambda}{3}\cdot\gamma^i$.
	Thus \Cref{thm:caDD} is guaranteed to return a solution to the cluster aggregation problem with distortion $4c_{s}\cdot\ln\frac{4}{c_{\Lambda}}\cdot\ddim'^{2}\log\ddim'\cdot\gamma^{i}=\frac{4c_{s}\cdot\ln\frac{4}{c_{\Lambda}}}{\lambda\cdot\ddim'}\cdot\Lambda$.
	Denote
	\begin{equation}
		\beta=\frac{4c_{s}\cdot\ln\frac{4}{c_{\Lambda}}}{\lambda\cdot\ddim'}\le\frac{1}{2\alpha}~.\label{eq:cslambdaGlobal}
	\end{equation}
	Here $c_s$ is the constant defined in \Cref{thm:caDD}. The inequality holds for an appropriate choice of $c_s,\lambda$ (see \Cref{rem:constantsChoiseDoubling}).
	The diameter of the resulting clusters will be at most
	\begin{equation*}
		2\cdot\left(\Lambda+\beta\cdot\Lambda\right)\le3\Lambda=3\lambda\cdot\ddim'^{3}\cdot\log\ddim'\cdot\gamma^{i}=\gamma^{i+1}~.
	\end{equation*}
	In particular, we set $\gamma=3\lambda\cdot\ddim'^{3}\cdot\log\ddim'=O(\ddim^{3}\log\ddim)$.
	Note that the resulting clusters $\calC_{i+1}$ have centers $N_V$ at pairwise distance $c_{\Lambda}\cdot\Lambda=\frac{c_{\Lambda}}{3}\cdot\gamma^{i+1}$
	as required by the induction.

	The actual clusters we will return $\calC_{i+1}$, are obtained by removing the auxilery vertices $N$ from the clusters returned by the cluster aggregation procedure. As these are degree $1$  vertices, the strong diameter can only decrease by this.
	Clearly, $\calC_{i}$ is a refinement of $\calC_{i+1}$.
	
	It remains to  prove the ball preservation property. Fix a vertex $v\in V$. 
	Consider a ball $B_G(v, R)$ around $v$ of radius $R=\frac{\gamma^{i+1}}{12\alpha}=\frac{\Lambda}{4\alpha}$.
	For every $u\in B_{G}(v,R)$, the cluster aggregation solution assigns $u$ to a portal $t_u\in N$.
	By the guarantees of cluster aggregation we have
	\begin{align*}
		d_{G+N}(v,t_{u}) & \le d_{G+N}(v,u)+d_{G+N}(u,t_{u})\\
		& \le d_{G}(v,u)+d_{G+N}(u,N)+\beta\cdot\Lambda\\
		& \le d_{G+N}(v,N)+2d_{G}(v,u)+\beta\cdot\Lambda\\
		& \le d_{G+N}(v,N)+\left(\frac{2}{4\alpha}+\beta\right)\cdot\Lambda\\
		& \le d_{G+N}(v,N)+\left(\frac{1}{2\alpha}+\frac{1}{2\alpha}\right)\cdot\Lambda=d_{G+N}(v,N)+\frac{\Lambda}{\alpha}~.
	\end{align*}

	As $N$ is a $\Lambda$-covering $(\alpha, \tau)$-sparse dangling net, it holds that $$\left|\left\{ t\in N\mid d_{G+N}(v,t)\le d_{G+N}(v,N)+\frac{\Delta}{\alpha}\right\} \right|\le\tau.$$ It follows that the vertices in $B_{G}(v,R)$ are assigned to at most
	$\tau$ different portals, as required.

	%
\end{proof}

Combining \Cref{thm:SSPDD} with \Cref{thm:decompToUST} we conclude our UST for doubling graphs:
\begin{theorem}\label{thm:USTDD}
	Every edge-weighted graph $G = (V,E, w)$ with doubling dimension $\ddim$ admits an $\tilde{O}(\ddim^7) \cdot \log n$-approximate universal Steiner tree. Furthermore, this tree can be computed in polynomial time.
\end{theorem}    

\section{Conclusion and Future Directions}\label{sec:conclusion}

In this work we gave the first poly-logarithmic universal Steiner trees in general graphs and strong sparse partition hierarchies. Our approach reduces poly-logarithmic strong sparse partition hierarchies to the cluster aggregation problem and dangling nets and then leverages a known connection between strong sparse partition hierarchies and universal Steiner trees. We gave $O(\log n)$-distortion solutions for cluster aggregation and improved solutions in trees and bounded pathwidth and doubling dimension graphs.

We conclude with some open questions and potential future directions:
\begin{enumerate}
    \item \textbf{Improved UST and Strong Sparse Partition Bounds:} The most obvious remaining open direction is to close the gap between our $O(\log ^7 n)$-approximate USTs and the known $\Omega(\log n )$ lower bound of \cite{jia2005universal}. 
    Note that even assuming that $G$ is the complete graph with metric weights, the best upper bound is $O(\log ^2 n)$ \cite{gupta2006oblivious}.
    Along these lines, it would be interesting to improve the reduction of USTs to hierarchical strong sparse partitions or to improve the $\gamma$-parameter in the strong hierarchical sparse partitions for general graphs (the $\alpha$ and $\tau$ parameters are tight up to a $\log\log n$ factor \cite{filtser20}). 
    \item \textbf{USTs and Strong Sparse Partition Hierarchies for New Graph Families:} Similarly, improving the bounds for new restricted graph families to get better USTs and strong sparse partitions is exciting. Currently, we know of no super-constant lower bound for UST for either constant treewidth or constant pathwidth graphs.
    \item \textbf{Improved Cluster Aggregation in Restricted Graph Families:} One particularly interesting piece of this puzzle for restricted graph families is the status of cluster aggregation in special graph families. In particular, we conjecture that planar graphs and, more generally, minor-free graphs always admit $O(1)$-distortion cluster aggregation solutions. Likewise, we conjecture that $\tilde{O}(\ddim)$-distortion cluster aggregation should be possible in graphs of doubling dimension $\ddim$ (our current upper bound is $\tilde{O}(\ddim^2)$).
    \item \textbf{Scattering Partitions:} A graph decomposition closely related to sparse partitions are the scattering partitions of \cite{filtser20}. A $(1,\tau,\Delta)$-scattering partition is a partition into connected clusters with (weak) diameter at most $\Delta$, such that every shortest path of length at most $\Delta$ intersects at most $\tau$ different clusters. Note that every strong sparse partition is also scattering, while weak sparse partitions and scattering partitions are incomparable. 
    In a similar spirit to \Cref{thm:decompToUST}, in \cite{filtser20} it was shown that if every induced subgraph of $G$ admits an $(1,\tau,\Delta)$-scattering partition for all $\Delta$, then $G$ admits an $O(\tau^3)$-stretch solution for the ``Steiner point removal problem'' (SPR). See \cite{filtser20} for background and definitions. Recently Chang \etal \cite{CCLMST23} showed that planar graphs admit $(1,O(1))$-scattering partition schemes (implying an $O(1)$-stretch solution for SPR problem on such graphs).
    However, scattering partitions for general graphs are not yet understood. Filtser \cite{filtser20} showed that $n$-vertex graphs admit $(1,O(\log^2 n))$-scattering partition schemes, and conjectured that they admit $(1,O(\log n))$-scattering partition schemes.
    %
\end{enumerate}

\section*{Acknowledgments}
The authors would like to thank Bernhard Haeupler for helpful discussions leading to the proof of \Cref{lem:ddimInependentEvents}.
The authors would also like to thank R.\ Ravi for many useful discussions related to \Cref{thm:caGen}.

Busch supported by National Science Foundation grant CNS-2131538.
Filtser supported by the Israel Science Foundation (grant No.\ 1042/22).
Hathcock supported by the National Science Foundation Graduate Research Fellowship under grant No.~DGE-2140739.
Hershkowitz funded by the SNSF, Swiss National Science Foundation grant 200021\_184622.  Rajaraman supported by National Science Foundation grant CCF-1909363.

\bibliography{abb,main,Fil20}
\bibliographystyle{alpha}

\end{document}